\documentclass[a4paper,UKenglish,cleveref, autoref, thm-restate]{lipics-v2021}

\usepackage{xcolor}

\title{Playing (Almost-)Optimally in Concurrent Büchi and co-Büchi Games}

\titlerunning{Playing (Almost-)Optimally in Concurrent Büchi and co-Büchi Games} 

\author{Benjamin Bordais, Patricia Bouyer and Stéphane Le Roux}{Université Paris-Saclay, CNRS, ENS Paris-Saclay, LMF, 91190 Gif-sur-Yvette, France}{}{}{}

\authorrunning{B. Bordais, P. Bouyer and S. Le Roux}

\Copyright{Benjamin Bordais, Patricia Bouyer and Stéphane Le Roux}



\relatedversion{} 

\nolinenumbers
\usepackage{stmaryrd}
\input{macros}

\begin{document}

\maketitle

\begin{abstract}
  We study two-player concurrent stochastic games on finite graphs,
  with B\"uchi and co-B\"uchi objectives. The goal of the first player
  is to maximize the probability of satisfying
  the given objective. Following Martin's determinacy theorem for
  Blackwell games, we know that such games have a value.  Natural
  questions are then: does there exist an optimal strategy, that is, a
  strategy achieving the value of the game? what is the memory
  required for playing (almost-)optimally?

  The situation is rather simple to describe for turn-based games,
  where positional pure strategies suffice to play optimally in games
  with parity objectives. Concurrency makes the situation intricate
  and heterogeneous. For most $\omega$-regular objectives, there do
  indeed not exist optimal strategies in general. For some objectives
  (that we will mention), infinite memory might also be required for
  playing optimally or almost-optimally.

  We also provide characterizations of local interactions of the
  players to ensure positionality of (almost-)optimal strategies for
  B\"uchi and co-B\"uchi objectives. This characterization relies on
  properties of game forms underpinning the formalism for defining
  local interactions of the two players. These well-behaved game forms
  are like elementary bricks which, when they behave well in
  isolation, can be assembled in graph games and ensure the good
  property for the whole game.
\end{abstract}



\section{Introduction}

\subparagraph*{Stochastic concurrent games.}

Games on graphs are an intensively studied mathematical tool, with
wide applicability in verification and in particular for the
controller synthesis problem, see for
instance~\cite{thomas02,BCJ18}. We consider two-player stochastic
concurrent games played on finite graphs. For simplicity (but this is
with no restriction), such a game is played over a finite bipartite
graph called an arena: some states belong to Nature while others
belong to the players. Nature is stochastic, and therefore assigns a
probabilistic distribution over the players' states. In each players'
state, a local interaction between the two players (called Player $\A$
and Player $\B$) happens, specified by a two-dimensional table. Such
an interaction is resolved as follows: Player $\A$ selects a
probability distribution over the rows 
while Player $\B$ selects a probability distribution over the columns
of the table; this results into a distribution over the cells of the
table, each one pointing to a Nature state of the graph. An example of
game arena (with no Nature states -- we could add dummy deterministic
Nature states) is given in
Figure~\ref{fig:local_optimal_not_uniformly} (this example comes from
\cite{AH00}). At state $q_0$, the
interaction between the two players is given by the table, and each
player has two actions: if Player $\A$ plays the second row and Player
$\B$ the first column, then the game proceeds to state $\top$: in that case, the game always goes back to $q_0$.

Globally, the game proceeds as follows: starting at an initial state
$q_0$, the two players play in the local interaction of the current
state, and the joint choice determines (stochastically) the next
Nature state of the game, itself moving randomly to players' states;
the game then proceeds subsequently from the new players' state. The
way players make choices is given by strategies, which, given the
sequence of states visited so far (the so-called history), assign
local strategies for the local interaction of the state the game is
in.  For application in controller synthesis, strategies will
correspond to controllers, hence it is desirable to have strategies
simple to implement. We will be in particular interested in strategies
which are \emph{positional}, i.e. strategies which only depend on
the current state of the game, not on the whole history.  When each
player has fixed a strategy (say $\s_{\A}$ for Player $\A$ and
$\s_{\B}$ for Player $\B$), this defines a probability distribution
$\prob{q_0}{\s_\A,\s_\B}$ over infinite sequences of states of the
game.  The objectives of the two players are opposite (we assume a
zero-sum setting): together with the game, a measurable set $W$ of
infinite sequences of states is fixed; the objective of Player $A$ is
to maximize the probability of $W$ while the objective of Player
$B$ is to minimize it. 

Back to the example of
Figure~\ref{fig:local_optimal_not_uniformly}. If Player $\A$
(resp. $\B$) plays the first row (resp. column) with probability
$p_\A$ (resp. $p_\B$), then the probability to 
move to $\bot$ in one step is $(1-p_\A) \cdot (1-p_\B)$.  If Player
$\A$ repeatedly plays the same strategy at $q_0$ with $p_\A<1$, by
playing $p_\B=0$, Player $\B$ will enforce $\bot$ almost-surely;
however, if Player $\A$ plays $p_\A=1$, then by playing $p_\B=1$,
Player $\B$ enforces staying in $q_0$, hence visiting $\top$ with
probability $0$. 
On the contrary 
Player $\A$ can ensure visiting $\top$ infinitely often with probability $1-\varepsilon$ for every $\varepsilon>0$, by playing iteratively at $q_0$ the first row of the table with probability $1-\varepsilon_k$ ($k$ the number of
visits to $q_0$) and the second row with probability $\varepsilon_k$,
where the sequence $(\varepsilon_k)_k$ decreases fast to zero (see Appendix~\ref{appen:buchi_varepsilon_infinite}).

\subparagraph*{Values and (almost-)optimal strategies.}

As mentioned above, Player $\A$ wants to maximize the probability of
$W$, while Player $\B$ wants to minimize this probability. Formally,
given a strategy $\s_{\A}$ for Player $\A$, its value is measured by
$\inf_{\s_{\B}} \prob{q_0}{\s_\A,\s_\B}(W)$, and Player $\A$ wants to
maximize that value. Dually, given a strategy $\s_{\B}$ for Player
$\B$, its value is measured by
$\sup_{\s_{\A}} \prob{q_0}{\s_\A,\s_\B}(W)$, and Player $\B$ wants to
  minimize that value. Following Martin's determinacy theorem for
  Blackwell games~\cite{martin98}, it actually holds that 
  the game has a \emph{value} given by
\[
  \val{}{q_0} = \sup_{\s_{\A}} \ \inf_{\s_{\B}}
  \prob{q_0}{\s_\A,\s_\B}(W) = \inf_{\s_{\B}} \ \sup_{\s_{\A}}
  \prob{q_0}{\s_\A,\s_\B}(W)
\]
While this ensures the existence of almost-optimal strategies (that
is, $\varepsilon$-optimal strategies for every $\varepsilon>0$) for
both players, it says nothing about the existence of optimal
strategies, which are strategies achieving $\val{}{q_0}$.  In general,
except for safety objectives, optimal strategies may not exist, as
witnessed by the example of
Figure~\ref{fig:local_optimal_not_uniformly}, which uses a B\"uchi
condition.  Also, it says nothing about the complexity of optimal
strategies (when they exist) and $\varepsilon$-optimal
strategies. Complexity of a strategy is measured in terms of memory
that is used by the strategy: while general strategies may depend on
the whole history of the game, a \emph{positional} strategy only
depends on the current state of the game; a \emph{finite-memory}
strategy records a finite amount of information using a finite
automaton; the most complex ones, the \emph{infinite-memory}
strategies require more than a finite automaton to record information
necessary to take decisions.

  Back to the game of Figure~\ref{fig:local_optimal_not_uniformly},
  assuming the B\"uchi condition ``visit $\top$ infinitely often'',
  the game is such that $\val{}{q_0}=1$. However Player $\A$ has no
  optimal strategy, and can only achieve $1-\varepsilon$ for every
  $\varepsilon>0$ with an infinite-memory strategy, and any positional strategy has value 0.

\begin{figure}
	\begin{minipage}[b]{0.35\linewidth}
		\hspace*{-1cm}
		\centering
		\includegraphics[scale=.9]{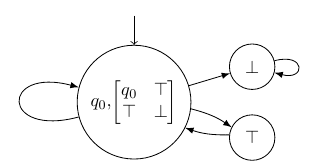}
		\caption{A concurrent game
                  with objective $\Bu(\{ \top \})$. }
		\label{fig:local_optimal_not_uniformly}
	\end{minipage}
        \hfill
	\begin{minipage}[b]{0.28\linewidth}
		\centering
		\includegraphics[scale=1.1]{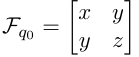}
		\caption{The local interaction at state $q_0$.}
		\label{fig:local_iter}              
	\end{minipage}
        \hfill
	\begin{minipage}[b]{0.28\linewidth}
		\centering
		\includegraphics[scale=1.1]{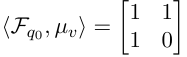}
		\caption{The local interaction at state $q_0$ valued by the vector $\mu_v$ giving the value of states.}
		\label{fig:valued_local_iter_not_uniform}      
	\end{minipage}
\end{figure}

\subparagraph*{The contributions of this work.}  We are interested in
memory requirements for optimal and $\varepsilon$-optimal strategies
in concurrent games with parity 
objectives. The situation is
rather
heterogeneous: 
while safety objectives enjoy very robust properties (existence of
positional optimal strategies in all cases, see for
instance~\cite[Thm. 1]{AM04b}), the situation appears as much more
complex for parity objectives (already with three colors): there may
not exist optimal strategies
, and when they exist, optimal strategies 
as well as $\varepsilon$-optimal strategies require in general infinite memory
(the case of 
optimal strategies was proven in \cite{AH00} while
the case of $\varepsilon$-optimal strategies is a consequence of the
B\"uchi case, studied in~\cite[Thm. 2]{AM04b}).  The case of
reachability objectives was studied with details in~\cite{BBSCSL22}: 
optimal strategies may not exist, but there exists a positional
strategy that is 1) optimal from each state from where there exists an
optimal strategy, and 2) $\varepsilon$-optimal from the other states.

In this paper, we focus on B\"uchi and co-B\"uchi objectives. Few
things were known for those games: specifically, it was shown
in~\cite[Thm. 2]{AM04b} that B\"uchi and co-B\"uchi concurrent games
may have no optimal
strategies, and that $\varepsilon$-optimal strategies may require
infinite memory for B\"uchi objectives. We show in addition that, when
optimal strategies exist everywhere, optimal strategies can be chosen
positional for B\"uchi objectives but may require infinite memory for
co-B\"uchi objectives. We
more importantly characterize ``well-behaved'' local interactions
(i.e. interactions of the two players at each state, which are given
by tables) for ensuring positionality of ($\varepsilon$-)optimal
strategies in the various settings where they do not exist in the
general case. We follow the approach used in~\cite{BBSCSL22} for
reachability objectives and abstract those local interactions into
game forms, where cells of the table are now seen as variables (some
of them being equal). For instance, the game form associated with
state $q_0$ in the running example has three outcomes: $x$, $y$ and
$z$, and it is given in Figure~\ref{fig:local_iter}.  Game forms can
be seen as elementary bricks that can be used to build games on
graphs. Given a property we want to hold on concurrent
  games (e.g. the existence of positional optimal strategies in Büchi
  games), we characterize those bricks that are safe
  w.r.t. that property, that is the game forms that behave well when
used individually, the ones ensuring that, when they are
  the only non-trivial local interaction in a concurrent game, the
  property holds. Then, we realize that they also behave well when
used collectively: if all local interactions are safe in a
  concurrent game, then the whole game ensures the property of
  interest. We obtain a clear-cut separation: if all local
  interactions are safe in a concurrent game, then the property is
  necessarily ensured; on the other hand, if a game form is not safe,
  one can build a game where it is the only non-trivial local
  interaction that does ensure the property.
Our contributions can be summarized as follows:
\begin{enumerate}
\item In the general setting of prefix-independent objectives, we
  characterize positional uniformly optimal strategies using locally
  optimal strategies (i.e. strategies which are optimal at local
  interactions) and constraints on values of end-components generated
  by the strategies (Lemma~\ref{lem:uniformly_optimal}).
\item We study B\"uchi concurrent games. We first show that there is a
  positional uniformly optimal strategy in a B\"uchi game as soon as
  it is known that optimal strategies exist from every state (Proposition~\ref{prop:positional_suffice_buchi}). To
  benefit from this result, we give a (sufficient and necessary)
  condition on (game forms describing) local interactions to ensure
  the existence of optimal strategies. In particular, if all local
  interactions are well-behaved (w.r.t. the condition), then it will
  be the case that positional uniformly optimal strategies exist (Theorem~\ref{lem:rm_in_buchi}). We
  also give a weaker necessary and sufficient condition on local
  interactions to ensure the existence of positional
  $\varepsilon$-optimal strategies in B\"uchi games, since in general,
  infinite memory might be required (Theorem~\ref{lem:varepsilon_buchi_sufficient}).
\item We study co-B\"uchi concurrent games. We first show that optimal
  strategies might require infinite memory in co-B\"uchi games,in
  contrast with B\"uchi games (Subsection~\ref{subsec:optimal_strat_in_co_buchi}). We also characterize local
  interactions that ensure the existence of positional optimal
  strategies in co-B\"uchi games (Theorem~\ref{lem:co_buchi_sufficient}). It is not useful to do the same for
  positional $\varepsilon$-optimal strategies since it is always the
  case that such strategies exist \cite{CAH06}.
\end{enumerate}
%


All these results show the contrast between concurrent games and
  turn-based games: indeed, in the latter (which have attracted more
  attention these last years), pure (that is, deterministic)
  positional optimal strategies always exist for parity objectives
  (hence in particular, for B\"uchi and co-B\"uchi
  objectives)~\cite{MM02,CJH04,zielonka04}.
  The results presented in this paper hence show the complexity
  inherent to concurrent interactions in games. Those have
  nevertheless retained some attention over the last twenty
  years~\cite{AH00,CAH06,AHK07,CAH11,KNPS21}, and are relevant in
  applications~\cite{KNPS+22}.



\section{Game Forms}

A \emph{discrete probabilisty distribution} over a non-empty finite
set $Q$ is a function $\mu: Q \rightarrow [0,1]$ such that
$\sum_{x \in Q} \mu(x) = 1$. The \emph{support} $\Supp(\mu)$ of a
probabilistic distribution $\mu: Q \rightarrow [0,1]$ corresponds to
the set of non-zeros of the distribution:
$\Supp(\mu) = \{ q \in Q \mid \mu(q) \in {]0,1]} \}$. 
The set of all distributions over the set $Q$ is denoted $\Dist(Q)$.

Informally, game forms are two-dimensional tables with variables while
games in normal forms are game forms whose outcomes are real values in $[0,1]$. Formally:
\begin{definition}[Game form and game in normal form]
  \label{def:arena_game_nf}
  A \emph{game form} (GF for short) is a tuple
  $\formNF = \langle \St_\A,\St_\B,\outComeNF,\outCNF \rangle$ where
  $\St_\A$ (resp. $\St_\B$) is a non-empty finite set of actions available to Player $\A$ (resp. $\B$), $\outComeNF$ is a
  non-empty set of outcomes, and
  $\outCNF: \St_\A \times \St_\B \rightarrow \outComeNF$ is a function
  that associates an outcome to each pair of actions. When the set
  of outcomes $\outComeNF$ is equal to $[0,1]$, we say that $\formNF$
  is a \emph{game in normal form}. For a valuation
  $v \in [0,1]^\outComeNF$ of the outcomes, the notation
  $\gameNF{\formNF}{v}$ refers to the game in normal form
  $\langle \St_\A,\St_\B,[0,1],v \circ \outCNF
  \rangle$. 
\end{definition}
  An example of game form (resp. game in normal form) is given in
 Figure~\ref{fig:local_iter} (resp. \ref{fig:valued_local_iter_not_uniform}), where $\St_\A$ (resp. $\St_\B$)
  are rows (resp. columns) of the table and $x,y,z$ (resp. $0,1$) are the possible outcomes of the game form (resp. game in normal form). We use 
  game forms to represent interactions between two
  players. The strategies available to Player $\A$ (resp. $\B$) are
  convex combinations of actions given as the rows
  (resp. columns) of the table. In a game in normal form, Player $\A$
  tries to maximize the outcome, whereas Player $\B$ tries to minimize
  it.
\begin{definition}[Outcome of a game in normal form]
  \label{def:outcome_game_form}
  Let $\formNF = \langle \St_\A,\St_\B,[0,1],\outCNF \rangle$ be a
  game in normal form. The set $\Dist(\St_\A)$ (resp. $\Dist(\St_\B)$)
  is the set of (mixed) strategies available to Player $\A$
  (resp. $\B$). For a pair of strategies
  $(\sigma_\A,\sigma_\B) \in \Dist(\St_\A) \times \Dist(\St_\B)$, the
  outcome $\outM_\formNF(\sigma_\A,\sigma_\B)$ in $\formNF$ of the
  strategies $(\sigma_\A,\sigma_\B)$ is 
  $\outM_\formNF(\sigma_\A,\sigma_\B) := \sum_{a \in \St_\A} \sum_{b
    \in \St_\B} \sigma_\A(a) \cdot \sigma_\B(b) \cdot \outCNF(a,b) \in
  [0,1]$.
\end{definition}	

\begin{definition}[Value of a game in normal form and optimal strategies]
  \label{def:alternative_value_game_normal_form}
  Let $\formNF = \langle \St_\A,\St_\B,[0,1],\outCNF \rangle$ be a
  game in normal form, and $\sigma_\A \in \Dist(\St_\A)$ be a strategy
  for Player $\A$.  The \emph{value} of strategy $\sigma_\A$ is
  $\va_\formNF(\sigma_\A) := \inf_{\sigma_\B \in \Dist(\St_\B)}
  \outM_{\formNF}(\sigma_\A,\sigma_\B)$, and analogously for Player
  $\B$, with a $\sup$ instead of an $\inf$. When
  $\sup_{\sigma_\A \in \Dist(\St_\A)} \va_\formNF(\sigma_\A) =
  \inf_{\sigma_\B \in \Dist(\St_\B)} \va_\formNF(\sigma_\B)$, it
  defines the \emph{value} of the game $\formNF$, denoted
  $\va_\formNF$.
	
  A strategy $\sigma_\A \in \Dist(\St_\A)$ ensuring
  $\va_\formNF = \va_\formNF(\sigma_\A)$ is said to be \emph{optimal}. The
  set of all optimal strategies for Player $\A$ is denoted
  $\Opt_\A(\formNF) \subseteq \Dist(\St_\A)$, and analogously for
  Player $\B$. Von Neumann's minimax theorem~\cite{vonNeuman} ensures
  the existence of optimal strategies (for both players).
\end{definition}
In the following, strategies in games in normal forms will be called
$\GF$-strategies, in order not to confuse them with strategies in
concurrent (graph) games.

\section{Concurrent Stochastic Games}
We introduce the definition of a concurrent arena played on a finite
graph, and of a concurrent game by adding a winning condition.
\begin{definition}[Finite stochastic concurrent arena and game]
  A \emph{concurrent arena} $\Aconc$ is a tuple $\AConc$ where $Q$ is
  a non-empty set of states, $\setA$ (resp. $\setB$) is the non-empty
  finite set of actions available to Player $\A$ (resp. $\B$), $\distribSet$ is the set of Nature states,
  $\delta: 
  Q \times A \times B \rightarrow \distribSet$ is the transition
  function and $\distribFunc: \distribSet \rightarrow \Dist(Q)$ is the
  distribution function.

  A \emph{concurrent game} is a pair $\Games{\Aconc}{W}$ where
  $\Aconc$ is a concurrent arena and $W \subseteq Q^\omega$ is
  Borel. The set $W$ is called the \emph{(winning) objective} and it
  corresponds to the set of paths winning for Player $\A$ and losing
  for Player $\B$.
\end{definition}



For the rest of the section, we fix a concurrent game
$\G = \Games{\Aconc}{W}$. An important class of objectives are the
\emph{prefix-independent} objectives, that is objectives $W$ such that
an infinite path
is in $W$ if and only if one of its suffixes is in $W$:
for all $\rho \in Q^\omega$ and $\pi \in Q^+$,
$\rho \in W \Leftrightarrow \pi \cdot \rho \in W$. In this paper, we
more specifically focus on Büchi and co-B\"uchi objectives,
which informally correspond to the infinite paths seeing infinitely
often, or finitely often, a given set of states. We also recall the
definitions of the reachability and safety objectives.
\begin{definition}[Büchi and co-Büchi, Reachability and Safety
  objectives]
  Consider a target set of states $T \subseteq Q$. We define the
  following objectives:
  \begin{itemize}
  \item
    $\Bu(T) := \{ \rho \in Q^\omega \mid \forall i \in \N, \exists j
    \geq i,\; \rho_j \in T \}$;
    $\Reach(T) := \{ \rho \in Q^\omega \mid \exists i \in \N, \rho_i
    \in T \}$;
  \item
    $\coBu(T) := \{ \rho \in Q^\omega \mid \exists i \in \N,\; \forall
    j \geq i,\; \rho_j \notin T \}$;
    $\Sft(T) := \{ \rho \in Q^\omega \mid \forall i \in \N, \rho_i
    \notin T \}$.
  \end{itemize}
\end{definition}

In concurrent games, game forms appear at each state and specify how
the interaction of the Players determines the next (Nature) state. In
fact, this corresponds to the local interactions of the game.
\begin{definition}[Local interactions]
  The \emph{local interaction} at state $q \in Q$ is the game form
  $\formNF_q = \langle \setA,\setB,\distribSet,\delta(q,\cdot,\cdot)
  \rangle$. That is, the $\GF$-strategies available for Player $\A$
  (resp. $\B$) are the actions in $\setA$ (resp. $\setB$) and the
  outcomes are the Nature states. 
\end{definition}

As an example, the local interaction at state $q_0$ in
Figure~\ref{fig:local_optimal_not_uniformly} is represented in
Figure~\ref{fig:local_iter}, up to a renaming 
of the outcomes.

A strategy of a given Player then associates to every history
(i.e. every finite sequence of states) a
$\GF$-strategy in the local interaction at the current state, in other
words it associates a distribution on the actions available at the
local interaction to the given Player.

\begin{definition}[Strategies]
  A \emph{strategy} for Player $\A$ is a function
  $\s_\A: Q^+ \rightarrow \Dist(A)$ such that, for all
  $\rho = q_0 \cdots q_n \in Q^+$, $\s_\A(\rho) \in \Dist(A)$ is a
  $\GF$-strategy for Player $\A$ in the game form $\formNF_{q_n}$.  A
  strategy $\s_\A: Q^+ \rightarrow \Dist(A)$ for Player $\A$ is
  \emph{positional} if, for all $\pi = \rho \cdot q \in Q^+$ and
  $\pi' = \rho' \cdot q' \in Q^+$, if $q = q'$, then
  $\s_\A(\pi) = \s_\A(\pi')$ (that is, the strategy only depends on
  the current state of the game). We denote by $\SetStrat{\Aconc}{\A}$
  and $\SetPosStrat{\Aconc}{\A}$ the set of all strategies and
  positional strategies in arena $\Aconc$ for Player $\A$. The
  definitions are analogous for Player $\B$.
\end{definition}

Before defining the outcome of the game given a strategy for a Player,
we define the probability to go from state $q$ to state $q'$, given
two $\GF$-strategies in the game form $\formNF_q$.
\begin{definition}[Probability Transition]
  \label{def:mu_state}
  Let $q \in Q$ be a state and
  $(\sigma_\A,\sigma_\B) \in \Dist(A) \times \Dist(B)$ be two
  $\GF$-strategies in the game form $\formNF_q$. For a state
  $q' \in Q$, the probability to go from $q$ to $q'$ if the players
  play $\sigma_\A$ and $\sigma_\B$ in $q$ is equal to
  $\prob{}{\sigma_\A,\sigma_\B}(q,q') := \outM_{\langle
    \formNF_q,\distribFunc(\cdot)(q') \rangle}(\sigma_\A,\sigma_\B)$.
\end{definition}
From this, given two strategies, we deduce the probability of any
cylinder supported by a finite path, and consequently of any Borel set
in $Q^\omega$. From a state $q_0 \in Q$, given two strategies $\s_\A$
and $\s_\B$, this probability distribution is denoted
$\prob{\Aconc,q_0}{\s_\A,\s_\B}: \Borel(Q) \rightarrow [0,1]$. More
details are given in Appendix~\ref{appen:def_prob_distrib}. Let us now define the value of a strategy and of the game.

\begin{definition}[Value of strategies and of the game]
  Let 
  $\s_\A \in \SetStrat{\Aconc}{\A}$ be a Player $\A$ strategy. The
  \emph{value} of $\s_\A$ is the function
  $\MarVal{\G}[\s_\A]: Q \rightarrow [0,1]$ such that for every
  $q \in Q$,
  $\MarVal{\G}[\s_\A](q) := \inf_{\s_\B \in \SetStrat{\Aconc}{\B}}
  \prob{\Aconc,q}{\s_\A,\s_\B}[W]$.

  The \emph{value} for Player $\A$ is the function
  $\MarVal{\G}[\A]: Q \rightarrow [0,1]$ such that for all $q \in Q$,
  we have
  $\MarVal{\G}[\A](q) := \sup_{\s_\A \in \SetStrat{\Aconc}{\A}}
  \MarVal{\G}[\s_\A](q)$.
  The value for Player $\B$ is defined similarly by reversing the
  supremum and infimum.

  
  By Martin's result on the determinacy of Blackwell games
  \cite{martin98}, for all concurrent games $\G = \Games{\Aconc}{W}$,
  the values for both Players are equal, which defines the
  \emph{value} of the game:
  $\MarVal{\G} := \MarVal{\G}[\A] = \MarVal{\G}[\B]$. 
  Furthermore, a strategy $\s_\A \in \SetStrat{\Aconc}{\A}$ which
  ensures $\MarVal{\G}[\s_\A](q) = \MarVal{\G}(q)$ from some state
  $q \in Q$,
  is said to be \emph{optimal} from $q$. If, in addition it ensures
  $\MarVal{\G}[\s_\A] = \MarVal{\G}$, it is said to be \emph{uniformly
    optimal}.
  \label{def:determinacy}
\end{definition}

We mention a very useful result of~\cite{CH07}. Informally, it shows
that if there is a state with a positive value, then there is a state
with value $1$.
\begin{theorem}[Theorem 1 in \cite{CH07}]
  \label{thm:chaterjee_value_0_1}
  Let $\G$ be a concurrent game with a prefix-independent
  objective. If there is a state $q \in Q$ such that
  $\MarVal{\G}(q) > 0$ (resp. $\MarVal{\G}(q) < 1$), then there is a
  state $q' \in Q$ such that $\MarVal{\G}(q') = 1$
  (resp. $\MarVal{\G}(q') = 0$).
\end{theorem}

Finally, we define the Markov decision process which is induced by a
positional strategy, and its end-components.
\begin{definition}[Induced Markov decision process]
  Let $\s_\A \in \SetPosStrat{\Aconc}{\A}$ be a positional strategy.
  The \emph{Markov decision process} $\Gamma$ (MDP for short) induced
  by the strategy $\s_\A$ is the triplet
  $\Gamma := \langle Q,B,\iota \rangle$ where $Q$ is the set of
  states, $B$ is the set of actions and
  $\iota: Q \times B \rightarrow \Dist(Q)$ is a map associating to a
  state and an action a distribution over the states. For all
  $q \in Q$, $b \in B$ and $q' \in Q$, we set
  $\iota(q,b)(q') := \prob{\s_\A(q),b}{q,q'}$.
\end{definition}
The induced MDP $\Gamma$ is a special case of a concurrent game where
Player $\A$ does not play 
($A$ could be a singleton) and the set of Player $\B$ strategies is
the same as in 
$\Aconc$. 
The useful objects in MDPs are the end components~\cite{dealfaro97},
i.e. sub-MDPs that are strongly connected.
\begin{definition}[End component and sub-game]
  \label{def:end_component}
  Let $\s_\A \in \SetPosStrat{\Aconc}{\A}$  be a positional
  strategy, and $\Gamma$ its induced MDP.
  An \emph{end component} (EC for short) $H$ in $\Gamma$ is a pair
  $(Q_H,\beta)$ such that $Q_H \subseteq Q$ is a subset of states and
  $\beta: Q_H \rightarrow \mathcal{P}(B) \setminus \emptyset$
  associates to each state a non-empty set of actions compatible with
  the EC $H$ such that:
  \begin{itemize}
  \item for all $q \in Q_H$ and $b \in \beta(q)$,
    $\Supp(\iota(q,b)) \subseteq Q_H$;
  \item the underlying graph $(Q_H,E)$ is strongly connected, where
    $(q,q') \in E$ if and only if there is $b \in \beta(q)$ such that
    $q' \in \Supp(\iota(q,b))$.
  \end{itemize}
  We denote by $\distribSet_H \subseteq \distribSet$ the set of Nature
  states compatible with the EC $H$:
  $\distribSet_H = \{ d \in \distribSet \mid \Supp(d) \subseteq Q_H
  \}$. Note that, for all $q \in Q_H$ and $b \in \beta(q)$, we have
  $\delta(q,\Supp(\s_\A(q)),b) \subseteq \distribSet_H$.
	
  The end component $H$ can be seen as a concurrent arena. In that
  case, it is denoted $\Aconc_H^{\s_\A}$.
\end{definition} 

\section{Uniform Optimality with Positional Strategies}
\label{sec:CNS_uniform_optimal}
In this section, we present a necessary and sufficient condition for a
positional strategy to be uniformly optimal. Note that this result
holds for any prefix-independent objective for which, in any finite
MDP with the complement objective $Q^\omega \setminus W$, there is a
positional optimal strategy. This is in particular the case for the
Büchi and co-Büchi objectives (in fact, it holds for all parity objectives).

We assume a concurrent game $\G = \Games{\Aconc}{W}$ is given for this
section, that $W$ is Borel and prefix-independent, and that
$v := \chi_\G \in [0,1]^Q$ is the value (function) of the game.  We
first define the crucial notion of \emph{local optimality}:
informally, a Player $\A$ positional strategy is locally optimal if,
at each local interaction, it is optimal in the game in normal form
induced 
by the values of the game. 
As the outcomes of a local interaction are the Nature states, we have
to be able to lift the valuation $v$ of the states into a valuation of
the Nature states. This is done via a convex combination in the
definition below.

\begin{definition}[Lifting a valuation of the states]
  Let $w: Q \rightarrow [0,1]$ be an arbitrary valuation of the
  states.  We define 
  $\mu_w: \distribSet \rightarrow [0,1]$, the \emph{lift} of the
  valuation $w$ to Nature states in the following way:
  $\mu_w(d) := \sum_{q \in Q} \distribFunc(d)(q) \cdot w(q)$ for all
  $d \in \distribSet$.
\end{definition}

We can now define the local optimality of a positional strategy.
\begin{definition}[Local optimality]
  A Player $\A$ positional strategy
  $\s_\A \in \SetPosStrat{\Aconc}{\A}$ is \emph{locally optimal}
  if for all $q \in Q$:
  $\va_{\gameNF{\formNF_q}{\mu_v}}(\s_\A(q)) = v(q)$ (i.e. the
  $\GF$-strategy $\s_\A(q)$ is optimal in 
  $\gameNF{\formNF_q}{\mu_v}$).
\end{definition}

Interestingly, in the MDP induced by a locally optimal strategy, all
the states in a given EC have the same value (w.r.t. the valuation
$v$).
This is stated in the proposition below.
\begin{proposition}[Proposition 18 in \cite{BBSCSL22}]
  For every locally optimal Player $\A$ positional strategy
  $\s_\A \in \SetPosStrat{\Aconc}{\A}$, for all EC $H = (Q_H,\beta)$
  in the MDP induced by the strategy $\s_\A$, there exists a value
  $v_H \in [0,1]$ such that, for all $q \in Q_H$, we have
  $v(q) = v_H$.
  \label{prop:ec_same_vale}
\end{proposition}

We now discuss how
local optimality relates to uniform
optimality. We first observe that
local optimality is necessary for uniform optimality, and we
illustrate this on the example of
Figure~\ref{fig:local_optimal_necessary}: in this (partly depicted)
game, Nature states are omitted, and
values 
w.r.t. the valuation $v$ are written in red close to the states. For instance, if Player $\A$ plays the top row and Player $\B$ the left column, the state $q_0$ is reached. The local interaction $\formNF_{q_0}$ at state $q_0$, valued with the lift
$\mu_v$ 
is then depicted in Figure~\ref{fig:valued_local_iter}.  Assume
$\s_\A$ is a Player $\A$ positional strategy that is not locally
optimal at $q_0$: 
the convex
combination 
of the values in the second column (i.e. the values of the states
$q_1,q_2$) is less than $1/2$, i.e.
$p \cdot 1/4 + (1-p) \cdot 3/4 = 1/2 - \varepsilon < 1/2$, where
$p \in [0,1]$ is the probability chosen by $\s_\A(q_0)$ to play the
top row. A Player $\B$ strategy $\s_\B$ whose values at states
$q_1,q_2$ are $\big(\varepsilon/2\big)$-close
to 
$v(q_1),v(q_2)$ 
and that plays the second row with probability $1$ at $q_0$, ensures
that the value of the game w.r.t. $\s_\A,\s_\B$ is at most
$p \cdot (1/4 + \varepsilon/2) + (1-p) \cdot (3/4 + \varepsilon/2)
\leq 
1/2 - \varepsilon/2 < 1/2 = v(q_0)$. Thus, the strategy $\s_\A$ is not optimal
from $q_0$.

\begin{figure}
	\begin{minipage}[b]{0.5\linewidth}
		\hspace*{-1cm}
		\centering
		\includegraphics[scale=0.8]{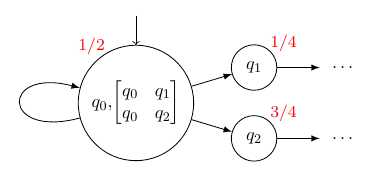}
		\caption{A concurrent game where the values of the states are depicted near them in red.}
		\label{fig:local_optimal_necessary}
	\end{minipage}
	\hspace{5pt}
	\begin{minipage}[b]{0.5\linewidth}
		\centering
		\includegraphics[scale=1.4]{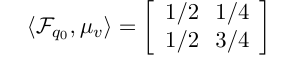}
		\caption{The local interaction at state $q_0$ valued
                  by the function $v$ giving the value of
                  states.}
		\label{fig:valued_local_iter}              
	\end{minipage}
\end{figure}

However, local optimality is not 
a sufficient condition for uniform optimality, as it can be seen in
Figure~\ref{fig:local_optimal_not_uniformly}. A game with a Büchi
objective $\Bu(\{ \top \})$ is depicted (Player $\A$ wins if the state
$\top$ is seen infinitely often). In this game, the valued local
interaction at state $q_0$ is also depicted in
Figure~\ref{fig:valued_local_iter_not_uniform}. A Player $\A$ locally
optimal strategy plays the top row with probability 1 at state
$q_0$. If such a strategy is played, Player $\B$ can ensure the game
never to leave the state $q_0$ (by playing the left column), thus
ensuring value 0 from $q_0$ (with $v(q_0) =
1$). 
In fact, in this game, there is no Player $\A$ optimal strategy.

Overall, from a state $q \in Q$, the local optimality of a Player $\A$
positional strategy $\s_\A$ ensures that, for every Player $\B$
strategy, the convex combination of the values $v_H$ of the ECs $H$,
weighted by the probability to reach them from $q$, is at least the
value of the state $q$. However, this does not guarantee anything
concerning the $\MarVal{\G}[\s_\A]$-value of the game if it never
leaves a specific EC $H$, it may be 0 whereas $v_H > 0$ (that is what
happens in the game of
Figure~\ref{fig:local_optimal_not_uniformly}). 
For the strategy $\s_\A$ to be uniformly optimal, it must ensure that
the value under $\s_\A$ of all states $q$ in the EC $H$
is at least $v_H$ (i.e. for all Player $\B$ strategies $\s_\B$
  which ensure staying in $H$, the value under $\s_\A$ and $\s_\B$ is
  at least
  $v_H$). 
Furthermore, by Theorem~\ref{thm:chaterjee_value_0_1} applied to $H$,
as soon as at least one state has a value smaller than $1$, there is
one state with value $0$. It follows that, for the strategy
  $\s_\A$ to be uniformly optimal, it must be the case that in every
  EC $H$ such that $v_H > 0$, in the game restricted to $H$, 
  the value of the game is $1$. Note that this exactly corresponds to
the condition the authors of \cite{BBSCSL22} have stated in the case
of a reachability
objective. 
Uniform optimality can finally be characterized as follows.
\begin{lemma}[Proof~in Appendix~\ref{sec:proof_lem_uniformly_optimal}]
  Assume that $W$ is Borel and prefix-independent, and that in all
  finite MDPs with objective $Q^\omega \setminus W$, there is a
  positional optimal strategy.
  Let $\s_\A \in \SetPosStrat{\Aconc}{\A}$ be a positional Player $\A$
  strategy. It is uniformly optimal if and only if:
  \begin{itemize}
  \item it is locally optimal;
  \item for all ECs $H$ in the MDP induced by the strategy $\s_\A$, if
    $v_H > 0$, then for all $q \in Q_H$, we have
    $\MarVal{\Aconc_H^{\s_\A}}(q) = 1$ (i.e. in the sub-game formed by
    the EC $H$, the value of state $q$ is 1).
  \end{itemize}
  \label{lem:uniformly_optimal}
\end{lemma}

Note that what is proved in the Appendix is slightly more general than Lemma~\ref{lem:uniformly_optimal} as it deals with an arbitrary valuation of the states, not only the one giving the value of the game. This generalization will be used in particular to prove that a positional strategy is $\varepsilon$-optimal from all states in Subsection~\ref{subsec:varepsilon_buchi}. Note that Lemma~\ref{lem:uniformly_optimal} is generalized and deals with arbitrary strategies and arbitrary prefix-independent objectives in \cite{bordais2022subgame}.

\section{Playing Optimally with Positional Strategies in Büchi Games}
\label{sec:optimal_in_buchi}
In this section, we focus on B\"uchi objectives. We will distinguish
several frameworks. First we consider the case when optimal strategies
exist from every state and show that in that case, there is a
positional strategy that is uniformly optimal. Furthermore, we show
that this always occurs as soon as all local interactions at states
not in the target are \emph{reach-maximizable}, the condition ensuring
the existence of optimal strategies in reachability
games~\cite{BBSCSL22}. Finally, we study the game forms necessary and
sufficient to ensure the existence of positional almost-optimal
strategies (i.e.  $\varepsilon$-optimal strategies, for all
$\varepsilon>0$).

\subsection{Playing optimally when optimal strategies exist from every state}
\label{subsec:optimal_everywhere_buchi}
First we observe that the value of a B\"uchi game is closely related
to the value of a well-chosen reachability game.  We let
$\G = \Games{\Aconc}{\Bu(T)}$ be a concurrent B\"uchi game.  We modify
$\G$ by replacing every state $q$ in the target $T$, whose value in
the Büchi game is $u := \MarVal{\G}(q)$ by a state that has
probability $u$ to go to $\top$ (the new target in the reachability
game) and probability $1-u$ to go to a sink non-target state
$\bot$. In that case, the value of the original Büchi game is the same
as the value of the obtained reachability game, which we denote
$\G_{\mathsf{reach}} := \Games{\Aconc_{\mathsf{reach}}}{\Reach(\{ \top
  \})}$. The transformation is formally described in
Appendix~\ref{app:buchi2reach}. It is straightforward to
show 
that the values of both games are the same from all states, and that
optimal strategies in the B\"uchi game will be also optimal in the
reachability game. Vice-versa, optimal strategies in the
  reachability game can be lifted to the B\"uchi game by augmenting it
  with a locally optimal strategy at target states.




\begin{proposition}[Proof in Appendix~\ref{appen:proof_prop_positional_suffice_buchi}]
  For all B\"uchi games 
  in which an optimal strategy exists from every state, 
  there exists a Player $\A$ positional strategy that is uniformly optimal.
	\label{prop:positional_suffice_buchi}
\end{proposition}

\subsection{Game forms ensuring the existence of optimal strategies}

Proposition~\ref{prop:positional_suffice_buchi} assumes the existence
of optimal strategies from every state. However, there exist B\"uchi
games with no optimal strategies, and even for which almost-optimal
strategies require infinite memory. An example of such a game is
depicted in Figure~\ref{fig:local_optimal_not_uniformly} (some details
will be given in Subsection~\ref{subsec:varepsilon_buchi}). This
justifies the interest of having Büchi games in which optimal
strategies exist from every state ``by design''.

In \cite{BBSCSL22}, the authors have proven a necessary and sufficient
condition on game forms to ensure the existence of optimal strategies
in all finite reachability games using these game forms as local interactions. This condition, called \emph{reach-maximizable} game forms (RM for short) is not detailed here but given in Appendix~\ref{subsubsec:partial_val_reach}. This formalizes as follows:
\begin{theorem}[Lem 33 and Thm 36 in \cite{BBSCSL22}]
  In all reachability games $\G = \Games{\Aconc}{\Reach(T)}$ where all
  interactions at states in $Q \setminus T$ are RM, there exist
  positional uniformly optimal strategies (for Player $\A$).
  Furthermore, if a game form $\formNF$ is not RM, one can build a
  reachability game where $\formNF$ is the only non-trivial
  interaction, in which there is no optimal strategy for Player $\A$.
	\label{thm:rm_in_reach}
\end{theorem}
Proposition~\ref{prop:buchi_equal_reach} translates a
Büchi game $\G = \Games{\Aconc}{\Bu(T)}$ into a reachability game
$\G_{\mathsf{reach}} := \Games{\Aconc_{\mathsf{reach}}}{\Reach(\{ \top
  \})}$ while keeping the same value and the same local interactions
in states outside the target $T$; furthermore, 
the local interactions in all states in $T$ in $\G_{\mathsf{reach}}$
become trivial (i.e. the outcome 
is independent of the actions of the players), which are 
RM. Hence, if all local interactions at states in $Q \setminus T$ 
in the original game $\G$ are RM, then all game forms in $\G_{\mathsf{reach}}$ are RM, 
thus there is a uniformly optimal positional strategy for Player $\A$ in that game by Theorem~\ref{thm:rm_in_reach}. Such a strategy can
then be translated back into the Büchi game to get a positional Player
$\A$ uniformly optimal strategy.
We obtain the following result.

\begin{theorem}[Proof in Appendix~\ref{appen:proof_lem_rm_in_buchi}]
  In all Büchi games $\G = \Games{\Aconc}{\Bu(T)}$ where all
  interactions at states in $Q \setminus T$ are RM, there exist
  positional uniformly optimal strategies (for Player
  $\A$). Furthermore, if a game form $\formNF$ is not RM, one can
  build a Büchi game where $\formNF$ is the only non-trivial
  interaction, in which there is no optimal strategy for Player
  $\A$. \label{lem:rm_in_buchi}
\end{theorem}


\subsection{GFs ensuring the existence of positional almost-optimal strategies}
\label{subsec:varepsilon_buchi}

While positional strategies are sufficient to play optimally in
B\"uchi games where optimal strategies do actually exist, the case is
really bad for those B\"uchi games where optimal strategies do not
exist. Indeed, it can be the case that infinite memory is necessary to play almost-optimally, as we illustrate in the next paragraph. Then we characterize game forms that ensure the existence of positional almost-optimal strategies.

\subparagraph*{An example of a B\"uchi game where playing
  almost-optimally requires infinite memory.} \label{ref:Buchi-cex}

Consider the Büchi game $\G = \Games{\Aconc}{\Bu(\top)}$ in
Figure~\ref{fig:local_optimal_not_uniformly}
. As argued earlier (in Section~\ref{sec:CNS_uniform_optimal}), there are no optimal strategies in
that game. 
First, notice that the value of the game at state $q_0$ is $1$. However, any finite-memory strategy (i.e. a strategy that can be described with a finite automaton) has value $0$. Indeed, consider such a Player $\A$ 
strategy $\s_\A$. There is some probability $p > 0$ such that if $\s_\A$ plays
an action with a positive probability, this probability is at least
$p$. If the strategy $\s_\A$ plays, at state $q_0$ the bottom row with
positive probability 
and if Player $\B$ plays the right column with probability $1$, then
state $\bot$ is reached with probability at least $p$. If this
happens infinitely often, then the state $\bot$ is eventually reached 
almost-surely: the value of the strategy 
is then $0$. Hence, Player $\A$ has to play, 
from some time on, 
the top row with probability 1. From that time on, Player $\B$ can play the left column with probability $1$, leading to avoid state $\top$ almost-surely. Hence the value of strategy $\s_\A$ is $0$ as well.  
In fact, all $\varepsilon$-optimal strategies (for every $\varepsilon>0$) cannot be finite-memory strategies. A Player $\A$ strategy with value at least $1 - \varepsilon$ (with $0 < \varepsilon < 1$) will have to play the bottom row with positive probability infinitely often, but that probability will have to
decrease arbitrarily close to $0$. More details are given in
Appendix~\ref{appen:buchi_varepsilon_infinite}.

\begin{figure}
	\begin{minipage}[b]{0.3\linewidth}
		\centering
		\includegraphics[scale=1]{formNFBu.pdf}
		\caption{The game form used as local interaction at $q_0$.}
		\label{fig:buchi_form}
	\end{minipage}
	\hfill
	\begin{minipage}[b]{0.33\linewidth}
		\centering
		\includegraphics[scale=0.8]{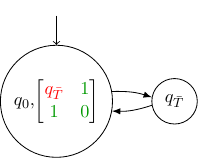}
		\caption{The Büchi game $\G_1$.}
		\label{fig:buchi_ok}
	\end{minipage}
	\hfill
	\begin{minipage}[b]{0.33\linewidth}
		\centering
		\includegraphics[scale=0.8]{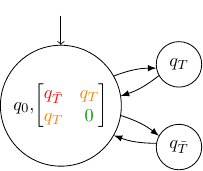}
		\caption{The Büchi game $\G_2$. 
		}
		\label{fig:buchi_not_ok}
	\end{minipage}
\end{figure}

\subparagraph*{Game forms ensuring the existence of positional
  $\varepsilon$-optimal strategies.} Below we characterize the game
forms which ensure the existence of positional almost-optimal
strategies in B\"uchi games. The approach is inspired by the one
developed in~\cite{BBSCSL22} for reachability games.

We start by discussing an example, and then generalize the
approach. Let us consider the game form
$\formNF
$ 
depicted in Figure~\ref{fig:buchi_form}
, where $\outComeNF = \{x,y,z\}$. We embed this game form into two
different environments, depicted in Figures~\ref{fig:buchi_ok}
and~\ref{fig:buchi_not_ok}. These define two B\"uchi
games 
using the following interpretation: (a) values \textcolor{green}{$0$}
and \textcolor{green}{$1$} in green represent output values giving the
probability to satisfy the B\"uchi condition when these outputs are
selected; (b) other outputs lead to either orange state
\textcolor{orange}{$q_T$} (
a target for the B\"uchi condition) or red state
\textcolor{red}{$q_{\bar{T}}$} (
not a target). In particular, the game of
Figure~\ref{fig:buchi_not_ok} is another representation of the game of
Figure~\ref{fig:local_optimal_not_uniformly}.
  
Let us compare these two games. First notice that there are no optimal
strategies in both cases, as already argued for the game in
Figure~\ref{fig:buchi_not_ok}; the arguments are similar for the game
of Figure~\ref{fig:buchi_ok}.
Interestingly, in the game $\G_1$ in Figure~\ref{fig:buchi_ok}, there
are positional almost-optimal strategies (it is in fact a reachability
game) whereas there are none in the game $\G_2$ in
Figure~\ref{fig:buchi_not_ok} (as already discussed above); despite
the fact that the local interaction at $q_0$ valued with $\mu_v$, for
$v$ the value vector of the game, is that of
Figure~\ref{fig:valued_local_iter_not_uniform} in both cases.

We analyze the two settings to better understand the differences. A
positional $\varepsilon$-optimal strategy $\s_\A$ at $q_0$ in both
games has to be an $\varepsilon$-optimal $\GF$-strategy
$\s_\A(q_0) \in \Dist(\A)$ in the game in normal form
$\gameNF{\formNF_{q_0}}{\mu_v}$ (similarly to how a uniformly optimal
positional strategy needs to be locally optimal
(Lemma~\ref{lem:uniformly_optimal})).  Consider for instance the
Player $\A$ positional strategy $\s_\A$ that plays (at $q_0$) the top
row with probability $1 - \varepsilon$
(which is $\varepsilon$-optimal in
$\gameNF{\formNF_{q_0}}{\mu_v}$). This strategy has value
$1-\varepsilon$ in $\G_1$, but has value $0$ in $\G_2$. In both cases,
if Player $\B$ plays the left column, the target is seen infinitely
often almost-surely (since the bottom row is played with positive
probability). The difference arises if Player $\B$ plays the right
column with probability $1$: in $\G_1$, the target is reached and
never left with probability $1 - \varepsilon$; however, in $\G_2$,
with probability $1-\varepsilon$, the game visits the target but loops
back to $q_0$, hence playing for ever this strategy leads with
probability $1$ to the green value \textcolor{green}{$0$}.
Actually, $\s_\A$ 
is $\varepsilon$-optimal 
if for each column, either there is
no green value but at least one orange \textcolor{orange}{$q_T$}, or
there are green values and their average 
is at least $1-\varepsilon$.

\medskip This intuitive explanation can be generalized to any game
form $\formNF$ as follows, which we will embed in several
  environments. To define an environment, we fix (i) a
partition $\outComeNF = \outComeNFLp \uplus \outComeNFEx$ of the
outcomes ($\mathsf{Lp}$ stands for loop -- i.e. orange and red
outcomes -- and $\mathsf{Ex}$ stands for exit -- i.e. the green
outcomes), (ii) a partial valuation of the outcomes
$\alpha : \outComeNFEx \to [0,1]$ and a probability
$p_T : \outComeNFLp \to [0,1]$ to visit the target $T$ in the next
step and loop back (above, the probability was either $1$ (represented
by orange state \textcolor{orange}{$q_T$}) or $0$ (represented by red
state \textcolor{red}{$q_{\bar{T}}$})). We then consider the B\"uchi
game $\G^\Bu_{\formNF,\alpha,p_T}$ which embeds game form $\formNF$ in
the environment given by $\alpha$ and $p_T$
as follows: an outcome $o \in \outComeNFLp$ leads to $q_T$ with
probability $p_T(o)$ and $q_{\bar{T}}$ otherwise; from $q_T$ or
$q_{\bar{T}}$, surely we go back to $q_0$; an outcome
$o \in \outComeNFEx$ leads to the target $T$ with probability
$\alpha(o)$ and outside $T$ otherwise (in both cases, it stays there
forever); this is formally defined in
Appendix~\ref{appen:proof_prop_varepsilon_buchi_necessary}. We want to
specify that there are positional $\varepsilon$-optimal strategies in
the game $\G^\Bu_{\formNF,\alpha,p_T}$ if and only if there are
$\varepsilon$-optimal $\GF$-strategies in the game form $\formNF$
ensuring the properties described above. However, to express what is
an $\varepsilon$-optimal strategy, we need to know the value of the
game $\G^\Bu_{\formNF,\alpha,p_T}$ at state $q_0$: to do so, we
use~\cite{AM04b}, in which the value of concurrent games with parity
objectives (generalizations of Büchi and co-Büchi objectives) is
computed using $\mu$-calculus. In our case, this can be expressed with
nested 
fixed point operations, as described in
Appendix~\ref{subsubsec:gf_in_buchi}.  We denote this value
$u_{\formNF,\alpha,p_T}^{\Bu}$ or simply $u$.  We can now define
almost-Büchi maximizable (aBM for short) game form.

\begin{definition}[Almost-Büchi maximizable game forms]
  Consider a game form $\formNF$, a partition of the outcomes $\outComeNF = \outComeNFLp \uplus \outComeNFEx$, a partial valuation of the outcomes $\alpha : \outComeNFEx \to [0,1]$ and a probability $p_T : \outComeNFLp \to [0,1]$ to visit the target $T$.
  The game form $\formNF$ is
  \emph{almost-Büchi maximizable} (aBM for short)
  w.r.t. $\alpha$ and $p_T$ if for all $0 < \varepsilon < u$, there
  exists a $\GF$-strategy $\sigma_\A \in \Dist(\St_\A)$ such that, for
  all $b \in \St_\B$, letting
  $A_b := \{ a \in \Supp(\sigma_\A) \mid \outCNF(a,b) \in
  \outComeNFEx \}$, either:
  \begin{itemize}
  \item $A_b = \emptyset$, and there exists $a \in \Supp(\sigma_\A)$
    such that $p_T(\outCNF(a,b)) > 0$ (i.e. if all outcomes loop back
    to $q_0$, there is a positive probability to visit $T$: left
    column in the game $\G_2$);
  \item $A_b \neq \emptyset$, and
    $\sum_{ a \in A_b } \sigma_\A(a) \cdot \alpha(\outCNF(a,b)) \geq
    (u - \varepsilon) \cdot \sigma_\A(A_b)$ (i.e. for the action $b$,
    the value of $\sigma_\A$ restricted to the outcomes in
    $\outComeNFEx$ is at least $u - \varepsilon$: right
    column in the game $\G_1$).
  \end{itemize}
  The game form $\formNF$ is \emph{almost-B\"uchi maximizable (aBM for
    short)} if for all partitions
  $\outComeNF = \outComeNFLp \uplus \outComeNFEx$, for all partial
  valuations $\alpha: \outComeNFEx \rightarrow [0,1]$ and probability
  function $p_T: \outComeNFLp \rightarrow [0,1]$, it is aBM
  w.r.t. $\alpha$ and $p_T$.
  \label{def:epsilon_buchi_maximizable}
\end{definition}

This definition relates to the existence of positional almost-optimal
strategies in $\G^\Bu_{\formNF,\alpha,p_T}$:
\begin{lemma}[Proof in Appendix~\ref{appen:proof_prop_varepsilon_buchi_necessary}]
  The game form $\formNF$ is aBM w.r.t. $\alpha$ and $p_T$ if and only
  if there are positional almost-optimal
  strategies 
  from 
  state $q_0$ in the game $\G^\Bu_{\formNF,\alpha,p_T}$.
  \label{lem:varepsilon_buchi_necessary}
\end{lemma}

More interestingly, if a concurrent Büchi game has all its game forms
aBM (possibly except at target states), then there exist positional
almost-optimal strategies.
\begin{theorem}[Proof in
  Appendix~\ref{appen:proof_lem_varepsilon_buchi_sufficient}]
  Consider a Büchi game $\G = \Games{\Aconc}{\Bu(T)}$ and assume that
  all local interactions at states in $Q \setminus T$ are aBM. Then,
  for every $\varepsilon > 0$, there is a
  positional strategy that is $\varepsilon$-optimal from every state
  $q \in Q$.
	\label{lem:varepsilon_buchi_sufficient}
\end{theorem}


\begin{proofs} 
  Let $\varepsilon > 0$. We build a positional Player $\A$ strategy
  $\s_\A$ and then apply (a slight generalization of)
  Lemma~\ref{lem:uniformly_optimal} to show that it is
  $\varepsilon$-optimal. Let $v := \MarVal{\G}$ be the value vector of
  the game and $u \in v[Q] \setminus \{ 0 \}$ be some positive value
  of the game. Consider the set $Q_u := v^{-1}[\{u\}] \subseteq Q$ of
  states whose values w.r.t. $v$ is $u$. We 
  define the strategy $\s_\A$ on each 
  $Q_u$ for $u \in v[Q]$ 
  and then glue the portions of $\s_\A$ together. This is possible
  because 
  the Player $\A$ strategy we
  build 
  enforces end-components in which the value given by $v$ of all
  states is the same (similarly to Proposition~\ref{prop:ec_same_vale}
  for locally optimal strategies).

	
  In Figure~\ref{fig:explain1}, the set $Q_u$ corresponds to the white
  area. Target states ($T$) are in orange, while non-target states are
  in red ($q_1,q_2,q_3$ in the figure).  From every state of $Q_u$,
  there may be several split arrows, which correspond to choices by
  Player $\B$ (actions $b \in \St_\B$); black split edges stay within
  $Q_u$ while blue edges partly lead outside $Q_u$; once a split edge
  is chosen by Player $\B$, Player $\A$ may ensure any leaving edge
  with some positive probability.

  In a state $q \in Q_u \cap T$, the strategy $\s_\A$ only needs to be
  locally optimal, i.e. such that
  $\va_{\gameNF{\formNF_q}{\mu_v}}(\s_\A(q)) = v(q)$. Then, the states
  in $Q_u \setminus T$ will be considered one by one. Let
  $q \in Q_u \setminus T$, we consider the Büchi game
  $\G^\Bu_{\formNF_q,\alpha,p_T}$ built from the aBM game form
  $\formNF_q$ and the immediate environment of $q$ as follows:
  the partial valuation $\alpha$ of the outcomes (i.e. the Nature
  states) is defined on those with a positive probability to reach
  $Q \setminus Q_u$ (i.e. 
  the green area -- green outcomes in Figures~\ref{fig:buchi_ok},~\ref{fig:buchi_not_ok}), and $p_T$ maps
  a Nature state $d$ staying in $Q_u$ 
  to the probability $\distribFunc(d)[Q_u \cap T]$ to reach the
  target $T$ (i.e. the probability to reach the orange states in
  Figure~\ref{fig:explain1} -- they correspond to the orange and red
  outcomes in Figures~\ref{fig:buchi_ok},~\ref{fig:buchi_not_ok}).
  
  In Figure~\ref{fig:explain1}, states in red are non-winning yet
  (non-target in the first stage) and therefore if Player $\B$ can
  choose a black split-edge leading only to red states, then the value
  of game $\G^\Bu_{\formNF_q,\alpha,p_T}$ is $0$ (this is the case of
  $q_2$). On the other hand, if all split-edges are either blue or
  black with at least one orange end, then the value of game
  $\G^\Bu_{\formNF_q,\alpha,p_T}$ is positive (this is the case of
  states $q_1$ and $q_3$).  Then, we realize that there must be at
  least one state $q \in Q_u \setminus T$ such that
  $u_{\formNF_q,\alpha,p_T}^{\Bu} \geq u$ (this is a key argument, and
  it is due to the definition of 
  $u_{\formNF_q,\alpha,p_T}^{\Bu}$, which relates to how the value in
  B\"uchi games is computed via fixed-point operations). In
  Figure~\ref{fig:explain1}, there are actually two such states, $q_1$
  and $q_3$: the value at $q_1$ is $1$ while the value at $q_3$ is an
  average of the values at the two (green) ends of the (blue)
  split-edge. 
  We can then use the facts that $\formNF_{q_1}$ and $\formNF_{q_3}$
  are aBM and apply the definition with a well-chosen
  $0 < \varepsilon_u \leq \varepsilon$ to obtain the $\GF$-strategies
  played by the strategy $\s_\A$ at states $q_1$ and $q_3$.

  We then iterate the process by going to the second stage by
  considering that the previously dealt states ($q_1$ and $q_3$ in our
  example)
  are now orange, as in Figure~\ref{fig:explain2}: they are now
  considered as targets. The property that
  $u_{\formNF_q,\alpha,p_T}^{\Bu} \geq u$ (with a new $p_T$ taking
  into account the larger set of targets) then propagates throughout
  the game to all states in the white area. The strategy $\s_\A$ is
  now fully defined on $Q_u$, and we can check that, under any Player
  $\B$
  strategy, 
  if the game eventually stays within an EC within $Q_u$, it will
  reach the target $T$ infinitely often almost-surely. Indeed, from
  each state, there is either a positive probability to leave the EC
  (i.e. see a green outcome) or a positive probability to get closer
  to the target (i.e. see an orange outcome).
\end{proofs}

\begin{figure}
	\begin{minipage}[b]{0.49\linewidth}
		\centering
		\includegraphics[scale=0.3]{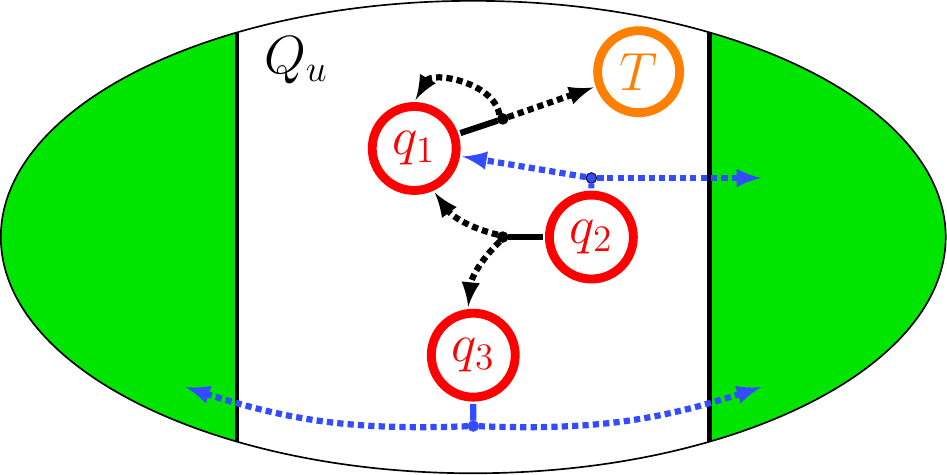}
		\caption{First stage of the construction of strategy
                  $\s_\A$ within $Q_u$.}
		\label{fig:explain1}
	\end{minipage}
	\hfill
	\begin{minipage}[b]{0.49\linewidth}
		\centering
		\includegraphics[scale=0.3]{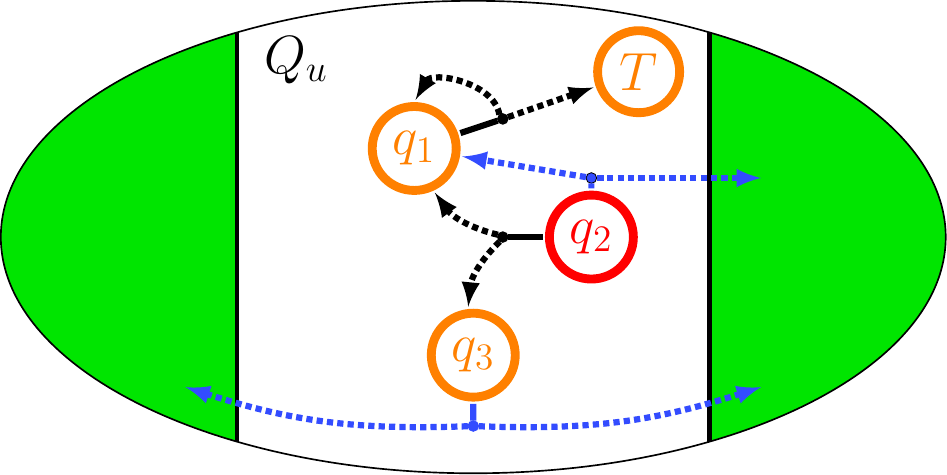}
		\caption{Second stage of the construction of strategy
                  $\s_\A$ within $Q_u$.}
		\label{fig:explain2}
	\end{minipage}
\end{figure}


With Lemma~\ref{lem:varepsilon_buchi_necessary} and
  Theorem~\ref{lem:varepsilon_buchi_sufficient}, we obtain a corollary
  analogous to Theorem~\ref{lem:rm_in_buchi} for aBM game forms and the
  existence of almost-optimal strategies.
\begin{corollary}
  In all Büchi games $\G = \Games{\Aconc}{\Bu(T)}$ where all
  interactions at states in $Q \setminus T$ are aBM, there exist
  positional uniformly almost-optimal strategies (for Player $\A$).
  Furthermore, if a game form $\formNF$ is not aBM, one can
    build a Büchi game where $\formNF$ is the only non-trivial
    interaction, in which there is no positional almost-optimal
    strategy for Player $\A$.
	\label{coro:aBM_in_buchi}
\end{corollary}
Note that game forms that are not aBM may behave well in some
environments, even though they do not behave well in some other
environments. Hence, it might be the case that in a specific Büchi game
with local interactions that are not aBM, there are some positional
almost-optimal strategies. This observation stands for all the type of
game forms we have characterized: RM, aBM and also coBM in the next
section.

Finally, note that it is decidable whether a game form is aBM.
\begin{proposition}[Proof in Appendix~\ref{appen:ebm_rm_gf}]
	It is decidable if a game form is aBM. Moreover, all RM game forms are aBM game forms and there is a game form that is aBM but not RM.
	\label{prop:rm_ebm}
\end{proposition}


\section{Playing Optimally in co-Büchi Games}
\label{sec:co-Buchi}

Although they may seem quite 
close, Büchi and co-Büchi objectives do not enjoy the same properties
in the setting of concurrent games. For instance, we have seen that
there are Büchi games in which a state has value $1$ but any
finite-memory strategy has value $0$.
This cannot happen in co-Büchi games, since strategies with values
arbitrarily close to $1$ can be found among positional
strategies~\cite{CAH06}. We have also seen in
Section~\ref{sec:optimal_in_buchi} that in all concurrent Büchi games,
if there is an optimal strategy from all states, then there is a
uniformly optimal positional strategy. As we will see in the next
subsection, this does not hold for co-Büchi games. This shows
  that concurrency in games complicates a lot the model: the results
  of this paper has to be put in regards with the model of turn-based
  games where pure positional strategies are sufficient to play
  optimally, for parity
  objectives~\cite{MM02
  }.

\subsection{Optimal strategies may require infinite memory in
  co-B\"uchi games}
\label{subsec:optimal_strat_in_co_buchi}

Consider the game depicted in Figure~\ref{fig:co_buchi_infinite},
which uses the same convention as in the previous section. The
objective is $W = \coBu(\{ q_T,q_T',\top \})$ where the state $\top$
is not represented, but implicitly present via the green values (for
instance, green value \textcolor{green}{$1/2$} leads to $\top$ with
probability $1/2$ and to $\bot$ with probability $1/2$). Let
$A := \{ a_1,a_2,a_3 \}$ with $a_1$ the top row and $a_3$ the bottom
row and $B := \{ b_1,b_2,b_3 \}$ with $b_1$ the leftmost column and
$b_3$ the rightmost column. 
If a green value is not reached, Player $\A$ wins if and only if eventually, the red state is not seen anymore.
The values of the states
$q_0,q_T,q_{T'}$ and $q_{\bar{T}}$ are the same and are at most $1/2$. Indeed, if
Player $\B$ almost-surely plays $b_3$ at $q_0$,
she ensures that the value of the game from $q_0$ is at most
$1/2$. Let us argue that any Player $\A$ positional strategy has value
less than $1/2$ and exhibit an infinite-memory Player $\A$ strategy
whose value is $1/2$.

Consider a Player $\A$ positional strategy $\s_\A$.
We define a Player $\B$ strategy $\s_\B$ 
as follows: if $\s_\A(q_0)(a_3) = \varepsilon > 0$,
then we set $\s_\B(q_0)(b_3) := 1$ 
and the value of the game w.r.t. $\s_\A,\s_\B$ is at most $1/2-\varepsilon<1/2$.
If 
$\s_\A(q_0)(a_1) = 1$, 
we set $\s_\B(q_0)(b_2) := 1$ and the
state $q_T' \in T$ is visited infinitely often almost-surely. Otherwise, $\s_\A(q_0)(a_2) > 0$ and $\s_\A(q_0)(a_3) = 0$, hence
choosing $\s_\B(q_0)(b_1) := 1$
ensures that the state $q_T \in T$ is visited infinitely often
almost-surely. 
In the last two cases, $\s_\A$ has value $0$. 
Overall, any positional strategy 
$\s_\A$ has value 
less than $1/2$.

We briefly describe a Player $\A$ optimal strategy $\s_\A$ (whose
value is $1/2$). The idea is the following: along histories that have
not visited 
$q'_T$ yet (this happens 
when Player $\B$ has not played $b_2$), $\s_\A$ plays $a_1$ with very high
probability $1 - \varepsilon_k < 1$ and $a_2$ with probability
$\varepsilon_{k} > 0$, where $k$ denotes the number of steps.  The
values $(\varepsilon_{k})_{k \in \N}$ are chosen so that, if Player
$\B$ only plays $b_1$ with probability $1$, then the state $q_T \in T$
is seen finitely often almost-surely. Some details are
given in Appendix~\ref{appen:infinite_memory_cobuchi}. After the first
visit to $q'_T$, Player $\s_\A$ switches to a positional strategy of
value $\frac{1}{2} - \varepsilon'_{k}$, for $\varepsilon'_k>0$ and
$k$ is the number of steps after the first visit to $q'_T$. A first
visit to $q'_T$ occurs when 
$b_2$ is played by Player $\B$.
The value of $\s_\A$ after that point is then
$(1 - \varepsilon_{k}) \cdot (\frac{1}{2} - \varepsilon'_{k}) +
\varepsilon_{k} \cdot 1$
. It suffices to choose $\varepsilon'_k$ small enough
so that 
the above value is at least $1/2$. 
Such a Player $\A$ strategy is optimal from $q_0$. 
It follows that, contrary to the Büchi case, requiring that positional optimal strategies exists from all states in a co-Büchi game is stronger than requiring that optimal strategies exists from all states.



\begin{figure}
	\begin{minipage}[b]{0.3\linewidth}
		\hspace*{-1.5cm}
		\centering
		\includegraphics[scale=0.75]{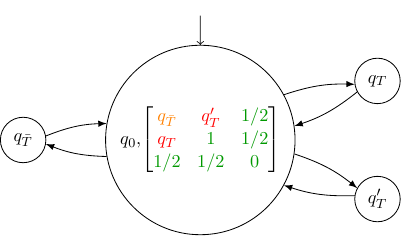}
		\caption{The co-Büchi game $\G_1$.}
		\label{fig:co_buchi_infinite}
	\end{minipage}
        \hfill
	\begin{minipage}[b]{0.3\linewidth}
		\hspace*{-0.5cm}
		\centering
		\includegraphics[scale=1]{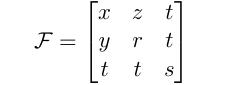}
		\caption{The local interaction at state
                  $q_0$. 
              }
		\label{fig:game_form_cobuchi}
	\end{minipage}
        \hfill
	\begin{minipage}[b]{0.3\linewidth}
		\centering
		\includegraphics[scale=0.75]{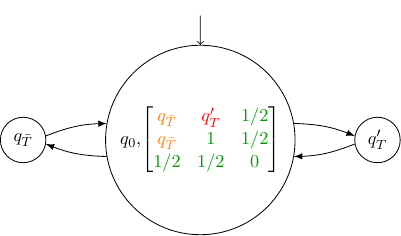}
		\caption{Another co-Büchi game ($\G_2$) built on $\formNF$.
              }
		\label{fig:game_cobuchi}
	\end{minipage}
\end{figure}

\subsection{GFs in co-Büchi games which ensure positional
  optimal strategies}

Contrary to the Büchi objectives, RM game forms do not suffice to
ensure the existence of positional uniformly optimal strategies in
co-Büchi games,
see Proposition~\ref{prop:rm_cobm}.
In this subsection, we characterize the game forms ensuring this
property.

We proceed as in Subsection~\ref{subsec:varepsilon_buchi}, and we
explain the approach on an example. Consider the game form $\formNF$
depicted in Figure~\ref{fig:game_form_cobuchi}, that is the local
interaction at state $q_0$ of the game in
Figure~\ref{fig:co_buchi_infinite}. Let
$\outComeNF = \outComeNFEx \uplus \outComeNFLp$ for
$\outComeNFEx := \{ t,r,s \}$, $\outComeNFLp := \{ x,y,z \}$ and
consider the partial valuation
$\alpha: \outComeNFEx \rightarrow [0,1]$ such that $\alpha(t) := 1/2$,
$\alpha(r) := 1$ and $\alpha(s) := 0$. We then consider two
probability functions $p_T^1,p_T^2: \outComeNFLp \rightarrow [0,1]$
such that $p_T^1(x),p_T^2(x) := 0$, $p_T^1(z),p_T^2(z) := 1$,
$p_T^1(y) := 1$ whereas $p_T^2(y) := 0$. The games
$\G_1 := \G^\coBu_{\formNF,\alpha,p_T^1}$
(Figure~\ref{fig:co_buchi_infinite}) and
$\G_2 := \G^\coBu_{\formNF,\alpha,p_T^2}$
(Figure~\ref{fig:game_cobuchi}) are defined just like their Büchi
counterparts, except for the objective which is $W = \coBu(T)$ with
$q_T,q_T' \in T$ and $q_{\bar{T}}
\notin
T$. 
We have already argued 
that there is no positional
optimal strategy in the game $\G_1$
(Figure~\ref{fig:co_buchi_infinite}). The issue is the following: when
considering a locally optimal strategy (
recall Lemma~\ref{lem:uniformly_optimal}), there is a column where, in the
support of the strategy, there is a red outcome (i.e. positive
probability to reach $T$) and no green outcome
. 
Then, if Player $\B$ plays, with
probability $1$, the corresponding action
, she 
ensures that the set $T$ 
is visited infinitely often almost-surely. Now, in the game $\G_2$ of
Figure~\ref{fig:game_form_cobuchi}, any Player $\A$ positional
strategy playing $a_3$ with probability $0$ and $a_2$ with positive
probability is uniformly optimal. Indeed, in that case, (i) if $b_1$
is played, then the set $T$ is not seen; (ii) if $b_2$ or $b_3$ are
played, then the game will end in a green outcome almost-surely.

In the general case, as for B\"uchi games, the value of co-Büchi games
can be computed with nested fixed points (see
Appendix~\ref{subsubsec:gf_in_cobuchi}). Using this value, we can
define the notion of co-Büchi maximizable game forms (coBM for short),
which, while requiring a slightly more complex setting,
  formalizes the intuition given in the previous example, see
  Appendix~\ref{proof:thm_coBM_necessary}. It ensures a lemma
analogous to Lemma~\ref{lem:varepsilon_buchi_necessary} in the context
of co-Büchi games.

Furthermore, 
coBM game forms also ensure that, when they are used in a co-Büchi game, there always exists a positional uniformly optimal strategy.
\begin{theorem}[Proof in Appendix~\ref{proof:thm_coBM_sufficient}]
  In all co-Büchi games $\G = \Games{\Aconc}{\coBu(T)}$ where all
  local interactions at states in $T$ are RM and all local
  interactions at states in $Q \setminus T$ are coBM, there exist
  positional uniformly optimal strategies (for Player
  $\A$). Furthermore, if a game form $\formNF$ is not coBM, one can
  build a co-Büchi game where $\formNF$ is the only non-trivial
  interaction where there is no positional optimal strategy for Player
  $\A$.
  \label{lem:co_buchi_sufficient}
\end{theorem}

\begin{proposition}[Proof in Appendix~\ref{appen:gf_rm_not_coBM}]
  It is decidable if a game form is coBM. All coBM game forms are RM
  game forms. There exists a game form that is RM but not
  coBM. 
  \label{prop:rm_cobm}
\end{proposition}

\section{Conclusion}
\definecolor{mygreen}{rgb}{0, 0.2, 0.4}
\definecolor{myyellow}{rgb}{0, 0.48, 0.65}
\definecolor{myorange}{rgb}{0, 0.72, 0.92}
\definecolor{myred}{rgb}{0.61, 0.87, 1}
\begin{table}[t]
	\centering
	\renewcommand{\arraystretch}{0.9}
	\begin{tabular}{|c|c|c|c|c|}
		\hline
		\multirow{2}{2cm}{} & \multicolumn{2}{c|}{\textbf{Positional Optimal Strategy}} & \multicolumn{2}{c|}{\textbf{Positional $\varepsilon$-Optimal Strategy}}\\
		\cline{2-5}
		& $T$ & $Q \setminus T$ & $T$ & $Q \setminus T$\\
		\hline
		$\Safe(T)$ & \phantom{AllAll} {\color{mygreen}\bf All} \phantom{AllAll} & \phantom{AllAll} {\color{mygreen}\bf All} \phantom{AllAll}& \phantom{AllAll} {\color{mygreen}\bf All} \phantom{AllAll}& \phantom{AllAll} {\color{mygreen}\bf All} \phantom{AllAll}\\ \hline
		$\Reach(T)$ & {\color{mygreen}\bf All} & {\color{myorange}\bf RM} & {\color{mygreen}\bf All} & {\color{mygreen}\bf All} \\ \hline
		$\Bu(T)$ & {\color{mygreen}\bf All} & {\color{myorange}\bf RM} & {\color{mygreen}\bf All} & {\color{myyellow}\bf aBM}  \\ \hline
		$\coBu(T)$ & {\color{myorange}\bf RM} & {\color{myred}\bf coBM} & {\color{mygreen}\bf All} & {\color{mygreen}\bf All}  \\ \hline
	\end{tabular}
	\caption{Game forms necessary and sufficient for the existence of
		positional strategies for (almost-)optimality.}
	\label{tab:GFprop}
\end{table}



We have studied game
forms 
and defined various conditions 
such that these game forms in isolation behave properly w.r.t. some
fixed property (like existence of optimal strategies for B\"uchi
objectives)
, and we have 
proven that they can be used collectively in graph games while
preserving this
property. 
These conditions, summarized in
Table~\ref{tab:GFprop}
, give the unique way to construct games which will satisfy good
memory properties by construction for playing (almost-)optimally.

Let us explicit how to read a specific row of this table, say the third one for the Büchi objective. A Büchi objective is defined along with
a target $T \subseteq Q$. The game forms necessary and sufficient to ensure
the existence of positional optimal strategies are given in the
leftmost part of the table, and to ensure the existence of positional
almost-optimal strategies, in the rightmost part of the table. For
instance, for the leftmost part, if a game form is not RM, there is a
Büchi game built from it -- where it is the only non-trivial local interaction -- in which there is no positional optimal strategy. Conversely, if in a Büchi game, all local interactions at states outside the target $T$ are RM game forms (no further assumption is made on the game forms appearing at states in
T), then there is a positional optimal strategy. The rightmost part of
the table can be read similarly.

Finally, we would like to mention that all two-variable game forms
$\formNF = \langle \St_\A,\St_\B,\outComeNF,\outCNF \rangle$
(i.e. such that $|\outComeNF| \leq 2$) are coBM (this is a direct
consequence of the definition), hence RM. From this, we obtain as a
corollary of our results that in all finite Büchi and co-B\"uchi games
where all local interactions are two-variable game forms, both players
have positional uniformly optimal strategies.
\begin{corollary}[Proof in Appendix~\ref{proof:coro_two_var_unif_optim}]
  In a Büchi or co-B\"uchi game $\G$ such that for all $q \in Q$,
  $|\delta(q,A,B)| \leq 2$, both players have a positional uniformly
  optimal strategy.
  \label{coro:two_var_unif_optim}
\end{corollary}



\bibliography{biblio}

\begin{thebibliography}{10}

\bibitem{baier2008principles}
Christel Baier and Joost-Pieter Katoen.
\newblock {\em Principles of model checking}.
\newblock MIT press, 2008.

\bibitem{BCJ18}
Roderick Bloem, Krishnendu Chatterjee, and Barbara Jobstmann.
\newblock {\em Handbook of Model Checking}, chapter Graph games and reactive
  synthesis, pages 921--962.
\newblock Springer, 2018.

\bibitem{BBSCSLarXiv}
Benjamin Bordais, Patricia Bouyer, and St{\'{e}}phane~Le Roux.
\newblock Optimal strategies in concurrent reachability games.
\newblock {\em CoRR}, abs/2110.14724, 2021.
\newblock URL: \url{https://arxiv.org/abs/2110.14724}, \href
  {http://arxiv.org/abs/2110.14724} {\path{arXiv:2110.14724}}.

\bibitem{BBSCSL22}
Benjamin Bordais, Patricia Bouyer, and St{\'{e}}phane~Le Roux.
\newblock Optimal strategies in concurrent reachability games.
\newblock In Florin Manea and Alex Simpson, editors, {\em 30th {EACSL} Annual
  Conference on Computer Science Logic, {CSL} 2022, February 14-19, 2022,
  G{\"{o}}ttingen, Germany (Virtual Conference)}, volume 216 of {\em LIPIcs},
  pages 7:1--7:17. Schloss Dagstuhl - Leibniz-Zentrum f{\"{u}}r Informatik,
  2022.
\newblock \href {https://doi.org/10.4230/LIPIcs.CSL.2022.7}
  {\path{doi:10.4230/LIPIcs.CSL.2022.7}}.

\bibitem{bordais2022subgame}
Benjamin Bordais, Patricia Bouyer, and St{\'e}phane~Le Roux.
\newblock Sub-game optimal strategies in concurrent games with
  prefix-independent objectives.
\newblock Submitted FSTTCS 2022, 2022.

\bibitem{CH07}
Krishnendu Chatterjee.
\newblock Concurrent games with tail objectives.
\newblock {\em Theor. Comput. Sci.}, 388(1-3):181--198, 2007.
\newblock \href {https://doi.org/10.1016/j.tcs.2007.07.047}
  {\path{doi:10.1016/j.tcs.2007.07.047}}.

\bibitem{CAH06}
Krishnendu Chatterjee, Luca de~Alfaro, and Thomas~A. Henzinger.
\newblock The complexity of quantitative concurrent parity games.
\newblock In {\em Proceedings of the Seventeenth Annual {ACM-SIAM} Symposium on
  Discrete Algorithms, {SODA} 2006, Miami, Florida, USA, January 22-26, 2006},
  pages 678--687. {ACM} Press, 2006.
\newblock URL: \url{http://dl.acm.org/citation.cfm?id=1109557.1109631}.

\bibitem{CAH11}
Krishnendu Chatterjee, Luca de~Alfaro, and Thomas~A. Henzinger.
\newblock Qualitative concurrent parity games.
\newblock {\em {ACM} Trans. Comput. Log.}, 12(4):28:1--28:51, 2011.
\newblock \href {https://doi.org/10.1145/1970398.1970404}
  {\path{doi:10.1145/1970398.1970404}}.

\bibitem{CJH04}
Krishnendu Chatterjee, Marcin Jurdzi{\'n}ski, and {Th}omas~A. Henzinger.
\newblock Quantitative stochastic parity games.
\newblock In {\em Proc. of 15th Annual ACM-SIAM Symposium on Discrete
  Algorithms (SODA'04)}, pages 121--130. SIAM, 2004.

\bibitem{dealfaro97}
Luca {de Alfaro}.
\newblock {\em Formal Verification of Probabilistic Systems}.
\newblock PhD thesis, Stanford University, 1997.

\bibitem{AHK07}
Luca de~Alfaro, Thomas Henzinger, and Orna Kupferman.
\newblock Concurrent reachability games.
\newblock {\em Theoretical Computer Science}, 386(3):188--217, 2007.

\bibitem{AH00}
Luca de~Alfaro and Thomas~A. Henzinger.
\newblock Concurrent omega-regular games.
\newblock In {\em 15th Annual {IEEE} Symposium on Logic in Computer Science,
  Santa Barbara, California, USA, June 26-29, 2000}, pages 141--154. {IEEE}
  Computer Society, 2000.
\newblock \href {https://doi.org/10.1109/LICS.2000.855763}
  {\path{doi:10.1109/LICS.2000.855763}}.

\bibitem{AM04b}
Luca de~Alfaro and Rupak Majumdar.
\newblock Quantitative solution of omega-regular games.
\newblock {\em Journal of Computer and System Sciences}, 68:374--397, 2004.

\bibitem{everett57}
Hugh Everett.
\newblock Recursive games.
\newblock {\em Annals of Mathematics Studies -- Contributions to the Theory of
  Games}, 3:67--78, 1957.

\bibitem{filar2012competitive}
Jerzy Filar and Koos Vrieze.
\newblock {\em Competitive Markov decision processes}.
\newblock Springer Science \& Business Media, 2012.

\bibitem{KNPS21}
Marta Kwiatkowska, Gethin Norman, Dave Parker, and Gabriel Santos.
\newblock Automatic verification of concurrent stochastic systems.
\newblock {\em Formal Methods in System Design}, 58:188--250, 2021.
\newblock \href {https://doi.org/10.1007/s10703-020-00356-y}
  {\path{doi:10.1007/s10703-020-00356-y}}.

\bibitem{KNPS+22}
Marta Kwiatkowska, Gethin Norman, David Parker, Gabriel Santos, and Rui Yan.
\newblock Probabilistic model checking for strategic equilibria-based decision
  making.
\newblock In {\em Proc. 47th International Symposium on Mathematical
  Foundations of Computer Science (MFCS'22)}, LIPIcs. Leibniz-Zentrum f{\"u}r
  Informatik, 2022.
\newblock To appear; Available as arXiv:2206.15148.

\bibitem{martin98}
Donald~A. Martin.
\newblock The determinacy of blackwell games.
\newblock {\em The Journal of Symbolic Logic}, 63(4):1565--1581, 1998.

\bibitem{MM02}
Annabelle McIver and Carroll Morgan.
\newblock Games, probability and the quantitative $\mu$-calculus $qm\mu$.
\newblock In {\em Proc. 9th International Conference on Logic for Programming,
  Artificial Intelligence, and Reasoning (LPAR'02)}, volume 2514 of {\em
  Lecture Notes in Computer Science}, pages 292--310. Springer, 2002.

\bibitem{thomas02}
Wolfgang {\relax Th}omas.
\newblock Infinite games and verification.
\newblock In {\em Proc. 14th International Conference on Computer Aided
  Verification (CAV'02)}, volume 2404 of {\em Lecture Notes in Computer
  Science}, pages 58--64. Springer, 2002.
\newblock Invited Tutorial.

\bibitem{vonNeuman}
John von Neumann and Oskar Morgenstern.
\newblock {\em Theory of Games and Economic Behavior}.
\newblock Princeton Univ. Press, Princeton, 1944.

\bibitem{zielonka04}
Wies{\l}aw Zielonka.
\newblock Perfect-information stochastic parity games.
\newblock In {\em Proc. 7th International Conference on Foundations of Software
  Science and Computation Structures (FoSSaCS'04)}, volume 2987 of {\em Lecture
  Notes in Computer Science}, pages 499--513. Springer, 2004.

\end{thebibliography}

\newpage \appendix

\noindent\textbf{\LARGE{\sffamily Technical appendix}}

\section{Probability distribution given two strategies}
\label{appen:def_prob_distrib}

In the remainder of the Appendix, we use the following notations. For a non-empty set $Q$ and
$\rho = \rho_0 \cdot \rho_1 \cdots \in Q^\omega$, for all $i \geq 0$
we denote by $\rho_{\leq i} \in Q^*$ the finite sequence
$\rho_{\leq i} = \rho_{0} \cdots \rho_{i}$ and by
$\rho_{\geq i} \in Q^\omega$ the infinite suffix
$\rho_i \cdot \rho_{i+1} \cdots$. Furthermore, for a finite path
$\rho \in Q^+$, we denote by the $\cyl(\rho) \subseteq Q^\omega$ the
cylinder generated by $\rho$, i.e.
$\cyl(\rho) := \{ \rho \cdot \rho' \in Q^\omega \mid \rho' \in
Q^\omega \}$.

\begin{definition}[Probability distribution given two strategies]
	Consider 
	$(\s_\A,\s_\B) \in \SetStrat{\Aconc}{\A} \times
	\SetStrat{\Aconc}{\B}$ two arbitrary strategies for Player
	$\A$ and $\B$. We denote by
	$\mathbb{P}_{\s_\A,\s_\B}: Q^+ \rightarrow \Dist(Q)$ the
	function giving the probabilistic distribution over the next
	state of the arena given the sequence of states already
	seen. That is, for all finite path
	$\pi = \pi_0 \ldots \pi_n \in Q^+$ and $q \in Q$, we have
	$\mathbb{P}_{\s_\A,\s_\B}(\pi)[q] :=
	\prob{}{\s_\A(\pi),\s_\B(\pi)}(\pi_n,q)$.
	
	Then, the probability of occurrence of a finite path
	$\pi = \pi_0 \cdots \pi_n \in Q^+$ from a state $q_0 \in Q$
	with the pair of strategies $(\s_\A,\s_\B)$ is equal to
	$\prob{\Aconc,q_0}{\s_\A,\s_\B}(\pi) := \Pi_{i = 0}^{n-1}
	\mathbb{P}_{\s_\A,\s_\B}(\pi_{\leq i})[\pi_{i+1}]$ if
	$\pi_0 = q_0$ and $0$ otherwise.  The probability of a
	cylinder set $\cyl(\pi)$ is
	$\prob{\Aconc,q_0}{\s_\A,\s_\B}[\cyl(\pi)] :=
	\prob{}{\s_\A,\s_\B}(\pi)$ for any finite path $\pi \in
	Q^*$. This induces the probability of any Borel set in the
	usual way, and we denote by
	$\prob{\Aconc,q_0}{\s_\A,\s_\B}: \Borel(Q) \rightarrow [0,1]$
	the corresponding probability measure.
	\label{def:prob_distrib_given_strat}
\end{definition}

\section{Proof from Section~\ref{sec:CNS_uniform_optimal}}
\label{sec:proof_lem_uniformly_optimal}
We first recall a proposition from \cite{BBSCSLarXiv} (specifically, Proposition 42 in \cite{BBSCSLarXiv}).
\begin{proposition}
	\label{prop:outcome_valuation}
	Consider a concurrent game $\G = \Games{\Aconc}{W}$ and an arbitrary valuation of the states $v \in [0,1]^Q$. Consider also a state $q \in Q$ and strategies $\sigma_\A,\sigma_\B \in \Dist(A) \times \Dist(B)$ for both players in the game form $\formNF_q$. We have the following relation:
	\[\sum_{q' \in Q} \prob{}{\sigma_\A,\sigma_\B}(q,q')
	\cdot v(q') = \outM_{\gameNF{\formNF_q}{\mu_v}}(\sigma_\A,\sigma_\B)\]
\end{proposition}

We want to prove Lemma~\ref{lem:uniformly_optimal}. To do so, 
we show a slightly more general result that can be applied to prove that a strategy is uniformly $\varepsilon$-optimal. Specifically, given a valuation $v$ of the states, we introduce the notion of locally dominating that valuation and uniformly guaranteeing that valuation 
(that is a generalization of local optimality and uniform optimality, respectively).
\begin{definition}[Locally dominating and guaranteeing a valuation]
	Consider a concurrent game $\G = \Games{\Aconc}{W}$ with a prefix-independent objective $W$ and a valuation $v: Q \rightarrow [0,1]$ of the states. We say that a Player $\A$ positional strategy $\s_\A$ \emph{locally dominates} the valuation $v$ if, for all states $q \in Q$, we have $\va_{\gameNF{\formNF_q}{\mu_v}}(\s_\A(q)) \geq v(q)$ (i.e. the value of the local strategy $\s_\A(q)$, w.r.t. the valuation $\mu_v$ of the Nature states, in the local interaction $\formNF_q$ is at least $v(q))$. Furthermore, we say that the strategy $\s_\A$ \emph{guarantees} the valuation $v$ if, for all states $q \in Q$, we have $\MarVal{\G}[\s_\A](q) \geq v(q)$. In particular, if $\MarVal{\G} - \varepsilon \preceq v$ for some $\varepsilon > 0$, the strategy $\s_\A$ is said \emph{uniformly $\varepsilon$-optimal}.
\end{definition}
Then, we prove the following lemma:
\begin{lemma}
	Consider a concurrent game $\G = \Games{\Aconc}{W}$ with $W$ prefix-independent. Assume that, in all finite MDP with objective $Q^\omega \setminus W$, there is a positional optimal strategy. Let $v: Q \rightarrow [0,1]$ be a valuation of the states and consider a Player $\A$ positional strategy $\s_\A \in \SetPosStrat{\Aconc}{\A}$ and assume that:
	\begin{itemize}
		\item it locally dominates the valuation $v$;
		\item for all ECs $H$ in the MDP induced by the strategy $\s_\A$, if $\min_{q \in H} v(q) > 0$, then for all $q \in Q_H$, we have $\MarVal{\Aconc_H^{\s_\A}}(q) = 1$ (i.e. in sub-game that is the EC $H$, the value from the state $q$ is 1).
	\end{itemize}
	Then, the strategy $\s_\A$ guarantees the valuation $v$.
	\label{lem:uniform_guarantee}
\end{lemma}
\begin{proof}
	

	
	
	Consider the finite MDP $\Gamma$ induced by the strategy $\s_\A$. By assumption, there is an optimal Player $\B$ positional strategy in $\Gamma$. The arena $\Aconc$ with positional strategies $\s_\A,\s_\B$ (or equivalently, the MDP $\Gamma$ with the strategy $\s_\B$) then consists of a Markov chain $\mathcal{M} = (Q,\mathbb{P})$ with $\mathbb{P}: Q \times Q \rightarrow [0,1]$ where, for all $q,q' \in Q$, we have $\mathbb{P}(q,q') := \prob{}{\s_\A(q),\s_\B(q)}(q,q')$. Consider an arbitrary state $q_0 \in Q$. The probability measure $\mathbb{P}$ is extended to finite paths starting at $q_0$, cylinders and arbitrary Borel sets (in particular, to $W$). 
	
	In the Markov chain $\mathcal{M}$, a bottom strongly connected component $B$ (BSCC for short) is a subset of states $B \subseteq Q$ such that for all $q \in B$, we have $\mathbb{P}_q[B] = 1$ (that is, once a BSCC is entered, it is left with probability 0). Furthermore, the underlying graph is strongly connected. Note that, all BSCCs $B$ can be seen as an EC $H_B = (B,\Supp(\s_\A))$ in the MDP $\Gamma$ (with, for all states $q \in B$, $\Supp(\s_\A)(q) = \{ b \in B \mid \s_\A(q)(b) > 0 \}$). In particular, for all BSCCs $B$, there exists $v_B \in [0,1]$, such that, for all $q \in B$, we have $v(q) = v_B$ (this is because the strategy $\s_\A$ locally dominates the valuation $v$, from Proposition 18 in \cite{BBSCSL22}).	We denote by $\mathsf{BSCC}$ the set of all BSCCs of the Markov chain $\mathcal{M}$ and by $\mathcal{B} := \cup_{B \in \mathsf{BSCC}} B$ the union of all BSCCs. Fix a BSCC $B$. For all finite sequence $\rho = q_0 \cdots q_n \in Q^*$, we denote by $\rho \in B$ the fact that the last element of $\rho$ is in $B$: $q_n \in B$. Furthermore, we denote by $\mathbb{P}(\diamondsuit B)$ the probability to eventually reach the BSCC $B$: $\mathbb{P}(\diamondsuit B) := \sum_{\substack{\rho \in q_0 \cdot Q^+ \cap B}} \mathbb{P}(\rho)$. From \cite{baier2008principles} (Theorem 10.27), we know that with probability 1, the game ends up in a BSCC. That is:
	\begin{equation}
		\forall \varepsilon > 0,\; \exists n \in \N,\; \mathbb{P}(q_0 \cdot Q^n \setminus \mathcal{B}) \leq \varepsilon
		\label{eqn:end_up_bscc}
	\end{equation}
	Alternatively, we have:
	\begin{equation}
		\sum_{B \in \mathsf{BSCC}} \mathbb{P}(\diamondsuit B) = 1
		\label{eqn:sum_reach_bscc_one}
	\end{equation}
	
	Let us now show the following equality:
	\begin{equation}
		\sum_{B \in \mathsf{BSCC}} \mathbb{P}(\diamondsuit B) \cdot v_B \geq v(q_0)
		\label{eqn:convex_bscc_ok}
	\end{equation}
	To do so, let us show by induction on $n$ the following property $\mathcal{P}(n)$: '$\sum_{\rho \in q_0 \cdot Q^n} \mathbb{P}(\rho) \cdot v(\rho) \geq v(q_0)$' where $v(\rho) \in [0,1]$ refers to $v(q)$ for $q \in Q$ the last state of $\rho$. The property $\mathcal{P}(0)$ straightforwardly holds. Assume now that this property holds for some $n \in \N$. We have, by Proposition~\ref{prop:outcome_valuation} and since $\s_\A$ locally dominates the valuation $v$:
	\begin{align*}
		\sum_{\rho \in q_0 \cdot Q^{n+1}} \mathbb{P}(\rho) \cdot v(\rho) & = \sum_{\rho = q_0 \cdots q_n \cdot q_{n+1} \in q_0 \cdot Q^{n+1}} \mathbb{P}(\rho) \cdot v(q_{n+1}) \\
		& = \sum_{\rho' = q_0 \cdots q_n \in q_0 \cdot Q^{n}} \sum_{q_{n+1} \in Q} \mathbb{P}(\rho' \cdot q_{n+1}) \cdot v(q_{n+1}) \\
		& = \sum_{\rho' = q_0 \cdots q_n \in q_0 \cdot Q^{n}} \sum_{q_{n+1} \in Q} \mathbb{P}(\rho') \cdot \mathbb{P}(q_n,q_{n+1}) \cdot v(q_{n+1}) \\
		& = \sum_{\rho' = q_0 \cdots q_n \in q_0 \cdot Q^{n}} \mathbb{P}(\rho') \cdot \sum_{q_{n+1} \in Q} \mathbb{P}(q_n,q_{n+1}) \cdot v(q_{n+1}) \\
		& = \sum_{\rho' = q_0 \cdots q_n \in q_0 \cdot Q^{n}} \mathbb{P}(\rho') \cdot \sum_{q_{n+1} \in Q} \mathbb{P}_{\s_\A(q_n),\s_\A(q_n)}(q_n,q_{n+1}) \cdot v(q_{n+1}) \\
		& = \sum_{\rho' = q_0 \cdots q_n \in q_0 \cdot Q^{n}} \mathbb{P}(\rho') \cdot \outM_{\gameNF{\formNF_{q_n}}{\mu_v}}(\s_\A(q_n),\s_\B(q_n)) \\
		& \geq \sum_{\rho' = q_0 \cdots q_n \in q_0 \cdot Q^{n}} \mathbb{P}(\rho') \cdot \va_{\gameNF{\formNF_{q_n}}{\mu_v}}(\s_\A(q_n)) \\
		& \geq \sum_{\rho' = q_0 \cdots q_n \in q_0 \cdot Q^{n}} \mathbb{P}(\rho') \cdot v(q_n) \\
		& \geq v(q_0)
	\end{align*}
	Hence, the property $\mathcal{P}(n)$ is ensured. In fact, it holds for all $n \in \N$. Consider now some $\varepsilon > 0$. Consider some $n \in \N$ as in~(\ref{eqn:end_up_bscc}). Then by the property $\mathcal{P}(n)$:
	\begin{align*}
		v(q_0) & \leq \sum_{\rho \in q_0 \cdot Q^{n}} \mathbb{P}(\rho) \cdot v(\rho) = \sum_{B \in \mathsf{BSCC}} \sum_{\rho \in q_0 \cdot Q^{n} \cap B} \mathbb{P}(\rho) \cdot v(\rho) + \sum_{\rho \in q_0 \cdot Q^{n} \setminus \mathcal{B}} \mathbb{P}(\rho) \cdot v(\rho) \\
		& \leq \sum_{B \in \mathsf{BSCC}} \sum_{\rho \in q_0 \cdot Q^{n} \cap B} \mathbb{P}(\rho) \cdot v_B + \sum_{\rho \in q_0 \cdot Q^{n} \setminus \mathcal{B}} \mathbb{P}(\rho) \\
		& \leq \sum_{B \in \mathsf{BSCC}} \mathbb{P}(\diamondsuit B) \cdot v_B + \sum_{\rho \in q_0 \cdot Q^{n} \setminus \mathcal{B}} \mathbb{P}(q_0 \cdot Q^n \setminus \mathcal{B}) \\
		& \leq \sum_{B \in \mathsf{BSCC}} \mathbb{P}(\diamondsuit B) \cdot v_B + \varepsilon
	\end{align*}
	As this holds for all $\varepsilon > 0$, we obtain Equation~(\ref{eqn:convex_bscc_ok}). 
	
	Let us now consider the probability of the event $W$ conditioned that a  BSCC $B$. Assume that $v_B > 0$ and consider the EC $H_B$ of the MDP $\Gamma$. Then, the strategy $\s_\B$ restricted to the states in $B$ is a strategy in the game $\Aconc_{H_B}^{\s_\A}$. By assumption, the value of the value with such a strategy is $1$. We obtain that, for all $q \in B$, the conditional probability of the event $W$ given that the state $q$ is reached is $1$: $\mathbb{P}[W \mid q] = 1$. It follows that $\mathbb{P}[W \mid B] = 1$. Overall, we have that, for all BSCCs $B \in \mathcal{B}$, $\mathbb{P}[W \mid B] \geq v_B$. 
	
	We can now conclude, since the objective $W$ is prefix-independent and by~(\ref{eqn:sum_reach_bscc_one}):
	\begin{align*}
		\mathbb{P}[W] & = \sum_{B \in \mathsf{BSCC}} \mathbb{P}[W \cap \diamondsuit B] + 
		\mathbb{P}[W \cap \lnot \diamondsuit \mathcal{B}] \\
		& = \sum_{B \in \mathsf{BSCC}} \mathbb{P}(\diamondsuit B) \cdot \mathbb{P}[W \mid \diamondsuit B] \\
		& \geq \sum_{B \in \mathsf{BSCC}} \mathbb{P}(\diamondsuit B) \cdot v_B \\
		& \geq v(q_0)
	\end{align*}
	As this holds for all states $q_0 \in Q$, we have that the Player $\A$ strategy dominates the valuation $v$ (recall that the Player $\B$ strategy is optimal in the MDP $\Gamma$), which concludes the proof.
\end{proof}
We believe that this result could also hold without the assumption that positional strategies suffice in finite MDPs with objective $Q^\omega \setminus W$. However, this would complicate the proof since the Markov chain we obtain would not be finite anymore.

Let us proceed to the proof of Lemma~\ref{lem:uniformly_optimal}.
\begin{proof}[Proof of Lemma~\ref{lem:uniformly_optimal}]
	We consider the valuation $v = \MarVal{\G}$ of the lemma.
	Let us argue that the conditions of the lemma are necessary conditions. Assume towards a contradiction that $\s_\A$ is not locally optimal. That is, there is a state $q \in Q$ such that $\va_{\gameNF{\formNF_q}{\mu_v}}(\s_\A(q)) \leq v(q) - \varepsilon$ for some $\varepsilon > 0$. Consider a Player $\B$ strategy $\s_\B'$ such that for all $q' \in Q$, the value of the strategy at state $q'$ is at most $\MarVal{\G}(q') + \varepsilon/2$: $\MarVal{\G}[\s_\B'](q') \leq \MarVal{\G}(q') + \varepsilon/2$. We then define a Player $\B$ strategy $\s_\B$ as follows: $\s_\B(q)$ is a $\GF$-strategy in the game form $\formNF_q$ such that $\outM_{\gameNF{\formNF_q}{\mu_v}}(\s_\A(q),\s_\B(q)) \leq v(q) - \varepsilon$. Furthermore, for all $q' \in Q$ and $\pi \in Q^*$, we set $\s_\B(q \cdot \pi) := \s_\B'(\pi)$. Then, as $W$ is prefix-independent and by  Proposition~\ref{prop:outcome_valuation}, we obtain:
	\begin{align*}
		\prob{\Aconc,q}{\s_\A,\s_\B}[\rho \in Q^\omega \mid \rho \in W] & = \prob{\Aconc,q}{\s_\A,\s_\B}[q \cdot \rho \in Q^\omega \mid \rho \in W] \\
		& = \sum_{q' \in Q} \prob{\Aconc,q}{\s_\A,\s_\B}[q \cdot q' \cdot \rho \in Q^\omega \mid \rho \in W] \\
		& = \sum_{q' \in Q} \prob{\Aconc,q}{\s_\A,\s_\B}[q \cdot q'] \cdot \prob{\Aconc,q}{\s_\A,\s_\B'}[q' \cdot \rho \in Q^\omega \mid \rho \in W] \\
		& \leq \sum_{q' \in Q} \mathbb{P}_{\s_\A,\s_\B}(q)[q'] \cdot (v(q') + \varepsilon/2) \\
		& = \sum_{q' \in Q} \mathbb{P}_{\s_\A,\s_\B}(q)[q'] \cdot v(q') + \varepsilon/2 \\
		& = \outM_{\langle \formNF_q,\mu_v \rangle}(\s_\A(q),\s_\B(q)) + \varepsilon/2\\
		& \leq v(q) - \varepsilon + \varepsilon/2 = v(q) - \varepsilon/2 < v(q)
	\end{align*}
	That is, the strategy $\s_\A$ is not optimal from $q$.
	
	Assume now, still towards a contradiction, that there is an EC $H$ in the MDP induced by the strategy $\s_\A$ such that $v_H > 0$ and there is some $q \in Q_H$ such $\MarVal{\Aconc_H^{\s_\A}}(q) < 1$. Then, by Theorem~\ref{thm:chaterjee_value_0_1}, there is a state $q' \in Q_H$ such that $\MarVal{\Aconc_H^{\s_\A}}(q) = 0$. It follows that there is a Player $\B$ strategy $\s_\B$ in the game $\Aconc_H^{\s_\A}$ such that the value from the state $q$ with the strategy $\s_\B$ is at most $v_H/2 < v_H$: $\MarVal{\Aconc_H^{\s_\A}}[\s_\B](q) \leq v_H/2$. Note that the strategy $\s_\B$ can be seen as a strategy in the game $\G$. Hence, the value of the game $\G$ of state $q$ with strategies $\s_\A$ and $\s_\B$ is at most $v_H/2 < v_H = v(q)$. That is, the strategy $\s_\A$ is not optimal from $q$.
	
	We then apply Lemma~\ref{lem:uniform_guarantee} to conclude that these conditions are also sufficient for uniform optimality.
\end{proof}

\section{Proofs from Section~\ref{sec:optimal_in_buchi}}
\subsection{Construction of a reachability game from a B\"uchi game}
\label{app:buchi2reach}

	Consider a Büchi game $\G = \Games{\Aconc}{\Bu(T)}$ for some $T \subseteq Q$. We build the reachability game $\G_{\mathsf{reach}} := \Games{\Aconc_{\mathsf{reach}}}{\Reach(\{ \top \})}$ with $\Aconc_{\mathsf{reach}} := \langle Q \cup \{ \top,\bot \},A,B,\distribSet \cup \{ d_{v} \mid v \in \mathsf{Val} \} \cup \{ d_0,d_1 \},\delta',\distribFunc' \rangle$ with $\mathsf{Val} := \{ \MarVal{\G}(q) \mid q \in T \} \cup \{ 0,1 \}$ such that:
	\begin{itemize}
		\item for all $q \in Q \setminus T$, $\delta'(q,\cdot,\cdot) := \delta(q,\cdot,\cdot)$;
		\item for all $q \in T$, $a \in A$ and $b \in B$, we have $\delta'(q,a,b) := d_{\MarVal{\G}(q)}$, $\delta'(\top,a,b) := d_{1}$ and $\delta'(\bot,a,b) := d_{0}$;
		\item for all $d \in \distribSet$, $\distribFunc'(d) := \distribFunc(d)$;
		\item for all $u \in \mathsf{Val}$, $\distribFunc'(d_u) := \{ \top \mapsto u, \bot \mapsto 1 - u \}$.
	\end{itemize}

%
\subsection{Proof of a first proposition}
\label{appen:proof_prop_buchi_equal_reach}
Recall the definition of Appendix~\ref{app:buchi2reach}, it ensures the following proposition.
\begin{proposition}
  For every $q \in Q$,
  $\MarVal{\G}(q) = \MarVal{\G_\mathsf{reach}}(q)$. Moreover, if there
  is an optimal strategy in $\G$ from some $q \in Q$, then there is
  also an optimal strategy in $\G_\mathsf{reach}$ from
  $q$. 
  \label{prop:buchi_equal_reach}
\end{proposition}

In fact, for arbitrary prefix-independent objectives, such a construction preserves the values of states if it keeps the same objectives (instead of changing it from Büchi to reachability in this case). 
We consider this more general construction for Büchi and co-Büchi objectives here, and then deduce Proposition~\ref{prop:buchi_equal_reach}.
\begin{definition}
	Consider a game $\G = \Games{\Aconc}{W}$ with an the objective $W$ being either $\Bu(T)$ ot $\coBu(T)$ for some $T \subseteq Q$. Let $S \subseteq Q$. We build the game $\G_S := \Games{\Aconc_S}{W'}$ (where $W'$ will be specified later) with $\Aconc_S := \langle Q \cup \{ \top,\bot \},A,B,\distribSet \cup \{ d_{v} \mid v \in \mathsf{Val} \} \cup \{ d_0,d_1 \},\delta',\distribFunc' \rangle$ with $\mathsf{Val} := \{ \MarVal{\G}(q) \mid q \in S \} \cup \{ 0,1 \}$ such that:
	\begin{itemize}
		\item for all $q \in Q \setminus S$, $\delta'(q,\cdot,\cdot) := \delta(q,\cdot,\cdot)$;
		\item for all $q \in S$, $a \in A$ and $b \in B$, we have $\delta'(q,a,b) := d_{\MarVal{\G}(q)}$, $\delta'(\top,a,b) := d_{1}$ and $\delta'(\bot,a,b) := d_{0}$;
		\item for all $d \in \distribSet$, $\distribFunc'(d) := \distribFunc(d)$;
		\item for all $u \in \mathsf{Val}$, $\distribFunc'(d_u) := \{ \top \mapsto u, \bot \mapsto 1 - u \}$.
	\end{itemize}
	Now, if $W = \Bu(T)$, then $W' := \Bu(T \cup \{ \top \})$ and if $W = \coBu(T)$ then $W' := \coBu(T \cup \{ \bot \})$. (That is, in any case, the state $\top$ has value 1 and the state $\bot$ has value 0).
	\label{def:game_extracted}
\end{definition}
We obtain a proposition analogous to Proposition~\ref{prop:buchi_equal_reach}.
\begin{proposition}
	Consider a game $\G = \Games{\Aconc}{W}$ with $W$ either equal to $\Bu(T)$ or to $\coBu(T)$ for some subset of states $T \subseteq Q$ and the game $\G_S = \Games{\Aconc_S}{W'}$ built from it, for some $S \subseteq Q$. Then, for all states $q \in Q$, the value of both games from state $q$ is the same: $\MarVal{\G}(q) = \MarVal{\G_S}(q)$. Furthermore, any Player $\A$ strategy that is optimal from some state in $\G$ can be translated into an optimal strategy in $\G_S$.
	\label{prop:W_equal_W}
\end{proposition}
\begin{proof}	
	We have $\MarVal{\G_S}(\top) = 1$ and $\MarVal{\G_S}(\bot) = 0$, therefore for all states $q \in S$, we have $\MarVal{\G}(q) = \MarVal{\G_S}(q)$ by construction. 
	
	Consider now a state $q \in Q \setminus S$ and a Player $\A$ strategy $\s_\A$ in the game $\G$. Let us consider what happens in the game $\G_S$ if Player $\A$ plays the strategy $\s_\A$ from $q$ (it plays arbitrary -- and irrelevant -- actions in $\top$ and $\bot$). Let $\s_\B$ be a Player $\B$ strategy in $\G_S$. With these choices, what happens with $\s_\A,\s_\B$ is the same in both games $\G$ and $\G_S$ for paths in $(Q \setminus S)^+$. In the following, for a set of states $X \subseteq Q$, we denote by $\diamondsuit X$ the event 'reaching the set $X$' and by $\diamondsuit^{\mathsf{st}} (x,X)$ the event 'reaching the set $X$ and the first element seen in $X$ is $x$'. Furthermore, note that if $S$ is not reached, neither is $\top$ and $\bot$. Hence, the events $W' \cap \lnot \diamondsuit S$ and $W \cap \lnot \diamondsuit S$ have the same probability measure. Therefore, we obtain:
	\begin{align*}
	\prob{\Aconc_S,q}{\s_\A,\s_\B}[W'] & = \prob{\Aconc_S,q}{\s_\A,\s_\B}[W' \cap \diamondsuit S] + \prob{\Aconc_S,q}{\s_\A,\s_\B}[W' \cap \lnot \diamondsuit S] \\
	& = \sum_{s \in S} \prob{\Aconc_S,q}{\s_\A,\s_\B}[W' \cap \diamondsuit^{\mathsf{st}} (s,S)] + \prob{\Aconc_S,q}{\s_\A,\s_\B}[W \cap \lnot \diamondsuit S]\\
	& = \sum_{s \in S} \prob{\Aconc_S,q}{\s_\A,\s_\B}[W' \mid \diamondsuit^{\mathsf{st}} (s,S)] \cdot \prob{\Aconc_S,q}{\s_\A,\s_\B}[\diamondsuit^{\mathsf{st}} (s,S)] + \prob{\Aconc,q}{\s_\A,\s_\B}[W \cap \lnot \diamondsuit S]\\
	& = \sum_{s \in S} \MarVal{\G_S}(s) \cdot \prob{\Aconc_S,q}{\s_\A,\s_\B}[\diamondsuit^{\mathsf{st}} (s,S)]+ \prob{\Aconc,q}{\s_\A,\s_\B}[W \cap \lnot \diamondsuit S]\\
	& = \sum_{s \in S} \MarVal{\G}(s) \cdot \prob{\Aconc,q}{\s_\A,\s_\B}[\diamondsuit^{\mathsf{st}} (s,S)]+ \prob{\Aconc,q}{\s_\A,\s_\B}[W \cap \lnot \diamondsuit S]\\
	\end{align*}
	Overall, we obtain:
	\begin{equation}
	\prob{\Aconc_S,q}{\s_\A,\s_\B}[W'] = \sum_{s \in S} \MarVal{\G}(s) \cdot \prob{\Aconc,q}{\s_\A,\s_\B}[\diamondsuit^{\mathsf{st}} (s,S)]+ \prob{\Aconc,q}{\s_\A,\s_\B}[W \cap \lnot \diamondsuit S]
	\label{eqn:expression_proba_reach}
	\end{equation}
	
	Consider in addition some $\varepsilon > 0$ and a Player $\B$ strategy $\s_\B^\varepsilon$ equal to the strategy $\s_\B$ on $(Q \setminus S)^+$ and such that, for all states $s \in S$, the value $\MarVal{\G}[\s_\B](s)$ of the game $\G$ from $s$ with the strategy $\s_\B^\varepsilon$ is at most $\MarVal{\G}(s) + \varepsilon$. With that strategy, we obtain:
	\begin{align*}
	\prob{\Aconc,q}{\s_\A,\s_\B^\varepsilon}[W] & = \prob{\Aconc,q}{\s_\A,\s_\B^\varepsilon}[W \cap \diamondsuit S] + \prob{\Aconc,q}{\s_\A,\s_\B^\varepsilon}[W \cap \lnot \diamondsuit S] \\
	& = \sum_{s \in S} \prob{\Aconc,q}{\s_\A,\s_\B^\varepsilon}[W \cap \diamondsuit^{\mathsf{st}} (s,S)] + \prob{\Aconc,q}{\s_\A,\s_\B^\varepsilon}[W \cap \lnot \diamondsuit S] \\
	& = \sum_{s \in S} \prob{\Aconc,q}{\s_\A,\s_\B^\varepsilon}[W \mid \diamondsuit^{\mathsf{st}} (s,S)] \cdot \prob{\Aconc,q}{\s_\A,\s_\B^\varepsilon}[\diamondsuit^{\mathsf{st}} (s,S)] + \prob{\Aconc,q}{\s_\A,\s_\B^\varepsilon}[W \cap \lnot \diamondsuit S] \\
	& \leq \sum_{s \in S} (\MarVal{\G}(s) + \varepsilon) \cdot \prob{\Aconc,q}{\s_\A,\s_\B}[\diamondsuit^{\mathsf{st}} (s,S)] + \prob{\Aconc,q}{\s_\A,\s_\B^\varepsilon}[W \cap \lnot \diamondsuit S]\\
	& \leq \sum_{s \in S} \MarVal{\G}(s) \cdot \prob{\Aconc,q}{\s_\A,\s_\B}[\diamondsuit^{\mathsf{st}} (s,S)] + \prob{\Aconc,q}{\s_\A,\s_\B^\varepsilon}[W \cap \lnot \diamondsuit S] + \varepsilon \\
	& = \prob{\Aconc_S,q}{\s_\A,\s_\B}[W'] + \varepsilon
	\end{align*}
	As this holds for all $\varepsilon > 0$, it follows that $\MarVal{\G}[\s_\A](q) \leq \prob{\Aconc_S,q}{\s_\A,\s_\B}[W']$. As this holds for all Player $\B$ strategy $\s_\B$, it follows that $\MarVal{\G}[\s_\A](q) \leq \MarVal{\G_S}[\s_\A](q) \leq \MarVal{\G_S}(q)$. 
	It follows that $\MarVal{\G}(q) \leq \MarVal{\G_S}(q)$ and (as we will show that the values in both games are the same) a Player $\A$ optimal strategy in $\G$ can be translated into a Player $\A$ optimal strategy in $\G_S$.
	
	We then proceed similarly to show that $\MarVal{\G_S}(q) \leq \MarVal{\G}(q)$. 
	Consider a Player $\A$ strategy $\s_\A$ from $q$ in $\G_S$. For $\varepsilon > 0$, let us define a Player $\A$ strategy $\s_\A^\varepsilon$ in the game $\G$ in the following way: $\s_\A^\varepsilon$ coincides with $\s_\A$ on $(Q \setminus S)^+$ and then plays $\varepsilon$-optimal strategies from all states in $S$, i.e. the value of the game $\G$ from $s$ is at least $\MarVal{\G}(s) - \varepsilon$ with strategy $\s_\A^\varepsilon$. Consider a Player $\B$ strategy $\s_\B$ in $\G$ (that can be seen as a strategy in $\G_S$), we have:
	\begin{align*}
	\prob{\Aconc,q}{\s_\A^\varepsilon,\s_\B}[W] & = \prob{\Aconc,q}{\s_\A^\varepsilon,\s_\B}[W \cap \diamondsuit S] + \prob{\Aconc,q}{\s_\A^\varepsilon,\s_\B}[W \cap \lnot \diamondsuit S] \\
	& = \sum_{s \in S} \prob{\Aconc,q}{\s_\A^\varepsilon,\s_\B}[W \cap \diamondsuit^{\mathsf{st}} (s,S)] + \prob{\Aconc,q}{\s_\A^\varepsilon,\s_\B}[W \cap \lnot \diamondsuit S] \\
	& = \sum_{s \in S} \prob{\Aconc,q}{\s_\A^\varepsilon,\s_\B}[W \mid \diamondsuit^{\mathsf{st}} (s,S)] \cdot \prob{\Aconc,q}{\s_\A^\varepsilon,\s_\B}[\diamondsuit^{\mathsf{st}} (s,S)] + \prob{\Aconc,q}{\s_\A^\varepsilon,\s_\B}[W \cap \lnot \diamondsuit S] \\
	& \geq \sum_{s \in S} (\MarVal{\G}(s) - \varepsilon) \cdot \prob{\Aconc,q}{\s_\A,\s_\B}[\diamondsuit^{\mathsf{st}} (s,S)] + \prob{\Aconc,q}{\s_\A^\varepsilon,\s_\B}[W \cap \lnot \diamondsuit S] \\
	& \geq \sum_{s \in S} \MarVal{\G}(s) \cdot \prob{\Aconc,q}{\s_\A,\s_\B}[\diamondsuit^{\mathsf{st}} (s,S)] + \prob{\Aconc,q}{\s_\A^\varepsilon,\s_\B}[W \cap \lnot \diamondsuit S] - \varepsilon\\
	& = \prob{\Aconc_S,q}{\s_\A,\s_\B}[W'] - \varepsilon
	\end{align*}
	This last equality comes from Equation~(\ref{eqn:expression_proba_reach}). As this holds for all Player $\B$ strategies $\s_\B$, we have $\MarVal{\G_S}[\s_\A](q) - \varepsilon \leq \MarVal{\G}[\s_\A^\varepsilon](q)
	\leq \MarVal{\G}[\A](q)$. As this holds for all $\varepsilon > 0$, we obtain $\MarVal{\G_S}[\s_\A](q) \leq \MarVal{\G}(q)$. And as this holds for all Player $\A$ strategy $\s_\A$, we have: $\MarVal{\G_S}(q) \leq \MarVal{\G}(q)$. 
	
	Overall, we have $\MarVal{\G}(q) = \MarVal{\G_S}(q)$ for all $q \in Q$.
\end{proof}

We can now proceed to the proof of Proposition~\ref{prop:buchi_equal_reach}.
\begin{proof}[Proof of Proposition~\ref{prop:buchi_equal_reach}]
	In the case where $S = T$ and $W = \Bu(T)$, we have $W' = \Bu(T \cup \{ \top \})$ but $\top$ is not reached if $T$ is not. Hence, the events $W'$ and $\Reach(\top)$ have the same probability measure. Furthermore, we have $\Aconc_T = \Aconc_\mathsf{reach}$. It follows that, in the games $\Games{\Aconc_T}{\Bu(T \cup \{ \top \})}$ and $\G^\mathsf{reach} = \Games{\Aconc_\mathsf{reach}}{\Reach(\{ \top \})}$, all states have the same values. We conclude by applying Proposition~\ref{prop:W_equal_W}.
\end{proof}

\subsection{Proof of Proposition~\ref{prop:positional_suffice_buchi}}
\label{appen:proof_prop_positional_suffice_buchi}
Let us show the following proposition:
\begin{proposition}
	Consider a Büchi game $\G = \Games{\Aconc}{\Bu(T)}$ and assume that there is an optimal strategy from all states in the reachability game $\G^\mathsf{reach} = \Games{\Aconc^\mathsf{reach}}{\Reach(T)}$ obtained from it. Then, there exists a Player $\A$ positional strategy that is uniformly optimal in $\G$.
	\label{prop:uniform_reach_uniform_buchi}
\end{proposition}
\begin{proof}
	By Theorem 28 in \cite{BBSCSL22}, there exists a positional Player $\A$ strategy $\s_\A^\mathsf{reach}$ in $\G^\mathsf{reach}$ that is uniformly optimal. Note that we can apply Lemma~\ref{lem:uniformly_optimal}, since, although the reachability objective is not prefix-independent, it is in that case equivalent to a Büchi objective since the only outgoing edge of the target $\top$ is a self-loop. Hence, by Lemma~\ref{lem:uniformly_optimal}, the strategy $\s_\A^\mathsf{reach}$ is locally optimal in $\G^\mathsf{reach}$. Let us now define a Player $\A$ positional strategy $\s_\A$ in $\G$ in the following way:
	\begin{itemize}
		\item for all $q \in Q \setminus T$, we set $\s_\A(q) := \s_\A^\mathsf{reach}(q)$;
		\item for all $t \in T$, we set $\s_\A(t) \in \Opt_{\gameNF{\formNF_t}{\mu_v}}$ where $v := \MarVal{\G}$ is the valuation giving the values of the states, i.e. the strategy $\s_\A$ is locally optimal on $T$.
	\end{itemize}

	Let us show that this Player $\A$ positional strategy is uniformly optimal. We want to apply Lemma~\ref{lem:uniformly_optimal}. Since $\s_\A^\mathsf{reach}$ is uniformly optimal in $\G^\mathsf{reach}$, and the values of all states in $Q$ are the same in $\G$ and $\G^\mathsf{reach}$, it follows that the strategy $\s_\A$ is locally optimal. Furthermore, consider an EC $H = (Q_H,\beta)$ such that $v_H > 0$ in the MDP $\Gamma$ induced by the strategy $\s_\A$ in the game $\G$ and a Player $\B$ strategy $\s_\B$ in the game $\Aconc_H^{\s_\A}$. Note that the strategy $\s_\B$ is also a strategy in $\G$ and $\G^\mathsf{reach}$. Consider now a state $q \in Q_H \setminus T$. By Equation~(\ref{eqn:expression_proba_reach}) (in the proof of Proposition~\ref{prop:buchi_equal_reach}), we have $\prob{\Aconc_\mathsf{reach},q}{\s_\A^\mathsf{reach},\s_\B}[\Reach(\top)] = \sum_{t \in T} \MarVal{\G}(t) \cdot \prob{\Aconc,q}{\s_\A,\s_\B}[\diamondsuit^{\mathsf{st}} (t,T)]$. In addition, for all $q \in Q_H$, we have $\MarVal{\G}(t) = v_H$. Hence, for all $t \in T$ such that $\MarVal{\G}(t) \neq v_H$, we have $\prob{\Aconc,q}{\s_\A,\s_\B}[\diamondsuit^{\mathsf{st}} (t,T)] = 0$. Furthermore, note that for all $t \in T$, we have $\prob{\Aconc^\mathsf{reach},q}{\s_\A^\mathsf{reach},\s_\B}[\diamondsuit^{\mathsf{st}} (t,T)] = \prob{\Aconc,q}{\s_\A,\s_\B}[\diamondsuit^{\mathsf{st}} (t,T)]$ (as the strategies $\s_\A$ and $\s_\A^\mathsf{reach}$ coincide on $Q \setminus T$). Then, as the strategy $\s_\A^\mathsf{reach}$ is uniformly optimal in $\G^\mathsf{reach}$, we have:
	\begin{displaymath}
		v_H = \MarVal{\G}(q) = \MarVal{\G^\mathsf{reach}}(q) \leq \prob{\Aconc_\mathsf{reach},q}{\s_\A^\mathsf{reach},\s_\B}[\Reach(\top)] = v_H \cdot \prob{\Aconc,q}{\s_\A,\s_\B}[\diamondsuit T]
	\end{displaymath}
	That is: $\prob{\Aconc,q}{\s_\A,\s_\B}[\diamondsuit T] = 1$. As this holds for all $q \in Q_H \setminus T$, it follows that $\prob{\Aconc,q}{\s_\A,\s_\B}[\Bu(T)] = 1$. As this holds for all Player $\B$ strategy $\s_\B$, we have $\MarVal{\Aconc_H^{\s_\A}}(q) = 1$. We can then apply Lemma~\ref{lem:uniformly_optimal} to obtain that the strategy $\s_\A$ is uniformly optimal in $\G$. 
\end{proof}

We can then deduce the proof of Proposition~\ref{prop:positional_suffice_buchi}:
\begin{proof}[Proof of Proposition~\ref{prop:positional_suffice_buchi}]
	By Proposition~\ref{prop:buchi_equal_reach}, as optimal strategies in the Büchi game $\G$ are translated into optimal strategies in the reachability game $\G^\mathsf{reach}$ it follows that there are optimal strategies from all states in reachability game $\G^\mathsf{reach}$. We can then conclude by applying Proposition~\ref{prop:uniform_reach_uniform_buchi}.
\end{proof}

\subsection{Proof of Theorem~\ref{lem:rm_in_buchi}}
\label{appen:proof_lem_rm_in_buchi}
\begin{proof}	
	We want to apply Proposition~\ref{prop:uniform_reach_uniform_buchi}. To do so, we have to show that there exists an optimal strategy from all states in the reachability game $\G^\mathsf{reach}$. This is done by applying Theorem~\ref{thm:rm_in_reach}. Indeed, all local interaction at states $t \in T \cup \{ \top,\bot \}$ are trivial and therefore RM (see Proposition 39 from \cite{BBSCSL22} as trivial interactions are a special case of determined local interactions). Furthermore, as the local interactions at states in $Q \setminus T$ are not modified from $\Aconc$ to $\Aconc^\mathsf{reach}$, they are also RM (by assumption). Therefore, we can apply Theorem~\ref{thm:rm_in_reach} to obtain that there exists an optimal strategy from all states in $\G^\mathsf{reach}$. We conclude by applying Proposition~\ref{prop:uniform_reach_uniform_buchi}.
\end{proof}


\section{Values in Büchi and co-Büchi games, valuation of game forms}
We do not prove any new result in this section, however we recall how to compute the values in Büchi and co-Büchi games (with nested fixed points). We also formally define partial valuations in game forms, and the corresponding values. The formal definition of RM game forms is given in Subsubsection~\ref{subsubsec:partial_val_reach}.

\subsection{Additional Preliminaries}
For a set of states $Q$ and two valuations $v,v' \in [0,1]^Q$, we denote by $v \preceq v'$ the fact that, for all $q \in Q$ we have $v(q) \leq v'(q)$. We define $v \prec v'$ is $v \preceq v'$ and $v \neq v'$. For a subset of states $S \subseteq Q$ and two valuations $v_S \in [0,1]^S$ and $v_{\bar{S}} \in [0,1]^{Q \setminus S}$, we denote by $v_S \sqcup v_{\bar{S}} \in [0,1]^Q$ the valuation $v$ such that, for all $q \in Q$:
\begin{equation*}
v(q) := \begin{cases}
v_S(q) \quad &\text{if} \, q \in S \\
v_{\bar{S}}(q) \quad &\text{if} \, q \notin S \\
\end{cases}
\end{equation*}

An operator $\Delta: [0,1]^Q \rightarrow [0,1]^Q$ is non-decreasing if, for all $v,v' \in [0,1]^Q$, if $v \preceq v'$ then $\Delta(v) \preceq \Delta(v')$. 

We also recall Knaster-Tarski fixed-point theorem.
\begin{theorem}[Knaster-Tarski]
	In a complete lattice $(E,\preceq)$, for a non-decreasing operator $f: E \rightarrow E$, the set of fixed pints of $f$ is also a complete lattice. In particular, it has a least and a greatest fixed point. These are denoted $\lfp(f)$ and $\gfp(f)$ respectively.
\end{theorem}

Note that the least fixed-point of a non-decreasing operator $f$ also ensures the following property:
\begin{proposition}
	Consider a complete lattice $(E,\preceq)$, a non-decreasing operator $f: E \rightarrow E$ and its least fixed point $v = \lfp(f) \in E$. Then, for all $u \prec v$, we have that $f(u) \preceq u$ does not hold. 
	\label{prop:no_derase_before_lfp}
\end{proposition}
\begin{proof}
	Towards a contradiction, let $E_u := \{ u' \in E \mid u' \preceq u \} \neq \emptyset$. Then, we have $f(E_u) \subseteq E_u$ since, for all $u' \in E_u$, we have $f(u') \preceq f(u) \preceq u$ since $f$ is non-decreasing. Hence, by Knaster-Tarski Theorem, we have that $f$ has a fixed point $u' \in E_u$. Hence the contradiction since $u' \preceq u \prec v$ and $v$ is assumed the least fixed-point of $f$.
\end{proof}

Consider a non-empty set of states $Q$. Let us define an operator on valuations of states. We consider the operator $\OneStep: [0,1]^Q \rightarrow [0,1]^Q$ such that, for all valuations $v \in [0,1]^Q$ and all $q \in Q$, we have $\OneStep(v)(q) := \va_{\gameNF{\formNF_q}{\mu_v}}$. Furthermore, for all subset of states $S \subseteq Q$ and valuations $v_S \in [0,1]^S$, we have $\OneStepSet{S}{v_S}: [0,1]^{Q \setminus S} \rightarrow [0,1]^{Q \setminus S}$ such that, for all valuations $v_{\bar{S}} \in [0,1]^{Q \setminus S}$ and $q \in Q \setminus S$, we have $\OneStepSet{S}{v_S}(v_{\bar{S}})(q) := \OneStep(v_S \sqcup v_{\bar{S}})(q)$, i.e. we fix the valuation $v_S$.

These operators ensure the following proposition:
\begin{proposition}
	For every concurrent arena $\Aconc$, for every subset of states $S \subseteq Q$ and each valuation $v_S \in [0,1]^S$:
	\begin{itemize}
		\item the operator $\OneStep$ is non-decreasing;
		\item $([0,1]^{Q \setminus S},\preceq)$ is a complete lattice;
		\item the operator $\Delta$ is 1-Lipschitz
		\item $\OneStepSet{S}{v_S}$ has a least and a greatest fixed point in $[0,1]^{Q \setminus S}$. They are denoted $v_{\lf \bar{S}}[v_S] := \lfp(\OneStepSet{S}{v_S})$ and $v_{\gf \bar{S}}[v_S] := \gfp(\OneStepSet{S}{v_S})$ respectively.
	\end{itemize}
	\label{prop:properties_delta}
\end{proposition}
\begin{proof}
	See Proposition 43 in \cite{BBSCSLarXiv}.
\end{proof}

\subsection{Least fixed-point operation}
\subsubsection{Values in Reachability Game}
Consider a reachability game $\G = \Games{\Aconc}{\Reach(T)}$ for a subset of states $T \subseteq Q$. In such a game, all states in the target $T$ have value 1. Their values are given by the valuation $v_T \in [0,1]^T$ such that $v_T(t) = 1$ for all $t \in T$. The values of states in $Q \setminus T$ are then obtained via a least fixed-point operation. 

\begin{theorem}[\cite{everett57,filar2012competitive}]
	The values of the reachability game $\G = \Games{\Aconc}{\Reach(T)}$ are given by the least fixed point $v_{\lf \bar{T}}[v_T]$ of the operator $\OneStepSet{T}{v_T}$. Indeed: $\MarVal{\G} := v_T \sqcup v_{\lf \bar{T}}[T,v_T]$.
\end{theorem}

\subsubsection{Partial valuation and game forms}
\label{subsubsec:partial_val_reach}
Consider a game form $\formNF = \langle A,B,\outComeNF,\outCNF \rangle$, a partition $\outComeNF = \outComeNFLp \uplus \outComeNFEx$ and a partial valuation $\alpha: \outComeNFEx \rightarrow [0,1]$ of the outcomes. For all $u \in [0,1]$ we define the total valuation $\alpha[u] \in [0,1]^\outComeNF$ such that, for all $o \in \outComeNF$:
\begin{equation*}
\alpha[u](o) := \begin{cases}
u \quad &\text{if } o \in \outComeNFLp \\
\alpha(o) \quad &\text{if } o \in \outComeNFEx \\
\end{cases}
\end{equation*}
Then, we define the function $f_\formNF\langle \alpha\rangle: [0,1] \rightarrow [0,1]$ such that, for all $u \in [0,1]$, we have $f_\formNF\langle \alpha\rangle(u) := \va_{\gameNF{\formNF}{\alpha[u]}}$. This function ensures the properties below:
\begin{proposition}
	Consider a game form $\formNF = \langle A,B,\outComeNF,\outCNF \rangle$ a partition $\outComeNF = \outComeNFLp \uplus \outComeNFEx$ and a partial valuation $\alpha: \outComeNFEx \rightarrow [0,1]$ of the outcomes. Then:
	\begin{itemize}
		\item $([0,1],\leq)$ is a complete lattice;
		\item $f_\formNF\langle \alpha\rangle$ is non-decreasing;
		\item $f_\formNF\langle \alpha\rangle$ is 1-lipschitz;
		\item $f_\formNF\langle \alpha\rangle$ has a least and a greatest fixed-point denoted $u_{\formNF,\lf}(\alpha)$ and $u_{\formNF,\gf}(\alpha)$ respectively. We denote by $\alpha_{\formNF, \lf} := \alpha[u_{\formNF, \lf}(\alpha)]$ and $\alpha_{\formNF, \gf} := \alpha[u_{\formNF, \gf}(\alpha)]$ the corresponding valuations.
	\end{itemize}
	\label{prop:properties_f_alpha}
\end{proposition}
\begin{proof}
	This is similar to Proposition~\ref{prop:properties_delta}.
\end{proof}

Furthermore, we can link the values of states in a reachability game and the function $f_\formNF\langle \alpha\rangle$. To do so, we have to consider a partial valuation of the Nature states in a reachability game.
\begin{definition}[Partial valuation of Nature states in reachability games]
	Consider a concurrent reachability $\G = \Games{\Aconc}{\Reach(T)}$, a value $u \in [0,1]$ such that $u > 0$ and a non-empty subset of states $\emptyset \neq Q_u \subseteq Q$ such that, for all $q \in Q_u$, we have $\MarVal{\G}(q) = u$. Let $\distribSet_{\mathsf{Lp}}^u := \{ d \in \distribSet \mid \Supp(d) \subseteq Q_u \}$. Then, we define the partition of the Nature states $\distribSet = \distribSet_{\mathsf{Lp}}^{Q_u} \uplus \distribSet_{\mathsf{Ex}}^{Q_u}$ with 
	$\distribSet_{\mathsf{Ex}}^{Q_u} := \distribSet \setminus \distribSet_{\mathsf{Lp}}^{Q_u}$. We then define the partial valuation of the Nature states $\alpha^{Q_u}: \distribSet_{\mathsf{Ex}}^{Q_u} \rightarrow [0,1]$ such that, for all $d \in \distribSet_{\mathsf{Ex}}^{Q_u}$, we have $\alpha^{Q_u}(d) = \mu_v(d)$ for $v := \MarVal{\G}$ the valuation giving the values of the states. 
	\label{def:partial_val_in_nature_states}
\end{definition}

With this definition, we recall the lemma below, which corresponds to Proposition 57 in \cite{BBSCSLarXiv}. We will use it both when dealing with the Büchi and co-Büchi objectives. 
\begin{lemma}[Proposition 57 in \cite{BBSCSLarXiv}]
	Consider a concurrent reachability game $\G = \Games{\Aconc}{\Reach(T)}$, a value $u \in [0,1]$ such that $u > 0$ and a non-empty subset of states $\emptyset \neq Q_u \subseteq Q$ such that, for all $q \in Q_u$, we have $\MarVal{\G}(q) = u$. We set $Q_u^\diamondsuit := \{ q  \in Q_u \mid u_{\formNF_q,\diamondsuit}(\alpha^{Q_u}) = u \}$ the set of states in $Q_u$ whose values in the reachability game when only considering its immediate successors is at least $u$.	
	Then, we have $Q_u^\diamondsuit \neq \emptyset$.
	\label{lem:value_in_RM_up_to_u}
\end{lemma}

We can now define what are reach-maximizable game forms.
\begin{definition}[Reach-maximizable game forms, Definitions 32, 34 from \cite{BBSCSL22}]
	Consider a game form $\formNF$. For a partition of the outcomes $\outComeNF = \outComeNFEx \uplus \outComeNFLp$ and a partial valuation $\alpha: \outComeNFEx \rightarrow [0,1]$, we say that the game form $\formNF$ is \emph{reach-maximizable} w.r.t. to valuation $\alpha$ if either $u_{\formNF \lf}(\alpha) = 0$ or there exists an optimal strategy $\sigma_\A \in \Opt_{\gameNF{\formNF}{\alpha_{\formNF, \lf}}}$ such that, for all $b \in B$, there exists some action $a \in A$ in the support of $\sigma_\A$ -- i.e. $a \in \Supp(\sigma_\A)$ -- such that $\outCNF(a,b) \in \outComeNFEx$. Such strategies are said to be \emph{reach-maximizing}.
	
	Then, the game form $\formNF$ is \emph{reach-maximizable} if it is reach-maximizable w.r.t. all partial valuations $\alpha$.
\end{definition}

\subsection{Values in Büchi Games and co-Büchi Games}
In this subsection, we express the values of Büchi and co-Büchi games with fixed-point operations. These characterizations come from \cite{AM04b} and were proven using $\mu$-calculus. Note that these results are central to our study of the game forms of interest since the properties we express on game forms come from the way the values of Büchi and co-Büchi games are obtained via fixed-point operations.

\subsubsection{Values in a Büchi Game}
Consider a Büchi game $\G = \Games{\Aconc}{\Bu(T)}$ for a subset of states $T \subseteq Q$. Contrary to the reachability objective, the values of the states in the target $T$ is not necessarily 1. Now the values of the game are obtained via a nested fixed-point operation. Specifically, we consider the operator $\OneStepLFP{T}: [0,1]^{T} \rightarrow [0,1]^T$ such that, for all $v_T \in [0,1]^T$ and $q \in T$, we have:
$\OneStepLFP{T}(v_T)(q) := \OneStep(v_T \sqcup v_{\lf \bar{T}}[v_T])(q)$. That is, we fix $v_T$ the valuation of states in $T$, then we consider the least fixed point of the operator $\OneStepSet{T}{v_T}$, which gives the valuation $v_{\lf \bar{T}}[v_T]$ of the states in $Q \setminus T$. Note that, if the values of the states in $T$ are fixed, the game we obtain is a reachability game (hence its value is computed with a least fixed-point). Finally, we apply the operator $\OneStep$ on the obtained valuation $v_T \sqcup v_{\lf \bar{T}}[T,v_T]$. In fact, the operator $\OneStepLFP{T}$ ensures the proposition below.
\begin{proposition}
	Consider an arena $\Aconc$ and a set of states $T \subseteq Q$. Then, the operator $\OneStepLFP{T}: [0,1]^{T} \rightarrow [0,1]^T$ is non-decreasing and its set of fixed-points is a complete lattice. 
	\label{prop:property_delta_lfp}
\end{proposition}
\begin{proof}
	This is a direct consequence of the fact that the function $\OneStep$ is non-decreasing.
\end{proof}
The values of states in $T$ are then given by the greatest fixed-point of this operator, denoted $v_{\gf T \lf \bar{T}} := \gfp(\OneStepLFP{T})$, as stated in the theorem below.

\begin{theorem}[Theorem 2 in \cite{AM04b}]
	Consider a Büchi game $\G = \Games{\Aconc}{\Bu(T)}$. The values of the states in $T$ are given by the valuation $v_T := v_{\gf T \lf \bar{T}} = \gfp(\OneStepLFP{T})$. The values of the states in $Q \setminus T$ are given by the valuation $v_{\bar{T}} := v_{\lf \bar{T}}[v_T]$.
	\label{thm:value_buchi}
\end{theorem}

\subsubsection{Values in a co-Büchi game}
Consider now a co-Büchi game $\G = \Games{\Aconc}{\coBu(T)}$ for a subset of states $T \subseteq Q$. This objective is the dual of the Büchi objective. In fact, the values in this game are obtained via a least fixed-point of a function defined by a greatest fixed-point (the converse of a Büchi game: the greatest fixed-point of a function defined by a least fixed-point). Specifically, we consider the operator $\OneStepGFP{T}: [0,1]^{T} \rightarrow [0,1]^T$ such that, for all $v_T \in [0,1]^T$ and $q \in T$, we have:
$\OneStepGFP{T}(v_T)(q) := \OneStep(v_T \sqcup v_{\gf \bar{T}}[v_T])(q)$. That is, we fix $v_T$ the valuation of states in $T$, then we consider the greatest fixed point of the operator $\OneStepSet{T}{v_T}$, which gives the valuation $v_{\gf \bar{T}}[v_T]$ of the states in $Q \setminus T$. Finally, we apply the operator $\OneStep$ on the obtained valuation $v_T \sqcup v_{\gf \bar{T}}[v_T]$. The operator $\OneStepGFP{T}$ ensures the proposition below.
\begin{proposition}
	Consider an arena $\Aconc$ and a set of states $T \subseteq Q$. Then, the operator $\OneStepGFP{T}: [0,1]^{T} \rightarrow [0,1]^T$ is non-decreasing and its set of fixed-point is a complete lattice.
	\label{prop:proeprties_delta_gfp} 
\end{proposition}
\begin{proof}
	This is similar to Proposition~\ref{prop:property_delta_lfp}.
\end{proof}
The values of states in $T$ are then given by the least fixed-point of this operator, denoted $v_{\lf T \gf \bar{T}} := \lfp(\OneStepGFP{T})$, as stated in the theorem below.

\begin{theorem}[Theorem 2 in \cite{AM04b}]
	Consider a co-Büchi game $\G = \Games{\Aconc}{\coBu(T)}$. The values of the states in $T$ are given by the valuation $v_T := v_{\lf T \gf \bar{T}} = \lfp(\OneStepGFP{T})$. The values of the states in $Q \setminus T$ are given by the valuation $v_{\bar{T}} := v_{\gf \bar{T}}[v_T]$.
\end{theorem}

\subsection{Partial Valuation and Game Forms}
In this subsection, we define more involved (compared to Subsubsection~\ref{subsubsec:partial_val_reach}) partial valuations of outcomes of game forms to deal with the Büchi and co-Büchi objectives. 

\subsubsection{Game forms in Büchi games}
\label{subsubsec:gf_in_buchi}
Let us first consider the case of the Büchi objective. We now have to consider partial valuations with the information of the probability to reach $T$ and the probability to reach $\lnot T$. Specifically, let us consider a game form $\formNF = \langle A,B,\outComeNF,\outCNF \rangle$. Then, we consider a partition of the outcomes $\outComeNF = \outComeNFEx \uplus \outComeNFLp$ and a partial valuation $\alpha: \outComeNFEx \rightarrow [0,1]$ of the outcomes along with a probability function $p_T: \outComeNFLp \rightarrow [0,1]$. Then, extending $\alpha,p_T$ into a total valuations requires two values in $[0,1]$: one for $T$ and one for $\lnot T$. Hence, given $u_T,u_{\bar{T}} \in [0,1]$, we define the total valuation $\alpha[p_T(u_T \sqcup u_{\bar{T}})] \in [0,1]^\outComeNF$ such that, for all $o \in \outComeNF$:
\begin{equation*}
\alpha[p_T(u_T \sqcup u_{\bar{T}})](o) := \begin{cases}
p_T(o) \cdot u_T + (1 - p_T(o)) \cdot u_{\bar{T}} \quad &\text{if } o \in \outComeNFLp \\
\alpha(o) \quad &\text{if } o \in \outComeNFEx \\
\end{cases}
\end{equation*}
Then, we can fix $u_T$ and consider the function that given the value $u_{\bar{T}}$, outputs the outcome of the game form with the corresponding valuation. That is, given $\alpha: \outComeNFEx \rightarrow [0,1]$, $p_T: \outComeNFLp \rightarrow [0,1]$ and $u_T \in [0,1]$, we define the function $f_{\formNF,\bar{T}}\langle \alpha[p_T(u_T)]\rangle: [0,1] \rightarrow [0,1]$ such that, for all $u_{\bar{T}} \in [0,1]$ we have $f_{\formNF,\bar{T}}\langle \alpha[p_T(u_T)] \rangle(u) := \va_{\gameNF{\formNF}{\alpha[p_T(u_T \sqcup u)]}}$. This function ensures the following properties:
\begin{proposition}
	Consider a game form $\formNF = \langle A,B,\outComeNF,\outCNF \rangle$, a partition of the outcomes $\outComeNF = \outComeNFEx \uplus \outComeNFLp$ and a partial valuation $\alpha: \outComeNFEx \rightarrow [0,1]$, a probability function $p_T: \outComeNFLp \rightarrow [0,1]$ and a value $u_T \in [0,1]$. Then:
	\begin{itemize}
		\item $f_{\formNF,\bar{T}}\langle \alpha[p_T(u_T)]\rangle$ is non-decreasing;
		\item $f_{\formNF,\bar{T}}\langle \alpha[p_T(u_T)]\rangle$ has a least 
		fixed-point denoted $u_{\formNF,\lf \bar{T}}(\alpha[p_T(u_T)])$
		. We denote by $(\alpha[p_T(u_T)])_{\formNF,\lf \bar{T}} := \alpha[p_T(u_T \sqcup u_{\bar{T}})]$ 
		the corresponding valuation for $u_{\bar{T}} = u_{\formNF,\lf \bar{T}}(\alpha[p_T(u_T)])$.
	\end{itemize}
\end{proposition}
\begin{proof}
	This is similar to Proposition~\ref{prop:properties_f_alpha}.
\end{proof}


Let us now define the function $f_{\formNF,\lf \bar{T}}\langle \alpha[p_T]\rangle: [0,1] \rightarrow [0,1]$ such that, for all $u_T \in [0,1]$, we have $f_{\formNF,\lf \bar{T}}\langle \alpha[p_T]\rangle(u_T) := \va_{\gameNF{\formNF}{(\alpha[p_T(u_T)])_{\formNF,\lf \bar{T}}}}$. This function ensures the properties below:
\begin{proposition}
	Consider a game form $\formNF = \langle A,B,\outComeNF,\outCNF \rangle$, a partition of the outcomes $\outComeNF = \outComeNFEx \uplus \outComeNFLp$, a partial valuation $\alpha: \outComeNFEx \rightarrow [0,1]$ and a probability function $p_T: \outComeNFLp \rightarrow [0,1]$. Then:
	\begin{itemize}
		\item $f_{\formNF,\lf \bar{T}}\langle \alpha[p_T]\rangle$ is non-decreasing;
		\item $f_{\formNF,\lf \bar{T}}\langle \alpha[p_T]\rangle$ is 1-Lipschitz;
		\item $f_{\formNF,\lf \bar{T}}\langle \alpha[p_T]\rangle$ has a greatest fixed-point denoted $u_{\formNF,\gf T \lf \bar{T}}(\alpha[p_T])$. If we let $u_T := u_{\formNF,\gf T \lf \bar{T}}(\alpha[p_T])$, we denote by $(\alpha[p_T])_{\formNF,\gf T \lf \bar{T}} := \alpha[p_T(u_T \sqcup u_{\formNF,\lf \bar{T}}(\alpha[p_T(u_T)]))]$ the corresponding valuation.
	\end{itemize}
\end{proposition}
\begin{proof}
	This is similar to Proposition~\ref{prop:property_delta_lfp}.
\end{proof}
The value $u_{\formNF,\gf T \lf \bar{T}}(\alpha[p_T])$ corresponds to the value $u^{\Bu(T)}_{\formNF,\alpha,p_T}$ in the main part of the paper.

\subsubsection{Game forms in co-Büchi games}
\label{subsubsec:gf_in_cobuchi}
In the case of the co-Büchi objective, we in fact consider a slight generalization of what was presented in Section~\ref{sec:co-Buchi}. Indeed, in that case, 
we do not consider a partial valuation of the outcomes 
but rather a total valuation $\alpha: \outComeNF \rightarrow [0,1]$ along with a function $p_\alpha: \outComeNF \rightarrow [0,1]$. Specifically, let us consider a game form $\formNF = \langle A,B,\outComeNF,\outCNF \rangle$. Then, we consider a valuation $\alpha: \outComeNF \rightarrow [0,1]$ of the outcomes, along with two probability functions $p_\alpha,p_T: \outComeNF \rightarrow [0,1]$. Then, as in the Büchi case, extending $\alpha[p_\alpha,p_T]$ into a total valuations requires two values in $[0,1]$: one for $T$ and one for $\lnot T$. Hence, given $u_T,u_{\bar{T}} \in [0,1]$, we define the total valuation $\alpha[p_\alpha,p_T(u_T \sqcup u_{\bar{T}})] \in [0,1]^\outComeNF$ such that, for all $o \in \outComeNF$:
\begin{displaymath}
	\alpha[p_\alpha,p_T(u_T \sqcup u_{\bar{T}})](o) := p_\alpha(o) \cdot \alpha(o) + (1 - p_\alpha(o)) \cdot (p_T(o) \cdot u_T + (1 - p_T(o)) \cdot u_{\bar{T}})
\end{displaymath}
Note that in the case where $p_\alpha$ is equal to either $0$ or $1$ on every outcome, we retrieve the definition with $\alpha,p_T$ presented in Section~\ref{sec:co-Buchi} for $\outComeNFLp := \{ o \in \outComeNF \mid p_\alpha(o) = 0 \}$.

We can now fix $u_T$ and consider the function that given the value $u_{\bar{T}}$, outputs the outcome of the game form with the corresponding valuation. That is, we define the function $f_{\formNF,\bar{T}}\langle \alpha[p_\alpha,p_T(u_T)] \rangle: [0,1] \rightarrow [0,1]$ such that, for all $u \in [0,1]$ we have $f_{\formNF,\bar{T}}\langle \alpha[p_\alpha,p_T(u_T)] \rangle(u) := \va_{\gameNF{\formNF}{\alpha[p_\alpha,p_T(u_T \sqcup u)]}}$. This function ensures the following properties:
\begin{proposition}
	Consider a game form $\formNF = \langle A,B,\outComeNF,\outCNF \rangle$, a valuation $\alpha: \outComeNF \rightarrow [0,1]$, two probability functions $p_\alpha,p_T: \outComeNF \rightarrow [0,1]$ and a value $u_T \in [0,1]$. Then:
	\begin{itemize}
		\item $f_{\formNF,\bar{T}}\langle \alpha[p_\alpha,p_T(u_T)] \rangle$ is non-decreasing;
		\item $f_{\formNF,\bar{T}}\langle \alpha[p_\alpha,p_T(u_T)] \rangle$ has a greatest 
		fixed-point denoted $u_{\formNF,\gf \bar{T}}(\alpha[p_\alpha,p_T(u_T)])$. 
	\end{itemize}
\end{proposition}
\begin{proof}
	This is similar to Proposition~\ref{prop:properties_f_alpha}.
\end{proof}

As when expressing the values of a co-Büchi game with fixed-points, we fix $u_T$ and we consider a greatest fixed-point on $u_{\bar{T}}$. Specifically, consider a game form $\formNF$. Given $\alpha: \outComeNF \rightarrow [0,1]$, $p_\alpha,p_T: \outComeNF \rightarrow [0,1]$, we define the function $f_{\formNF,\gf \bar{T}}\langle \alpha[p_\alpha,p_T] \rangle: [0,1] \rightarrow [0,1]$ such that, for all $u_T \in [0,1]$, we have $f_{\formNF,\gf \bar{T}}\langle \alpha[p_\alpha,p_T] \rangle(u_T) := 
u_{\formNF,\gf \bar{T}}(\alpha[p_\alpha,p_T(u_T)])$. This function ensures the properties below: This function ensures the properties below:
\begin{proposition}
	Consider a game form $\formNF = \langle A,B,\outComeNF,\outCNF \rangle$, a valuation $\alpha: \outComeNF \rightarrow [0,1]$ and two probability functions $p_\alpha,p_T: \outComeNF \rightarrow [0,1]$. Then:
	\begin{itemize}
		\item $f_{\formNF,\gf \bar{T}}\langle \alpha[p_\alpha,p_T] \rangle$ is non-decreasing;
		\item $f_{\formNF,\gf \bar{T}}\langle \alpha[p_\alpha,p_T] \rangle$ is 1-Lipschitz;
		\item $f_{\formNF,\gf \bar{T}}\langle \alpha[p_\alpha,p_T] \rangle$ has a least fixed-point denoted $u_{\formNF,\lf T \gf \bar{T}}(\alpha[p_\alpha,p_T])$. If we let $u := u_{\formNF,\lf T \gf \bar{T}}(\alpha[p_\alpha,p_T])$
		, we denote by $(\alpha[p_\alpha,p_T])_{\formNF,\lf T \gf \bar{T}} := \alpha[p_\alpha,p_T(u \sqcup u)]$ the corresponding valuation.
	\end{itemize}
\end{proposition}
\begin{proof}
	This is similar to Proposition~\ref{prop:proeprties_delta_gfp}.
\end{proof}
The value $u^{\coBu(T)}_{\formNF,\alpha,p_T}$ in the main part of the paper corresponds to the value $u_{\formNF,\lf T \gf \bar{T}}(\alpha[p_\alpha,p_T])$ for $p_\alpha: \outComeNF \rightarrow [0,1]$ null on $\outComeNFLp$ and equal to 1 on $\outComeNFEx$.

\section{Proof from Subsection~\ref{subsec:varepsilon_buchi}}
\subsection{Proof of Lemma~\ref{lem:varepsilon_buchi_necessary}}
\label{appen:proof_prop_varepsilon_buchi_necessary}
Let us first define formally the game built on a game form, a partial valuation and a probability function.
\begin{definition}
	Consider a game form $\formNF = \langle A,B,\outComeNF,\outCNF \rangle$, a partition of the outcomes $\outComeNF = \outComeNFEx \uplus \outComeNFLp$, a partial valuation $\alpha: \outComeNFEx \rightarrow [0,1]$ and a probability function $p_T: \outComeNFLp \rightarrow [0,1]$. We define the game $\G^\Bu_{\formNF,\alpha,p_T} = \Games{\Aconc}{\Bu(T)}$ where $\Aconc := \AConc$ with:
	\begin{itemize}
		\item $Q := \{ q_0,q_T,q_{\bar{T}},\top,\bot \}$;
		\item $\distribSet := \distribSet_{\outComeNFLp} \cup \distribSet_{\outComeNFEx} \cup \{ d^{\outComeNFEx}_0,d^{\outComeNFEx}_1 \} \cup \{ d_{q_0} \}$ with $\distribSet_{\outComeNFLp} := \{ d^{\outComeNFLp}_x \mid x \in p_T[\outComeNFLp] \}$ and $\distribSet_{\outComeNFEx} := \{ d^{\outComeNFEx}_x \mid x \in \alpha[\outComeNFEx] \}$;
		\item we define the function $g: \outComeNF \rightarrow \distribSet_{\outComeNFLp} \cup \distribSet_{\outComeNFEx}$ such that, for all $o \in \outComeNF$:
		\begin{equation*}
			g(o) := \begin{cases}
				d^{\outComeNFEx}_x \quad &\text{for } x := \alpha(o) \text{ if } \, o \in \outComeNFEx \\
				d^{\outComeNFLp}_x \quad &\text{for } x := p_T(o) \text{ if } \, o \in \outComeNFLp \\
			\end{cases}
		\end{equation*}
	
		Then, for all $a \in A$, $b \in B$, we have: $\delta(q_0,a,b) := g \circ \outCNF(a,b)$.
		Furthermore, $\delta(q_T,a,b) = \delta(q_{\bar{T}},a,b) = d_{q_0}$, $\delta(\top,a,b) = d^{\outComeNFEx}_1$ and $\delta(\bot,a,b) = d^{\outComeNFEx}_0$;
		\item for all $x \in \alpha[\outComeNFEx] \cup \{ 0,1 \}$, we have $\distribFunc(d_x^{\outComeNFEx})(\top) := x$ and $\distribFunc(d_x^{\outComeNFEx})(\bot) := 1 - x$. For all $x \in p_T[\outComeNFLp]$, we have $\distribFunc(d_x^{\outComeNFLp})(q_T) := x$ and $\distribFunc(d_x^{\outComeNFLp})(q_{\bar{T}}) := 1 - x$. Finally, we have $\distribFunc(d_{q_0})(q_0) := 1$;
		\item $T := \{ q_T,\top \}$.
	\end{itemize}
\end{definition}

Let us state a useful proposition relating the value of the state $q_0$ in the Büchi game $\G^\Bu_{\formNF,\alpha,p_T}$ and the value $u_{\formNF,\gf T \lf \bar{T}}(\alpha[p_T])$.
\begin{proposition}
	\label{prop:same_value_game_game_form}
	Consider a game form $\formNF$, a partition of the outcomes $\outComeNF = \outComeNFEx \uplus \outComeNFLp$, a partial valuation $\alpha: \outComeNFEx \rightarrow [0,1]$ and a probability function $p_T: \outComeNFLp \rightarrow [0,1]$. Consider the game $\G := \G^\Bu_{\formNF,\alpha,p_T}$. Let $v := \MarVal{\G}$ be the valuation of the states in $\G$ giving the value of states. Then: $v(q_0) = u_{\formNF,\gf T \lf \bar{T}}(\alpha[p_T])$. 
\end{proposition}
This is in fact quite obvious, even if it may not seem so because of the convoluted definitions with fixed points of different elements in these equalities. 
\begin{proof}
	First, given $u_T \in [0,1]$, we define a valuation $v_T(u_T) \in [0,1]^T$ such that $v_T(u_T)(\top) := 1$ and $v_T(u_T)(q_T) := u_T$. Furthermore, given $u_{\bar{T}} \in [0,1]$, we define a valuation $v_{\bar{T}}(u_{\bar{T}}) \in [0,1]^{Q \setminus T}$ such that $v_{\bar{T}}(u_{\bar{T}})(\bot) := 0$ and $v_{\bar{T}}(u_{\bar{T}})(q_{\bar{T}}) = v_{\bar{T}}(u_{\bar{T}})(q_0) := u_{\bar{T}}$. Let now show that, for all $u_T,u_{\bar{T}} \in [0,1]$, we have: $\alpha[p_T(u_T \sqcup u_{\bar{T}})] = \mu_{v_T(u_T) \sqcup v_{\bar{T}}(u_{\bar{T}})} \circ g: \outComeNF \rightarrow [0,1]$. Indeed:
	\begin{itemize}
		\item let $o \in \outComeNFEx$ and $x := \alpha(o) \in \alpha[\outComeNFEx]$. Then:
		\begin{align*}
		\mu_{v_T(u_T),\sqcup v_{\bar{T}}(u_{\bar{T}})} \circ g(o) & = \mu_{v_T(u_T) \sqcup v_{\bar{T}}(u_{\bar{T}})}(d^{\outComeNFEx}_x) \\
		& = x \cdot v_T(u_T)(\top) + (1 - x) \cdot v_{\bar{T}}(u_{\bar{T}})(\bot) = x \cdot 1 + (1 - x) \cdot 0 \\
		& = x = \alpha(o) = \alpha[p_T(u_T \sqcup u_{\bar{T}})](o)
		\end{align*}
		\item let $o \in \outComeNFLp$ and $x := p_T(o) \in p_T[\outComeNFLp]$. Then:
		\begin{align*}
		\mu_{v_T(u_T) \sqcup v_{\bar{T}}(u_{\bar{T}})} \circ g(o) & = \mu_{v_T(u_T) \sqcup v_{\bar{T}}(u_{\bar{T}})}(d^{\outComeNFLp}_x) \\
		& = x \cdot v_T(u_T)(q_T) + (1 - x) \cdot v_{\bar{T}}(u_{\bar{T}})(q_{\bar{T}}) = x \cdot u_T + (1 - x) \cdot u_{\bar{T}} \\
		& = p_T(o) \cdot u_T + (1 - p_T(o)) \cdot u_{\bar{T}} \\
		& = \alpha[p_T(u_T \sqcup u_{\bar{T}})](o)
		\end{align*}
	\end{itemize}
	
	Let us now fix some $u_T \in [0,1]$. Let us show that, for all $u_{\bar{T}} \in [0,1]$ we have: $\OneStepSet{T}{v_T(u_T)}(v_{\bar{T}}(u_{\bar{T}})) = v_{\bar{T}}(f_{\formNF,\bar{T}}\langle \alpha[p_T(u_T)]\rangle(u_{\bar{T}}))$. Straightforwardly, they coincide on the state $\bot$ as they are both equal to $0$. Let us consider now what happens at state $q_{\bar{T}}$.
	\begin{align*}
	v_{\bar{T}}(f_{\formNF,\bar{T}}\langle \alpha[p_T(u_T)]\rangle(u_{\bar{T}}))(q_{\bar{T}}) & = f_{\formNF,\bar{T}}\langle \alpha[p_T(u_T)]\rangle(u_{\bar{T}}) = \va_{\gameNF{\formNF}{\alpha[p_T(u_T \sqcup u_{\bar{T}})]}} \\ 
	& = \va_{\gameNF{\formNF}{\mu_{(v_T(u_T) \sqcup v_{\bar{T}}(u_{\bar{T}}))} \circ g}} = \va_{\gameNF{\formNF_{q_{\bar{T}}}}{\mu_{(v_T(u_T) \sqcup v_{\bar{T}}(u_{\bar{T}}))}}} \\
	& = \OneStep(v_T(u_T) \sqcup v_{\bar{T}}(u_{\bar{T}}))(q_{\bar{T}})\\
	& = \OneStepSet{T}{v_T(u_T)}(v_{\bar{T}}(u_{\bar{T}}))(q_{\bar{T}}) \\  
	\end{align*}
	The result is the same for $q_0$. This holds for all $u_{\bar{T}} \in [0,1]$. Let us fix $u_{\bar{T}} := u_{\formNF,\lf \bar{T}}(\alpha[p_T(u_T)]$. By considering a least fixed point, we have:
	\begin{displaymath}
	v_{\lf \bar{T}}[v_T(u_T)] = v_{\bar{T}}(u_{\bar{T}})
	\end{displaymath}
	From this, it follows that:
	\begin{displaymath}
	\alpha[p_T(u_T)]_{\formNF,\lf \bar{T}} = \alpha[p_T(u_T \sqcup u_{\bar{T}})] = \mu_{v_T(u_T),v_{\bar{T}}(u_{\bar{T}})} \circ g
	\end{displaymath}
	
	Finally, let us show that, for all $u_T \in [0,1]$, we have $v_T(f_{\formNF,\lf \bar{T}}\langle \alpha[p_T]\rangle(u_T)) = \OneStepLFP{T}(v_T(u_T))$. The equality holds for the state $\top$, as they are both equal to $1$ on that state. Let us now consider the state $q_T$. 
	
	\begin{align*}
	v_T(f_{\formNF,\lf \bar{T}}\langle \alpha[p_T]\rangle(u_T))(q_T) & = f_{\formNF,\lf \bar{T}}\langle \alpha[p_T]\rangle(u_T) = \va_{\gameNF{\formNF}{(\alpha[p_T(u_T)])_{\formNF,\lf \bar{T}}}} \\ 
	& = \va_{\gameNF{\formNF_{q_T}}{\mu_{(v_T(u_T) \sqcup v_{\bar{T}}(u_{\bar{T}}))}}}  \\
	& = \OneStep(v_T(u_T) \sqcup v_{\bar{T}}(u_{\bar{T}}))(q_T)\\
	& = \OneStep(v_T(u_T) \sqcup v_{\lf \bar{T}}[v_T])(q_T)\\
	& = \OneStepLFP{T}(v_T(u_T))(q_T) \\  
	\end{align*}
	As this holds for all $u_T \in [0,1]$, by considering the greatest fixed point, we obtain: 
	\begin{displaymath}
	v_T(u_{\formNF,\gf T \lf \bar{T}}(\alpha[p_T])) = v_{\gf T \lf \bar{T}} \in [0,1]^T
	\end{displaymath}
	Let $u := u_{\formNF,\gf T \lf \bar{T}}(\alpha[p_T]) \in [0,1]$. By Theorem~\ref{thm:value_buchi}, the valuation $v \in [0,1]^Q$ giving the values of the Büchi game is equal to:
	\begin{displaymath}
	v = v_{\gf T \lf \bar{T}} \sqcup v_{\lf \bar{T}}[v_{\gf T \lf \bar{T}}]
	\end{displaymath}
	In fact, we obtain that:
	\begin{displaymath}
	v = v_T(u) \sqcup v_{\bar{T}}(u)
	\end{displaymath}
	Following, $v(q_0) = v_T(u)(q_0) = u = u_{\formNF,\gf T \lf \bar{T}}(\alpha,p_T)$. 
\end{proof}

We can now proceed to the proof of Lemma~\ref{lem:varepsilon_buchi_necessary}.
\begin{proof}[Proof of Lemma~\ref{lem:varepsilon_buchi_necessary}]
	Consider a game form $\formNF$, a partition of the outcomes $\outComeNF = \outComeNFEx \uplus \outComeNFLp$, a partial valuation $\alpha: \outComeNFEx \rightarrow [0,1]$ and a probability function $p_T: \outComeNFLp \rightarrow [0,1]$. Consider the game $\G := \G^\Bu_{\formNF,\alpha,p_T}$. Let $v := \MarVal{\G}$ be the valuation of the states in $\G$ giving the value of states. By Proposition~\ref{prop:same_value_game_game_form}, we have $v(q_0) = u_{\formNF,\gf T \lf \bar{T}}(\alpha[p_T])$. 
		
	
	Assume that the game form $\formNF$ is aBM w.r.t. $\alpha,p_T$. If $u = 0$, any Player $\A$ strategy is optimal. Assume now that $u > 0$ and consider some $\varepsilon > 0$. Then, there exists a $\GF$-strategy $\sigma_\A$ in the game form $\formNF$ as in Definition~\ref{def:epsilon_buchi_maximizable}. Consider then a Player $\A$ positional strategy $\s_\A$ playing the $\GF$-strategy $\sigma_\A$ at state $q_0$, i.e. $\s_\A(q_0) := \sigma_\A$. Let us show that it is $\varepsilon$-optimal by applying Lemma~\ref{lem:uniform_guarantee}. Specifically, consider the valuation $w: Q \rightarrow [0,1]$ such that: $w(q_0) = w(q_T) = w(q_{\bar{T}}) := u - \varepsilon$, $w(\top) := 1$ and $w(\bot) := 0$. Straightforwardly, we have $\va_{\gameNF{\formNF_\top}{\mu_w}}(\s_\A(\top)) = w(\top) $ and similarly for the state $\bot$. Furthermore, $\va_{\gameNF{\formNF_{q_T}}{\mu_w}}(\s_\A(q_T)) = \va_{\gameNF{\formNF_{q_{\bar{T}}}}{\mu_w}}(\s_\A(q_{\bar{T}})) = w(q_0) = w(q_T) = w(q_{\bar{T}})$. Consider now what happens in the local interaction $\formNF_{q_0}$. Note that, for all $x \in p_T[\outComeNFLp]$ we have $\mu_w(d_x^{\outComeNFLp}) = u - \varepsilon$. Furthermore, for all $o \in \outComeNFEx$, we have $\mu_w(g(o)) = x = \alpha(o)$. Recall that $A_b = \{ a \in \Supp(\sigma_\A) \mid \outCNF(a,b) \in \outComeNFEx \}$ from Definition~\ref{def:epsilon_buchi_maximizable}. It follows that, for all $b \in B$:
	\begin{align*}
		\outM_{\gameNF{\formNF_{q_0}}{\mu_w}}(\s_\A(q_0),b) & = \sum_{a \in A} \sigma_\A(a) \cdot \mu_w(\delta(q_0,a,b)) \\
		& = \sum_{a \in \Supp(\sigma_\A) \setminus A_b} \sigma_\A(a) \cdot \mu_w(\delta(q_0,a,b)) + \sum_{a \in A_b} \sigma_\A(a) \cdot \mu_w(\delta(q_0,a,b)) \\
		& = \sum_{a \Supp(\sigma_\A) \setminus A_b} \sigma_\A(a) \cdot \mu_w(d^{\outComeNFLp}_{p_T(\outCNF(a,b))}) + \sum_{a \in A_b} \sigma_\A(a) \cdot \mu_w(d^{\outComeNFEx}_{\alpha(\outCNF(a,b))}) \\
		& = \sum_{a \Supp(\sigma_\A) \setminus A_b} \sigma_\A(a) \cdot (u - \varepsilon) + \sum_{a \in A_b} \sigma_\A(a) \cdot \alpha(\outCNF(a,b)) \\
		& \geq \sum_{a \in A_b} \sigma_\A(a) \cdot (u - \varepsilon) + (u - \varepsilon) \cdot \sigma_\A[A_b] \\
		& = (u - \varepsilon) \cdot \sum_{a \in \Supp(\sigma_\A)} \sigma_\A(a) \\
		& = u - \varepsilon = w(q_0)
	\end{align*}
	The inequaity comes from the fact that $\sigma_\A$ is given by Definition~\ref{def:epsilon_buchi_maximizable}. Hence, the strategy $\s_\A$ locally dominates the valuation $w$. Consider now an end component in the MDP induced by the strategy $\s_\A$. 
	If this EC $H = (Q_H,\beta)$ does not contain the state $q_0$, there is no issue. Either it is $\{ \top \}$ of value 1 in any case, or it is $\{ \bot \}$ of value 0. Assume now that $H$ contains $q_0$ and consider the set of actions $\beta(q_0) \subseteq B$ compatible with this end component (see Definition~\ref{def:end_component}). For all $b \in \beta(q_0)$, it must be that $A_b = \emptyset$ (otherwise the game could leave the state $q_0$ to $\top$ or $\bot$). Hence, by choice of the $\GF$-strategy $\sigma_\A$, it follows that there is $a \in \Supp(\sigma_\A)$ such that $p_T(\outCNF(a,b)) > 0$. That is, if Player $\B$ plays this action $b$, there is a (fixed) positive probability to visit the $T$. As this holds for all actions $b \in \beta(q_0)$ it follows that for all Player $\B$ strategies in the game  $\Aconc_H^{\s_\A}$, the set $T$ is infinitely often almost-surely. That is, the value of the game $\Aconc_H^{\s_\A}$ is one for all states. We can then apply Lemma~\ref{lem:uniform_guarantee}: the strategy $\s_\A$ guarantees the valuation $w$. In particular, it is $\varepsilon$-optimal from the state $q_0$ (as $\MarVal{\G}(q_0) = v(q_0) = u = w(q_0) + \varepsilon$).
	
	Let us now assume that, for all $\varepsilon > 0$ there exists a positional Player $\A$ strategy that is $\varepsilon$-optimal from $q_0$. Let us show that the game form $\formNF$ is aBM w.r.t. $\alpha,p_T$. Let $0 < \varepsilon < u$. Consider a positional Player $\A$ strategy $\s_\A$ that is $\varepsilon$-optimal from $q_0$. Let $\sigma_\A := \s_\A(q_0)$ be a $\GF$-strategy in $\formNF_{q_0}$, consider some action $b \in B$ and consider a Player $\B$ positional strategy $\s_\B$ such that $\s_\B(q_0) := b$. Assume that $A_b = \emptyset$ (with $A_b$ defined w.r.t. $\s_\A(q_0)$ in $\formNF_{q_0}$). If, in the Markov chain induced by both strategies $\s_\A,\s_\B$, the only state reachable from $q_0$ is $q_{\bar{T}}$ then the set $T$ is never visited and the value of the game with strategies $\s_\A,\s_\B$ is $0 < u - \varepsilon$, which is not possible by assumption. Therefore, there must be an action $a \in \Supp(\sigma_\A)$ such that $\distribFunc(d^E_\outCNF(a,b))(q_T) > 0$. That is, $p_T(\outCNF(a,b)) > 0$. Assume now $A_b \neq \emptyset$. Let $w: Q \rightarrow [0,1]$ be the value of the game $\G$ with the Player $\A$ strategy $\sigma_\A$ fixed (it is therefore an MDP). Straightforwardly, $w(\top) = 1$ and $w(\bot) = 0$. In addition, by assumption, we have $w(q_0) \geq u - \varepsilon$. Furthermore, as the Büchi objective is prefix-independent, the valuation $w$ ensures that, for all $q \in Q$, we have $\va_{\gameNF{\formNF_{q}}{\mu_w}} = w(q)$. In particular,  $w(q_T) = \va_{\gameNF{\formNF_{q_T}}{\mu_w}} = w(q_0)$ and similarly for the state $q_{\bar{T}}$. Overall, we have:
	\begin{align*}
		w(q_0) \leq \outM_{\gameNF{\formNF_{q_0}}{\mu_w}}(\s_\A(q_0),b) & = \sum_{a \in A} \s_\A(a) \cdot \mu_w(\delta(q_0,a,b)) \\
		& = \sum_{a \in \Supp(\sigma_\A) \setminus A_b} \sigma_\A(a) \cdot \mu_w(\delta(q_0,a,b)) + \sum_{a \in A_b} \sigma_\A(a) \cdot \mu_w(\delta(q_0,a,b)) \\
		& = \sum_{a \in \Supp(\sigma_\A) \setminus A_b} \sigma_\A(a) \cdot \mu_w(d^{\outComeNFLp}_{p_T(\outCNF(a,b))}) + \sum_{a \in A_b} \sigma_\A(a) \cdot \mu_w(d^{\outComeNFEx}_{\alpha(\outCNF(a,b))}) \\
		& = \sum_{a \in \Supp(\sigma_\A) \setminus A_b} \sigma_\A(a) \cdot w(q_0) + \sum_{a \in A_b} \sigma_\A(a) \cdot \alpha(\outCNF(a,b)) \\
	\end{align*}
	We obtain:
	\begin{align*}
		\sum_{a \in A_b} \sigma_\A(a) \cdot \alpha(\outCNF(a,b)) & \geq w(q_0) \cdot (1 - \sum_{a \in \Supp(\sigma_\A) \setminus A_b} \sigma_\A(a))
		& = w(q_0) \cdot \sigma_\A[A_b] \geq (u - \varepsilon) \cdot \sigma_\A[A_b]
	\end{align*}
	We have that the $\GF$-strategy $\sigma_\A$ satisfies the specifications of Definition~\ref{def:epsilon_buchi_maximizable}. As this holds for all $0 < \varepsilon < u$, it follows that the game form $\formNF$ is aBM w.r.t. $\alpha,p_T$.
	
\end{proof}

\subsection{Proof of Theorem~\ref{lem:varepsilon_buchi_sufficient}}
\label{appen:proof_lem_varepsilon_buchi_sufficient}

Before proving Theorem~\ref{lem:varepsilon_buchi_sufficient}, let us state a useful proposition that links the value of a Büchi game and a reachability game obtained from it. 
More precisely:
\begin{lemma}
	Consider a game form $\formNF$, a partition of the outcomes $\outComeNF = \outComeNFEx \uplus \outComeNFLp$ and a partial valuation $\alpha: \outComeNFEx \rightarrow [0,1]$. Let $u_\mathsf{reach} := u_{\formNF,\diamondsuit}[\alpha]$. Then, let $\outComeNF_T := \{ o \in \outComeNFEx \mid \alpha(o) = u_\mathsf{reach} \}$. We then consider the partial valuation $\alpha^\Bu: \outComeNFEx \setminus \outComeNF_T \rightarrow [0,1]$ such that $\alpha^\Bu := \restriction{\alpha}{\outComeNFEx \setminus \outComeNF_T}$ and the probability function $p_T: \outComeNFLp \cup \outComeNF_T \rightarrow [0,1]$ such that, for all $o \in \outComeNFLp$, we have $p_T(o) := 0$ and for all $o \in \outComeNF_T$, we have $p_T(o) := 1$. In that case, we have $u_{\formNF,\gf T \lf \bar{T}}(\alpha^\Bu[p_T]) \geq u_\mathsf{reach} = u_{\formNF,\diamondsuit}(\alpha)$.
	\label{lem:partial_val_reach_buchi}
\end{lemma}
\begin{proof}
	Let $u_\Bu := u_{\formNF,\gf T \lf \bar{T}}(\alpha^\Bu[p_T])$. Let $u \in [0,1]$. For all $o \in \outComeNFEx \setminus \outComeNF_T$, we have $\alpha[p_T(u \sqcup u_\mathsf{reach})](o) = \alpha(o) = \alpha^\mathsf{reach}[u](o)$. Consider now some $o \in \outComeNFLp \cup \outComeNF_T$. 
	\begin{itemize}
		\item If $o \in \outComeNF_T$, we have $p_T(o) = 1$ and: 
		\begin{displaymath}
		\alpha^\Bu[p_T(u_\mathsf{reach} \sqcup u)](o) = p_T(o) \cdot u_\mathsf{reach} + (1 - p_T(o)) \cdot u = u_\mathsf{reach} = \alpha(o) = \alpha^\mathsf{reach}[u](o)
		\end{displaymath}
		\item If $o \in \outComeNFLp$, then $p_T(o) = 0$. Hence:
		\begin{displaymath}
		\alpha^\Bu[p_T(u_\mathsf{reach} \sqcup u)](o) = p_T(o) \cdot u_\mathsf{reach} + (1 - p_T(o)) \cdot u = u = \alpha^\mathsf{reach}[u](o)
		\end{displaymath}
	\end{itemize}
	That is, for all $u \in [0,1]$, we have $\alpha^\Bu[p_T(u_\mathsf{reach} \sqcup u)]  = \alpha^\mathsf{reach}[u]$. 
	
	For all $u \in [0,1]$, we have $f_{\formNF,\bar{T}}\langle \alpha^\Bu[p_T(u_\mathsf{reach})]\rangle(u) = \va_{\gameNF{\formNF}{\alpha^\Bu[p_T(u_\mathsf{reach} \sqcup u)]}} = 
	\va_{\gameNF{\formNF}{\alpha[u]}} =
	 f_\formNF\langle \alpha\rangle(u)$. Hence, the least fixed of the function $f_{\formNF,\bar{T}}\langle \alpha^\Bu[p_T(u_\mathsf{reach})]\rangle$ is equal to $u_\mathsf{reach}$. That is, $f_{\formNF,\lf \bar{T}}\langle \alpha[p_T]\rangle(u_\mathsf{reach}) = u_\mathsf{reach}$. It follows that $u_\Bu$, which if the greatest fixed point of the function $f_{\formNF,\lf \bar{T}}[\langle \alpha[p_T]\rangle$, is greater than or equal to $u_\mathsf{reach}$, i.e. $u_\Bu \geq u_\mathsf{reach}$.
	 
\end{proof}

We can now proceed to the proof of Theorem~\ref{lem:varepsilon_buchi_sufficient}.
\begin{proof}[Proof of Theorem~\ref{lem:varepsilon_buchi_sufficient}]
	Consider a Büchi game $\G = \Games{\Aconc}{\Bu(T)}$ and assume that all local interactions at states in $Q \setminus T$ are aBM. 
	For all $u \in [0,1]$, we denote by $Q_u$ the set of states in $Q$ whose values w.r.t. the valuation $v := \MarVal{\G}: Q \rightarrow [0,1]$ giving the values of the game is $u$: $Q_u := \{ q \in Q \mid v(q) = u \}$. 
	
	
	Let $V := \{ v(q) \mid q \in Q \}$. Let us write $V$ as $\{ 1 = v_0 > v_1 > v_2 > \ldots > v_k = 0 \}$ for $k = |V| - 1$. We define:
	\begin{itemize}
		\item $m := \min_{i \leq k-1} v_i - v_{i+1} > 0$;
		\item $p := \max_{d \in \distribSet} \max_{\substack{q \in Q \\ \distribFunc(d)(q) < 1}} \distribFunc(d)(q)$;
	\end{itemize}
	Finally, consider some $0 < \varepsilon \leq \min(m/2,1-p)$. 
	
	We define a sequence $0 = \varepsilon_0 < \varepsilon_1 < \varepsilon_2 < \ldots < \varepsilon_{k-1} \leq \varepsilon$ with $\varepsilon_{k} := 0$
	such that, for all probability functions $p \in \Dist(\llbracket 0,k \rrbracket)$ and $j \in \llbracket 1,k \rrbracket$ with $p_j \leq p$
	, we have the following implication, for some positive $\alpha,\eta$:
	\begin{equation}
	\sum_{i = 0}^{k} p_i \cdot v_i \geq v_j - \varepsilon \cdot m/2 \Rightarrow \sum_{\substack{i = 0 \\ i \neq j}}^{k} p_i \cdot \varepsilon_i \leq (1 - p_j) \cdot (\varepsilon_j - \cdot \varepsilon \cdot \alpha^2 \cdot \eta)
	\label{eqn:implication}
	\end{equation}
	This sequence will be explicitly defined at the end of this proof.
	
	Let $v': Q \rightarrow [0,1]$ be the valuation of the states such that, for all $0 \leq i \leq k$ and $q \in Q_{v_i}$, we have $v'(q) := v(q) - \varepsilon_i = v_i - \varepsilon_i \geq v_i - \varepsilon$. In particular, 
	since $\varepsilon \leq m/2$, note that $v'(q) = v'(q')$ implies $v(q) = v(q')$ for all $q,q' \in Q$. Let us now exhibit a positional Player $\A$ strategy dominating this valuation. In turn, this strategy will be uniformly $\varepsilon$-optimal (and uniformly almost-sure when restricted to $Q_1$ since $\varepsilon_0 = 0$).
	
	The strategy $\s_\A$ is arbitrarily defined on $Q_0$. Consider now some $0 \leq i \leq k-1$
	. Let $Q'_i := Q_{v_i} \setminus T$ be the states of $Q_{v_i}$ not in the target $T$. Let $\G_i$ be the reachability game $\G_i := \G_{Q \setminus Q'_i} = \Games{\Aconc_{Q \setminus Q'_i}}{\Reach(\{ \top \})}$ extracted from $\G$ from Definition~\ref{def:game_extracted}. We define inductively an increasing sequence of sets of states $(\mathcal{Q}^{i}_{n})_{n \geq 0}$ to properly define the strategy $\s_\A$. Specifically, we set $\mathcal{Q}^i_0 := \emptyset$ and for all $n \geq 0$, we consider the partial valuation $\alpha^{\mathcal{Q}^i_n}: \distribSet^{\mathcal{Q}^i_{k-1}}_{\mathsf{Ex}} \rightarrow [0,1]$ of the Nature states from Definition~\ref{def:partial_val_in_nature_states}. Then, by Lemma~\ref{lem:value_in_RM_up_to_u}, the set $\mathcal{Q}_n^{i,\diamondsuit} \neq \emptyset$ (defined in the same lemma). We let $\mathcal{Q}^i_{n+1} :=\mathcal{Q}^i_n \cup  \mathcal{Q}_n^{i,\diamondsuit}$. Then, the sequence $(\mathcal{Q}^i_{n})_{n \geq 0}$ is equal to the set $Q'_i$: $\mathcal{Q}^i_{j} = Q'_i$ for some $j \geq 0$. 
	
	Let us now define a Player $\A$ positional strategy $\s_\A^i$ on $Q_{v_i}$ in the game $\G$. For all states $q \in Q_{v_i} \cap T$, let $\s_\A^i(q)$ be a $\GF$-strategy that is optimal in the game form $\formNF_{q}$ w.r.t. the valuation of the states $v$: $\va_{\gameNF{\formNF_q}{\mu_v}}(\s_\A^i(q)) = v(q)$. Let us now define the strategy $\s_\A^i$ on $Q'_i = Q_{v_i} \setminus T$. Let $q \in Q_i'$. There exists some $k \in \N$ such that $q \in \mathcal{Q}^i_k \setminus \mathcal{Q}^i_{k-1}$. That is, $q \in \mathcal{Q}_{k-1}^{i,\diamondsuit}$. In other words, we have $u_{\formNF_q,\diamondsuit}(\alpha^{\mathcal{Q}^i_{k-1}}) = v_i$. Let us  denote by $\distribSet_T^{v_i}$ the subset of Nature states of $\distribSet^{v_i}_{\mathsf{Ex}}$ whose values w.r.t. $\alpha^{\mathcal{Q}^i_{k-1}}$ is $v_i$: $\distribSet_T^{v_i} := \{ d \in \distribSet^{v_i}_{\mathsf{Ex}} \mid \alpha^{\mathcal{Q}^i_{k-1}}(d) = v_i \}$. Then, we can apply Lemma~\ref{lem:partial_val_reach_buchi} with the partial valuation $\alpha^{\mathcal{Q}^i_{k-1}}$ to obtain a partial valuation $\alpha^{\mathcal{Q}^i_{k-1},\Bu}: \distribSet^{\mathcal{Q}^i_{k-1}}_{\mathsf{Ex}} \setminus \distribSet_T^{\mathcal{Q}^i_{k-1}} \rightarrow [0,1]$ and the probability function $p_T: \distribSet^{\mathcal{Q}^i_{k-1}}_{\mathsf{Lp}} \cup \distribSet_T^{\mathcal{Q}^i_{k-1}} \rightarrow [0,1]$ such that if $d \in \distribSet^{\mathcal{Q}^i_{k-1}}_{\mathsf{Lp}}$, then $p_T(d) := 0$ and if $d \in \distribSet_T^{\mathcal{Q}^i_{k-1}}$, then $p_T(d) := 1$. 
	We then obtain that the value $u_{\formNF_{q},\square T \diamondsuit \bar{T}}(\alpha^{\mathcal{Q}_{k-1},\Bu}[p_T]) \geq v_i$. Since the game form $\formNF_{q}$ is aBM, there exists a $\GF$-strategy $\sigma_\A^q \in \Dist(A)$ in the game form $\formNF_{q}$ for $\varepsilon_{err} := \min(m/2,(1 - p) \cdot \varepsilon \cdot \alpha ^2 \cdot \eta)$ as in Definition~\ref{def:epsilon_buchi_maximizable} (of aBM game forms). We set $\s_\A^i(q) := \sigma_\A^q$. 
	
	We then consider the positional Player $\A$ strategy $\s_\A$ that is obtained by 'gluing' together all strategies $\s_\A^i$. Specifically, for all $0 \leq i \leq k-1$ and $q \in Q_{v_i}$, we set $\s_\A(q) := \s_\A^i(q)$. The strategy $\s_\A$ plays arbitrarily on $Q_0$. Let us show that it guarantees the valuation $v'$ defined above. We want to apply Lemma~\ref{lem:uniform_guarantee}. First, let us show that the strategy $\s_\A$ locally dominates the valuation $v'$. Let $1 \leq i \leq k-1$. Consider some state $q \in Q_{v_i} \cap T$. The strategy $\s_\A$ is such that $\va_{\gameNF{\formNF_q}{\mu_v}}(\s_\A(q)) = v(q)$. Then, consider some action $b \in B$. Let $A_b := \{ a \in \Supp(\s_\A(q)) \mid \delta(q,a,b) \in \distribSet^{Q_{v_i}}_{\mathsf{Ex}} \}$. If $A_b \neq \emptyset$, we have $p_b := \s_\A(q)[A_b] > 0$. Then, for all $0 \leq j \leq k$, we let $p_j \in [0,1]$ be such that $p_j := \sum_{a \in A_b} \s_\A(q)(a) \cdot \distribFunc(\delta(q,a,b))[Q_{v_j}]/p_b$. Note that, by definition, we have $p_i \leq p$. Furthermore, we have $\sum_{0 \leq j \leq k} p_j \cdot v_j \geq v_i$ (since $\s_\A(q)$ is optimal in $\gameNF{\formNF_{q}}{\mu_v}$). That is, by Equation~\ref{eqn:implication}, we have: 
	\begin{displaymath}
		\sum_{\substack{j = 0 \\ j \neq i}}^{k} p_j \cdot \varepsilon_j \leq (1 - p_i) \cdot (\varepsilon_i - \varepsilon \cdot \alpha^2 \cdot \eta) \leq (1 - p_i) \cdot \varepsilon_i
	\end{displaymath}
	We obtain:
	\begin{align*}
		\outM_{\gameNF{\formNF_{q}}{\mu_{v'}}}(\s_\A(q),b) & = \sum_{a \in \Supp(\s_\A(q))} \s_\A(q)(a) \cdot \mu_{v'}(\delta(q,a,b)) \\
		& = \sum_{a \in \Supp(\s_\A(q)) \setminus A_b} \s_\A(q)(a) \cdot \mu_{v'}(\delta(q,a,b)) + \sum_{a \in A_b} \s_\A(q)(a) \cdot \mu_{v'}(\delta(q,a,b)) \\
		& = \s_\A(q)[A \setminus A_b] \cdot v'(q) + \sum_{a \in A_b} \s_\A(q)(a) \cdot \sum_{0 \leq j \leq k} \distribFunc(\delta(q,a,b))[Q_{v_j}] \cdot (v_j - \varepsilon_j) \\
		& = \s_\A(q)[A \setminus A_b] \cdot (v_i - \varepsilon_i) + \sum_{0 \leq j \leq k} \sum_{a \in A_b} \s_\A(q)(a) \cdot  \distribFunc(\delta(q,a,b))[Q_{v_j}] \cdot (v_j - \varepsilon_j) \\
		& = \s_\A(q)[A \setminus A_b] \cdot (v_i - \varepsilon_i) + \s_\A(q)[A_b] \left( p_i (v_i - \varepsilon_i) + \sum_{\substack{j = 0 \\ j \neq i}}^k p_j \cdot (v_j - \varepsilon _j) \right) \\
		& \geq \s_\A(q)[A \setminus A_b] \cdot (v_i - \varepsilon_i) + \s_\A(q)[A_b] \cdot \left( p_i (v_i - \varepsilon_i) + (1 - p_i) v_i - \sum_{\substack{j = 0 \\ j \neq i}}^k p_j \cdot \varepsilon _j \right)
		\\
		& \geq \s_\A(q)[A \setminus A_b] \cdot (v_i - \varepsilon_i) + \s_\A(q)[A_b] \cdot \left( p_i (v_i - \varepsilon_i) + (1 - p_i) v_i - (1 - p_i) \cdot \varepsilon _i \right) \\
		& = \s_\A(q)[A \setminus A_b] \cdot (v_i - \varepsilon_i) + \s_\A(q)[A_b] \cdot (v_i - \varepsilon_i) = v_i - \varepsilon_i = v'(q)
	\end{align*}
	As this holds for all action $b \in B$, we have $\va_{\gameNF{\formNF_{q}}{\mu_{v'}}}(\s_\A(q)) \geq v'(q)$. 

	Consider now some state $q \in Q'_i = Q_{v_i} \setminus T$. The arguments are similar, but we have to be a little more precise. Indeed, now we have $\sum_{0 \leq j \leq k} p_j \cdot v_j \geq v_i - \varepsilon_{err}$ with $\varepsilon_{err} \leq (1 - p) \cdot \varepsilon \cdot \alpha^2 \cdot \eta$ by definition of the strategy $\s_\A$
	. Furthermore, Equation~(\ref{eqn:implication}) gives that 
	$\sum_{\substack{j = 0 \\ j \neq i}}^{k} p_j \cdot \varepsilon_j \leq (1 - p_i) \cdot (\varepsilon_i - \varepsilon \cdot \alpha^2 \cdot \eta)$. Hence, we obtain:
	\begin{align*}
		\hspace*{-0.5cm}
		\outM_{\gameNF{\formNF_{q}}{\mu_{v'}}}(\s_\A(q),b) & = \s_\A(q)[A \setminus A_b] \cdot (v_i - \varepsilon_i) + \s_\A(q)[A_b] \cdot \sum_{0 \leq j \leq k} p_j \cdot (v_i - \varepsilon _i) \\
		& \geq \s_\A(q)[A \setminus A_b] \cdot (v_i - \varepsilon_i) + \s_\A(q)[A_b] \cdot \left( p_i \cdot (v_i - \varepsilon_{i}) + (1 - p_i) \cdot v_i - \varepsilon_{err} 
		- \sum_{\substack{j = 0 \\ j \neq i}}^k p_j \cdot \varepsilon_j) \right)\\
		& \geq \s_\A(q)[A \setminus A_b] \cdot (v_i - \varepsilon_i) + \s_\A(q)[A_b] \cdot \left( v_i - p_i \cdot \varepsilon_i  - \varepsilon_{err} - (1 - p_i) \cdot (\varepsilon_i - \varepsilon \cdot \alpha^2 \cdot \eta) \right)\\
		& \geq \s_\A(q)[A \setminus A_b] \cdot (v_i - \varepsilon_i) + \s_\A(q)[A_b] \cdot \left( v_i - \varepsilon_i  - (1 - p) \cdot \varepsilon \cdot \alpha^2 \cdot \eta + (1 - p_i) \cdot \varepsilon \cdot \alpha^2 \cdot \eta \right)\\
		& = \s_\A(q)[A \setminus A_b] \cdot (v_i - \varepsilon_i) + \s_\A(q)[A_b] \cdot \left( v_i - \varepsilon_i  + (p - p_i) \cdot  \varepsilon \cdot \alpha^2 \cdot \eta \right)\\
		& \geq v_i - \varepsilon_i = v'(q)
	\end{align*}
	Again, as this holds for all action $b \in B$, we have $\va_{\gameNF{\formNF_{q}}{\mu_{v'}}}(\s_\A(q)) \geq v'(q)$. For the case $i = 0$ and $v_i = 1$, in fact we have that $A_b$ is necessarily empty since the mean of the values not in $Q_1$ must be at least $1 - m/2 > v_2$. That is, $\va_{\gameNF{\formNF_{q}}{\mu_{v'}}}(\s_\A(q)) = 1 = v'(q)$ for all $q \in Q_1$. 
	
	Consider now an EC $H := (Q_H,\beta)$ in the MDP induced by the positional strategy $\s_\A$. If $Q_H \subseteq T$, then the value of the game $\Aconc_H^{\s_\A}$ of any state in $Q_H$ is 1. Let us now assume that $Q_H \setminus T \neq \emptyset$. Since the strategy $\s_\A$ locally dominates the valuation $v'$, there exists a value $v'_H \in [0,1]$ such that all states in $H$ have that value w.r.t. the valuation $v'$: for all $q \in Q_H$, we have $v'(q) = v'_H$. Let $0 \leq i \leq k$ be such that $v'_H = v_i - \varepsilon_i$. If $i = k$, then $v'_H = 0$. Hence, let us assume that $i \leq k-1$. Consider the least index $n \in \N$ such that there is a state $q \in (Q_H \setminus T) \cap \mathcal{Q}^i_n$ which exists since $(Q_H \setminus T) \subseteq Q'_i = \mathcal{Q}^i_j$ for some $j \in \N$. Consider an action $b \in \beta(q)$. Consider the set of actions $A_b := \{ a \in \Supp(\s_\A(q)) \mid \delta(q,a,b) \in \distribSet^{\mathcal{Q}^i_n}_{\mathsf{Ex}} \}$. It is necessarily empty as an action in that set would lead to a Nature state leaving the EC $H$. Hence, by definition of the strategy $\s_\A(q)$ (and by Definition~\ref{def:epsilon_buchi_maximizable} of aBM game forms), there exists some $a \in \Supp(\s_\A(q))$ such that $p_T(\delta(q,a,b)) = 1$ 
	. Since $\delta(q,a,b) \in \distribSet^{Q_{v_i}}_{\mathsf{Lp}}$ and $\distribFunc(\delta(q,a,b))[\mathcal{Q}^i_{n-1}] = 0$ (by definition of $n$) it follows that $\distribFunc(\delta(q,a,b))[Q^{v_i} \cap T] > 0$\footnote{Informally, there is a positive probability to exit $\mathcal{Q}^i_{n}$, it is not via $Q \setminus Q^{v_i}$ and it is not via $\mathcal{Q}^i_{n-1} \subseteq Q^{v_i} \setminus T$, hence it must via $Q^{v_i} \cap T$.}.	Hence, regardless of the actions $b \in \beta(q)$ considered, there is probability at least $p > 0$ (for some well chosen $p$, for instance the least probability to reach $T$ over all possible actions) to reach the target $T$. Hence, for all Player $\B$ strategies $\s_\B$ in the game $\Aconc_H^{\s_\A}$, if the state $q$ is seen infinitely often, then the set $T$ is seen infinitely often almost-surely. This is the case for any state in $(Q_H \setminus T) \cap \mathcal{Q}^i_n$. With a similar argument, we can argue that, for all Player $\B$ strategies $\s_\B$ in the game $\Aconc_H^{\s_\A}$, if a state $q \in (Q_H \setminus T) \cap \mathcal{Q}^{i+1}_n$ is seen infinitely often, then the set $T \cup \mathcal{Q}^{i}_n$ is seen infinitely often almost-surely. That is, the set $T$ is seen infinitely often almost surely. In fact, we can show that this holds for any state in $(Q_H \setminus T) \cap \mathcal{Q}^{i}_j = Q_H \setminus T$. It follows that this also holds for any state in $Q_H$. That is, the value of the game $\Aconc_H^{\s_\A}$ is equal to 1 for all states $q \in Q_H$. We can then conclude by applying Lemma~\ref{lem:uniform_guarantee}.
	
	Now, consider some positive $0 < x \leq m/2$. We define a sequence of $0 = \varepsilon_0 < \varepsilon_1 < \varepsilon_2 < \ldots < \varepsilon_{k-1} \leq \varepsilon$ with $\varepsilon_{k} := 0$
	such that, for all probability function $p \in \Dist(\llbracket 0,k \rrbracket)$ and $j \in \llbracket 1,k \rrbracket$ with $p_j = 0$, we have the following implication, for $\alpha,\eta$ defined below:
	\begin{displaymath}
	\sum_{i = 0}^{k} p_i \cdot v_i \geq v_j - 
	x \Rightarrow \sum_{i = 0}^{k} p_i \cdot \varepsilon_i \leq \varepsilon_j - \varepsilon \cdot \alpha^2 \cdot \eta
	\end{displaymath}
	with:
	\begin{itemize}
		\item $\alpha := \max_{1 \leq i \leq k-1} \frac{1 - (v_i -
		x)}{1 - v_{i+1}} < 1$ by definition of $m$ and since $x \leq m/2$;
		\item $\eta := (1 - \alpha)^{k+1}/2 > 0$;
	\end{itemize}
	In that case, for all $1 \leq i \leq k-1$, we set:
	\begin{displaymath}
		\varepsilon_i := \varepsilon \cdot \alpha \cdot (\sum_{j = 0}^{i} (1-\alpha)^j + \eta) 
		= \varepsilon \cdot (1 - (1-\alpha)^{i+1} + \eta \cdot \alpha) \leq \varepsilon \cdot (1 - (1-\alpha)^{k+1} + \alpha \cdot \eta) \leq \varepsilon
	\end{displaymath}

	Consider some probability function $p \in \Dist(\llbracket 0,k \rrbracket)$ and $j \in \llbracket 1,k \rrbracket$ with $p_j = 0$. 
	Assume that the convex combination of values (i.e. the $v_i$) is at least $v_j - \varepsilon \cdot m/2$, that is:
	\begin{displaymath}
	\sum_{i = 0}^{k} p_i \cdot v_i \geq v_j - x
	\end{displaymath}
	Then, we have that the convex combination of the 'errors' (i.e. the $\varepsilon_i$) is at most $\varepsilon_j$. Indeed, let $p_+ := \sum_{i = 0}^{j-1} p_i$ and $p_- := \sum_{i = j+1}^{k} p_i = 1 - p_+$ (the $+$ and $-$ refer to the the comparison of the corresponding values with $v_j$). Then, we have: 
	\begin{align*}
		v_j - x
		& \leq \sum_{i = 0}^{k} p_i \cdot v_i = \sum_{i = 0}^{j-1} p_i \cdot v_i + \sum_{i = j+1}^{k} p_i \cdot v_i \\ & \leq \sum_{i = 0}^{j-1} p_i \cdot v_0 + \sum_{i = j+1}^{k} p_i \cdot v_{j+1} = p_+ + p_- \cdot v_{j+1} = (1 - p_-) + p_- \cdot v_{j+1}
	\end{align*}
	Overall, we obtain:
	\begin{displaymath}
	p_- \leq \frac{1 - (v_j - x
		)}{1 - v_{j+1}} \leq \alpha
	\end{displaymath}
	In the case where $j = 0$ and $p_+ = 0$, we simply have a contradiction since $p_- \leq \alpha < 1$ (this implies $p_j = 1 - \alpha > 0$ which is assumed equal to 0). Then, we obtain (note that $\varepsilon_{k-1} - \varepsilon_{j-1} \geq 0$):
	\begin{align*}
	\sum_{i = 0}^{k} p_i \cdot \varepsilon_i & = \sum_{i = 0}^{j-1} p_i \cdot \varepsilon_i + \sum_{i = j+1}^{k} p_i \cdot \varepsilon_i \leq \sum_{i = 0}^{j-1} p_i \cdot \varepsilon_{j-1} + \sum_{i = j+1}^{k} p_i \cdot \varepsilon_{k-1} \\
	& = p_+ \cdot \varepsilon_{j-1} + p_- \cdot \varepsilon_{k-1} = \varepsilon_{j-1} + p_- \cdot (\varepsilon_{k-1} - \varepsilon_{j-1}) \\
	& \leq \varepsilon_{j-1} + \alpha \cdot (\varepsilon_{k-1} - \varepsilon_{j-1}) = \alpha \cdot \varepsilon_{k-1} + (1 - \alpha) \cdot \varepsilon_{j-1} \\
	& \leq \alpha \cdot \varepsilon + (1- \alpha) \cdot \varepsilon \cdot \alpha \cdot (\sum_{l = 0}^{j-1} (1-\alpha)^l + \eta) \\ 
	& = \varepsilon \cdot (\alpha + (1 - \alpha) \cdot \alpha \cdot \sum_{l = 0}^{j-1} (1-\alpha)^l) + \varepsilon \cdot \alpha \cdot \eta \cdot (1 - \alpha) \\
	& = \varepsilon \cdot \alpha \cdot \sum_{l = 0}^{j} (1-\alpha)^l + \varepsilon \cdot \alpha \cdot \eta - \varepsilon \cdot \alpha^2 \cdot \eta 
	= \varepsilon_j - \varepsilon \cdot \alpha^2 \cdot \eta
	\end{align*}
	Hence, we obtain:
	\begin{displaymath}
	\sum_{i = 0}^{k} p_i \cdot \varepsilon_i \leq \varepsilon_j - \varepsilon \cdot \alpha^2 \cdot \eta
	\label{eqn:convex_combination_err_ok}
	\end{displaymath}
	
	In fact, the sequence $0 = \varepsilon_0 < \varepsilon_1 < \varepsilon_2 < \ldots < \varepsilon_{k-1} \leq \varepsilon$ with $\varepsilon_{k} := 0$ we consider in the one defined above for $x := \frac{\varepsilon \cdot m/2}{1-p} \leq m/2$ (by definition of $\varepsilon$). In that case, consider a probability function $p \in \Dist(\llbracket 0,k \rrbracket)$ and $j \in \llbracket 1,k \rrbracket$ with $p_j \leq p < 1$, such that:
	\begin{displaymath}
	\sum_{i = 0}^{k} p_i \cdot v_i \geq v_j - \varepsilon \cdot m/2
	\end{displaymath}
	We have $\sum_{\substack{i = 0, i \neq j}}^{k} p_i \cdot v_i \geq (1 - p_j) \cdot v_j - \varepsilon \cdot m/2$. Hence, considering $p'_i := \frac{p_i}{1 - p_j}$ for all $0 \leq i \leq k$ with $i \neq j$, and $p'_j = 0$, we have $p' \in \Dist(\llbracket 0,k \rrbracket)$ and $\sum_{i = 0}^{k} p_i' \cdot v_i \geq \cdot v_j - \frac{\varepsilon \cdot m/2}{1 - p_j} \geq v_j - \frac{\varepsilon \cdot m/2}{1 - p} = v_j - x$. That is, $\sum_{i = 0}^{k} p_i' \cdot \varepsilon_i \leq \varepsilon_j - \varepsilon \cdot \alpha^2 \cdot \eta$. We obtain:
	\begin{displaymath}
		\sum_{\substack{i = 0 \\ i \neq j}}^{k} p_i \cdot \varepsilon_i \leq (1 - p_j) \cdot \left(\varepsilon_j - \varepsilon \cdot \alpha^2 \cdot \eta \right) 
	\end{displaymath}
\end{proof}


\subsection{Infinite memory to play $\varepsilon$-optimal strategies in Büchi games}
\label{appen:buchi_varepsilon_infinite}
Consider the game of Figure~\ref{fig:local_optimal_not_uniformly}. This game is explained in \cite{AH00}. Assume that the sets of actions are $A = \{ a_1,a_2 \}$ and $B = \{ b_1,b_2 \}$ with $a_1$ corresponding to the top row and $b_1$ to the left column. 

Consider some positive probability $p > 0$ and consider a Player $\A$ strategy $\s_\A$ such that the probability to play $a_2$ is either 0 or at least $p$. Let us show that the value of such a strategy is 0. Specifically, consider a Player $\B$ strategy $\s_\B$ such that, for all $\rho \cdot q_0 \in Q^+$, we have:
\begin{equation*}
	\s_\B(\rho \cdot q_0) := \begin{cases}
		b_1 \quad &\text{if} \, \s_\A(\rho \cdot q_0)(a_2) = 0\\
		b_2 \quad &\text{if} \, otherwise \\
	\end{cases}
\end{equation*}
Each time $\s_\A(\rho \cdot q_0)(a_2) \geq p$, then the probability to reach the state $\bot$ is reached with probability at least $p$. Hence, if this happens infinitely often, then the state $\bot$ is seen with probability 1. Otherwise, the state $q_0$ is never left after some moment on. In both cases, the set $T$ is seen only finitely often. It follows that the value of the game with strategies $\s_\A,\s_\B$ is 0. 

Now, consider some $\varepsilon > 0$ let us exhibit a Player $\A$ strategy of value at least $1 - \varepsilon$. The idea is the following. Consider a sequence $(\varepsilon_k)_{k \in \N}$ of positive values such that $\lim_{n \rightarrow \infty}\Pi_{i = 0}^n (1 - \varepsilon_i) \leq 1 - \varepsilon$. Then, consider a Player $\A$ strategy such that $\s_\A(\rho)(a_1) := 1 - \varepsilon_k$ where $k \in \N$ denotes the number of times the state $\top$ is visited in $\rho$. Then, the state $q_0$ is seen indefinitely with probability 0. If the state $\top$ has been seen already $k$ times, then the probability to stay at state $q_0$ for $n$ steps is at most $(1 - \varepsilon_k)^n \rightarrow_{n \rightarrow \infty} 0$. Furthermore, the probability to ever reach the state $\bot$ is at most $\lim_{n \rightarrow \infty}\Pi_{i = 0}^n (1 - \varepsilon_i) \geq 1 - \varepsilon$. Overall, regardless of Player $\B$'s strategy, the probability to visit the set $T$ infinitely often is at least $1 - \varepsilon$. Note that such a sequence $(\varepsilon_k)_{k \in \N}$ could be equal, for instance, to $\varepsilon_k := 1 - (1 - \varepsilon)^{\frac{1}{2^{k+1}}}$. Then, for $n \in \N$:

\begin{align*}
	\Pi_{i = 0}^n (1 - \varepsilon_i) = (1 - \varepsilon)^{\sum_{i = 0}^n \frac{1}{2^{i+1}}} = (1 - \varepsilon)^{(1 - \frac{1}{2^{n+1}})} \rightarrow_{n \rightarrow \infty} 1 - \varepsilon
\end{align*}

\subsection{Proof of Proposition~\ref{prop:rm_ebm}}
\label{appen:ebm_rm_gf}

Let us first consider a proposition stating how we can translate a partial valuation and probability function (on which we consider a greatest fixed point nested with a least fixed point) into another partial valuation (on which we consider only a least fixed point) and obtain the same value 
(informally, we are transforming a Büchi game into a reachability game, and obtaining the same result as in Proposition~\ref{prop:buchi_equal_reach}).
\begin{proposition}
	Consider a game form $\formNF$, a partition of the outcomes $\outComeNF = \outComeNFEx \uplus \outComeNFLp$, a partial valuation $\alpha: \outComeNFEx \rightarrow [0,1]$ of the outcomes and a probability function $p_T: \outComeNFLp \rightarrow [0,1]$. Let $u_\Bu := u_{\formNF,\gf T \lf \bar{T}}(\alpha[p_T])$, $\outComeNF_{\bar{T}} := \{ o \in \outComeNFLp \mid p_T(o) = 0 \}$ and $\outComeNF_T := \outComeNFLp \setminus \outComeNF_{\bar{T}}$. Then, let $\alpha^{\mathsf{reach}}: \outComeNFEx \cup \outComeNF_T \rightarrow [0,1]$ be such that $\restriction{\alpha^{\mathsf{reach}}}{\outComeNFEx} := \alpha$ and for all $o \in \outComeNF_T$, we have $\alpha^{\mathsf{reach}}(o) := u_\Bu$. In that case, we have $u_{\formNF,\diamondsuit}(\alpha^{\mathsf{reach}}) = u_\Bu = u_{\formNF,\gf T \lf \bar{T}}(\alpha[p_T])$ and $(\alpha[p_T])_{\formNF,\gf T \lf \bar{T}} = \alpha^{\mathsf{reach}}_{\formNF,\lf}$.
	\label{prop:partial_buchi_into_partial_reach}
\end{proposition}
\begin{proof}
Let $u_\mathsf{reach} := u_{\formNF,\diamondsuit}(\alpha^{\mathsf{reach}})$. Consider some $u \leq u_\Bu$. For all $o \in \outComeNFEx$, we have $\alpha[p_T(u_\Bu \sqcup u)](o) = \alpha(o) = \alpha^\mathsf{reach}[u](o)$. Consider now some $o \in \outComeNFLp$. 
\begin{itemize}
	\item If $o \in \outComeNF_T$, we have: 
	\begin{displaymath}
	\alpha[p_T(u_\Bu \sqcup u)](o) = p_T(o) \cdot u_\Bu + (1 - p_T(o)) \cdot u \leq u_\Bu = \alpha^\mathsf{reach}[u](o)
	\end{displaymath}
	with equality if $u = u_\Bu$.
	\item If $o \in \outComeNF_{\bar{T}}$, then $p_T(o) = 0$. Hence:
	\begin{displaymath}
	\alpha[p_T(u_\Bu \sqcup u)](o) = p_T(o) \cdot u_\Bu + (1 - p_T(o)) \cdot u  = u = \alpha^\mathsf{reach}[u](o)
	\end{displaymath}
\end{itemize}
It follows that, for all $u \leq u_\Bu$, we have $\alpha[p_T(u_\Bu \sqcup u)] \preceq \alpha^\mathsf{reach}[u]$ with equality if $u = u_\Bu$. 

Let us now show that $u_\mathsf{reach} \leq u_\Bu$. We have $f_\formNF\langle \alpha^\mathsf{reach}\rangle(u_\Bu) = \va_{\gameNF{\formNF}{\alpha^\mathsf{reach}[u_\Bu]}} = \va_{\gameNF{\formNF}{\alpha[p_T(u_\Bu \sqcup u_\Bu)]}} =  f_{\formNF,\bar{T}}\langle \alpha[p_T(u_\Bu)]\rangle(u_\Bu) = u_\Bu$ since $u_\Bu$ is a fixed point (in fact, the least fixed point) of the function $f_{\formNF,\bar{T}}\langle \alpha[p_T(u_\Bu)]\rangle$. Hence, $u_\Bu$ is a fixed point of the function $f_\formNF\langle \alpha^\mathsf{reach}\rangle$, it is therefore greater than or equal to its least fixed point $u_\mathsf{reach}$. That is, $u_\mathsf{reach} \leq u_\Bu$.

Let us now show that $u_\Bu \leq u_\mathsf{reach}$. Let $u < u_\Bu$. Since $u_\Bu$ is the least fixed point of the function $f_{\formNF,\bar{T}}\langle \alpha[p_T(u_\Bu)]\rangle$, for all $u' < u_\Bu$, we have $u' < f_{\formNF,\bar{T}}\langle \alpha[p_T(u_\Bu)]\rangle(u')$. Hence:
\begin{align*}
	u < f_{\formNF,\bar{T}}\langle \alpha[p_T(u_\Bu)]\rangle(u)  = \va_{\gameNF{\formNF}{\alpha[p_T(u_\Bu \sqcup u)]}} \leq \va_{\gameNF{\formNF}{\alpha^\mathsf{reach}[u]}} = f_\formNF\langle \alpha^\mathsf{reach}\rangle(u)
\end{align*}
In fact, for all $u < u_\Bu$, we have $u < f_\formNF\langle \alpha^\mathsf{reach}\rangle(u)$. Hence, $u_\Bu$ is lower than or equal to the least fixed point of the function $f_\formNF\langle \alpha^\mathsf{reach}\rangle$, which is $u_\mathsf{reach}$. That is, $u_\Bu \leq u_\mathsf{reach}$.
	
	Then, it follows that $(\alpha[p_T])_{\formNF,\gf T \lf \bar{T}} = \alpha^{\mathsf{reach}}_{\formNF,\lf}$. 
\end{proof}

\begin{figure}
	\begin{minipage}[b]{0.35\linewidth}
		\hspace*{-1cm}
		\centering
		\includegraphics[scale=1]{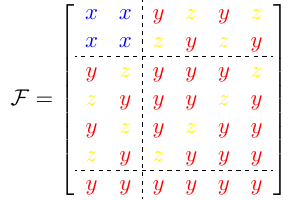}
		\caption{A game form $\formNF$ that is aBM but not RM with $\outComeNF := \{ x,y,z \}$.}
		\label{fig:gf_eBM_not_BM}
	\end{minipage}
	\hspace{5pt}
	\begin{minipage}[b]{0.35\linewidth}
		\centering
		\includegraphics[scale=1]{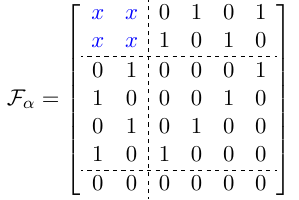}
		\caption{The game form $\formNF$ with a partial valuation witnessing that it is not RM.}
		\label{fig:gf_eBM_not_BM_valued}         
	\end{minipage}
	\begin{minipage}[b]{0.25\linewidth}
		\centering
		\includegraphics[scale=1]{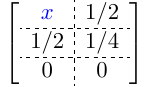}
		\caption{A simplification of the game form on the left.}
		\label{fig:gf_eBM_not_BM_valued_simplified}         
	\end{minipage}
\end{figure}

We can proceed to the proof of Proposition~\ref{prop:rm_ebm}.
\begin{proof}
	Let us first tackle the decidability of whether a game form is aBM. In \cite{BBSCSLarXiv}, it was shown that it is decidable if a game form is RM. It was done by encoding the RM property in the first order theory of the reals. The same method can be used for aBM game forms.
	
	Consider a game form $\formNF$ and assume that it is RM. Consider a partition $\outComeNF = \outComeNFEx \uplus \outComeNFLp$ and a partial valuation $\alpha: \outComeNFEx \rightarrow [0,1]$ of the outcomes, a probability function $p_T: \outComeNFLp \rightarrow [0,1]$ and assume that $u_{\formNF,\gf T \lf \bar{T}}(\alpha[p_T]) > 0$. Consider the partial valuation $\alpha^\mathsf{reach}$ from the previous Proposition~\ref{prop:partial_buchi_into_partial_reach}. Since the game form $\formNF$ is RM, there exists a reach maximizing strategy $\GF$-strategy $\sigma_\A$ in the game form $\formNF$ w.r.t. the partial valuation $\alpha^{\mathsf{reach}}$. One can then check that this $\GF$-strategy ensures the conditions of Definition~\ref{def:epsilon_buchi_maximizable} for all $\varepsilon > 0$ (this is direct from the definition of $\alpha^\mathsf{reach}$).
	
	 Let us now exhibit a game form that is aBM but not RM. Consider the game form depicted in Figure~\ref{fig:gf_eBM_not_BM}. Note that the colors have no semantics, they are only there to better grasp where every variable is. Let us first show that it is not RM. Consider the partial valuation $\alpha: \{y,z\} \rightarrow [0,1]$ such that $\alpha(z) := 1$ and $\alpha(t) := 0$. The resulting game form is depicted in Figure~\ref{fig:gf_eBM_not_BM_valued}. In fact, this can be 'simplified' as in Figure~\ref{fig:gf_eBM_not_BM_valued_simplified}. Indeed, not considering the bottom line of 0s, the bottom right $4 \times 4$ square can be seen as an outcome of value $1/4$: as both players can play uniformly on all rows and columns thus obtaining the value $1/4$. Similarly, the top right part can be seen as an outcome of value $1/2$, and so can the bottom left part. Then, one can see that the least fixed point of the function $f_\formNF\langle \alpha\rangle$ is $1/2$. Indeed, $1/2$ is a fixed point and for all $u < 1/2$, if Player $\A$ plays a $\GF$-strategy $\sigma_\A$ such that the first row is played with probability $2 \cdot u$ and the second row with probability $1 - 2 \cdot u$, then the outcome of the game is greater than $u$. That is, $f_\formNF\langle \alpha\rangle(u) > u$. Hence $1/2$ is the least fixed point of $f_\formNF\langle \alpha\rangle$. Then, one can see that there is no reach-maximizing $\GF$-strategy. To play optimally w.r.t. the valuation $\alpha[1/2]$, one has to play the first row with probability one, but then, in the first column, the only outcome in the support is $x$.
	 
	 Let us now show that this game form is aBM. Consider a partition of the outcomes $\outComeNF = \outComeNFEx \uplus \outComeNFLp$, a partial valuation $\alpha: \outComeNFEx \rightarrow [0,1]$ and a probability function $p_T: \outComeNFLp \rightarrow [0,1]$. Let $u := u_{\formNF,\gf T \lf \bar{T}}(\alpha[p_T])$ and $\outComeNF_T := \{ o \in \outComeNFLp \mid p_T(o) > 0 \}$. Assume that $u > 0$. We want to show that there exists a $\GF$-strategy such as in Definition~\ref{def:epsilon_buchi_maximizable} for all $\varepsilon > 0$. Let us call such a strategy an ok-strategy. First, one can see that if $y \in \outComeNF_T$ or if $y \in \outComeNFEx$ with $\alpha(y) \geq u$, then playing the bottom row with probability 1 is an ok-strategy for all $\varepsilon > 0$. Now, there are two possibilities left:
	 \begin{itemize}
	 	\item assume that $y \in \outComeNFEx$ and $\alpha(y) < u$. Then, consider the valuation $\alpha[p_T(u \sqcup u)]$. The value of the game with that valuation must be $u$ (as it is a fixed point). Consider the outcome of the game if Player $\B$ plays uniformly $1/4$ on the four rightmost columns: it must be at least $u$. This implies $z \in \outComeNFEx$ and $\frac{\alpha(y) +  \alpha(z)}{2} \geq u$. Now, if $x \notin \outComeNFLp \setminus \outComeNF_T$, any optimal strategy is an ok-strategy for all $\varepsilon > 0$. Hence, assume that $x \in \outComeNFLp \setminus \outComeNF_T$ and consider some $\varepsilon > 0$. The Player $\A$ strategy playing the first two rows with probability $(1 - \varepsilon)/2$ and the four middle rows with probability $\varepsilon/4$ is an ok-strategy for $\varepsilon$. 
	 	\item assume that $y \in \outComeNFLp \setminus \outComeNF_T$. Then, $z$ cannot be also in $\outComeNFLp \setminus \outComeNF_T$ since $u > 0$. Furthermore, if $z \in \outComeNFEx$ with $\alpha(z) < u$, we have a contradiction since the value of the game with valuation $\alpha[p_T(u \sqcup \alpha(z))]$ would be at most $\alpha(z) < u$ (consider a Player $\B$ $\GF$-strategy playing uniformly on the four rightmost columns). In fact, either $z \in \outComeNF_T$ or $z \in \outComeNFEx$ and $\alpha(z) \geq u$. In both cases, a Player $\A$ strategy playing the four middle rows with probability $1/4$ is an ok-strategy for all $\varepsilon > 0$.
	 \end{itemize}
 	Overall, we can conclude that the game form $\formNF$ is aBM but not RM.
\end{proof}

\section{Proofs from Section~\ref{sec:co-Buchi}}
\subsection{Playing optimally in the game of Figure~\ref{fig:co_buchi_infinite}}
\label{appen:infinite_memory_cobuchi}
We are given a probabilistic process such that, at each step either event $T$ or $\lnot T$ occur. Furthermore, at step $n \in \N$, the probability that the event $T$ occurs is equal to $\varepsilon_n$. We will denote by $\mathbb{P}$ the probability measure of the corresponding measurable sets. In the following, we will use the following notations:
\begin{itemize}
	\item for all $n \in \N$, $\mathsf{X}_n T$ refers to the event: the event $T$ occurs at step $n$.
	\item for all $n \in \N$: $\diamondsuit_n T := \cup_{k \geq n} (\mathsf{X}_k T)$ refers to the event: the event $T$ occurs at some point after step $n$.
	\item $\square \diamondsuit T := \cap_{n \in \N} (\diamondsuit_n T)$ refers to the event: the event $T$ occurs infinitely often.
\end{itemize}
An infinite path (which can be seen as an infinite sequence of elementary events) is winning for the Büchi objective if it corresponds to the event $\square \diamondsuit T$. We want to define a sequence $(\varepsilon_{k})_{k \in \N}$ such that, for all $k \in \N$ we have $\varepsilon_{k} > 0$ and $\mathbb{P}(\square \diamondsuit T) = 0$. In fact, it suffices to consider a sequence $(\varepsilon_{k})_{k \in \N}$ whose sum converges (i.e. such that $\sum_{n = 0}^{\infty} \varepsilon_{n} < \infty$), for instance $\varepsilon_{n} := \frac{1}{2^{n+1}}$ for all $n \in \N$. Indeed, for all $n \in \N$, we have:
\begin{displaymath}
\mathbb{P}(X_n T) = \varepsilon_n
\end{displaymath}
Furthermore:
\begin{displaymath}
\mathbb{P}(\diamondsuit_n T) = \mathbb{P}(\cup_{k \geq n} \mathsf{X}_k T) \leq \sum_{k = n}^\infty \mathbb{P}(\mathsf{X}_k T) = \sum_{k = n}^\infty \varepsilon_k
\end{displaymath}
It follows that:
\begin{displaymath}
\mathbb{P}(\square \diamondsuit T) = \mathbb{P}(\cap_{n \in \N} \diamondsuit_n T) = \lim\limits_{n \rightarrow \infty} \mathbb{P}(\diamondsuit_n T) \leq \lim\limits_{n \rightarrow \infty} \sum_{k = n}^\infty \varepsilon_k = 0
\end{displaymath}

\subsection{Definition of coBM game forms}
\label{proof:thm_coBM_necessary}
First, let us formally define co-BM game forms as the exact formalism we use is a little different from what was presented in the main part of paper. 
\begin{definition}[co-Büchi maximizable game forms complete definition]
	Consider a game form $\formNF$, a valuation of the outcomes $\alpha: \outComeNF \to [0,1]$, and two probability functions $p_\alpha,p_T: \outComeNF \rightarrow [0,1]$. Let $u := u_{\formNF,\lf T \gf \bar{T}}(\alpha[p_\alpha,p_T])$. Then, the game form $\formNF$ is \emph{co-Büchi maximizable}	w.r.t. $\alpha,p_\alpha,p_T$ if either $u = 0$ or there exists a $\GF$-strategy $\sigma_\A \in \Opt_\A(\gameNF{\formNF}{\alpha[p_\alpha,p_T(u,u)]})$ such that, for all
	$b \in \St_\B$ we have: if there is some $a \in \Supp(\sigma_\A)$ such that 
	$p_T(\outCNF(a,b)) > 0$ (i.e. if there is a red outcome), then there is some $a' \in \Supp(\sigma_\A)$ such that $p_\alpha(\outCNF(a',b)) > 0$ (i.e. then there is a green outcome). Such a  $\GF$-strategy is called co-Büchi maximizing.

	Then, the game form $\formNF$ is \emph{co-Büchi maximizable} if, for every valuation $\alpha: \outComeNF \rightarrow [0,1]$ and probability
	functions $p_\alpha,p_T: \outComeNF \rightarrow [0,1]$, it is co-Büchi maximizable	w.r.t. $\alpha,p_\alpha,p_T$.
	\label{def:co_buchi_maximizable_new}
\end{definition}

We can also define a co-Büchi game built from a game form, a valuation of the outcomes and two probability functions.
\begin{definition}
	Consider a game form $\formNF = \langle A,B,\outComeNF,\outCNF \rangle$, a valuation $\alpha: \outComeNF \rightarrow [0,1]$ and two probability functions $p_\alpha,p_T: \outComeNF \rightarrow [0,1]$. We define the game $\G^\coBu_{\formNF,\alpha[p_\alpha,p_T]} = \Games{\Aconc}{\coBu(T)}$ where $\Aconc := \AConc$ with:
	\begin{itemize}
		\item $Q := \{ q_0,q_T,q_{\bar{T}},\top,\bot \}$;
		\item $\distribSet := \{ d_{o} \mid o \in \outComeNF \} \cup \{ d_\top,d_\top \} \cup \{ d_{q_0} \}$;
		\item For all $a \in A$, $b \in B$, we have: $\delta(q_0,a,b) := d_{\outCNF(a,b)}$.
		Furthermore, $\delta(q_T,a,b) = \delta(q_{\bar{T}},a,b) = d_{q_0}$, $\delta(\top,a,b) = d_\top$ and $\delta(\bot,a,b) = d_\bot$;
		\item for all $o \in \outComeNF$, we have $\distribFunc(d_o) := \{ \top \mapsto \alpha(o) \cdot p_\alpha(o); \bot \mapsto (1 - \alpha(o)) \cdot p_\alpha(o); q_T \mapsto p_T(o) \cdot (1 - p_\alpha(o)); q_{\bar{T}} \mapsto (1 - p_T(o)) \cdot (1 - p_\alpha(o)) \}$. Furthermore, for all $q \in \{ \top,\bot,q_0 \}$, we have $\distribFunc(d_{q})(q) := 1$;
		\item $T := \{ q_T,\bot \}$.
	\end{itemize}
\end{definition}

We have a Proposition very similar to Proposition~\ref{prop:same_value_game_game_form} (the proof is analogous).
\begin{proposition}
	\label{prop:same_value_game_game_form_co_buchi}
	Consider a game form $\formNF = \langle A,B,\outComeNF,\outCNF \rangle$, a valuation $\alpha: \outComeNF \rightarrow [0,1]$ and two probability functions $p_\alpha,p_T: \outComeNF \rightarrow [0,1]$. Consider the co-Büchi game $\G := \G^\coBu_{\formNF,\alpha[p_\alpha,p_T]}$. Let $v := \MarVal{\G}$ be the valuation of the states in $\G$ giving the value of states. Then: $v(q_0) = u_{\formNF,\lf T \gf \bar{T}}(\alpha[p_T])$. 
\end{proposition}

We  obtain a lemma similar to Lemma~\ref{lem:varepsilon_buchi_necessary}.
\begin{lemma}
	Consider a game form $\formNF$, a valuation $\alpha: \outComeNF \rightarrow [0,1]$ and two probability functions $p_\alpha,p_T: \outComeNF \rightarrow [0,1]$. The game form $\formNF$ is co-Büchi maximizable w.r.t. $\alpha,p_\alpha,p_T$ if and only if there is a positional optimal strategy from the
	initial state $q_0$ in the game $\G^\coBu_{\formNF,\alpha[p_\alpha,p_T]}$.
	\label{prop:co_buchi_necessary_other}
\end{lemma}
\begin{proof}
	Consider a game form $\formNF$, a valuation $\alpha: \outComeNF \rightarrow [0,1]$, two probability functions $p_\alpha,p_T: \outComeNF \rightarrow [0,1]$ and the game $\G := \G^\coBu_{\formNF,\alpha[p_\alpha,p_T]}$. Let $v := \MarVal{\G}$ be the valuation in $\G$ giving the value of the states. Let $u := u_{\formNF,\lf T \gf \bar{T}}(\alpha[p_\alpha,p_T])$ and assume that $u > 0$. By applying Proposition~\ref{prop:same_value_game_game_form_co_buchi}, 
	we can show that $\alpha[p_\alpha,p_T(u,u)] = \mu_v \circ g$ for $g: \outComeNF \rightarrow \distribSet$ such that for all $o \in \outComeNF$, we have $g(o) := d_o$. 
	
	Let us now show that the game form $\formNF$ is co-Büchi maximizable w.r.t. $\alpha,p_\alpha,p_T$ if and only if there is a positional optimal strategy from the
	initial state $q_0$ in the game $\G^\coBu_{\formNF,\alpha[p_\alpha,p_T]}$. Note that such a strategy would be uniformly optimal. Furthermore, it is completely characterized by the $\GF$-strategy $\s_\A(q_0)$ that is played at the local interaction $\formNF_{q_0}$ (since what is played at states $q_T,q_{\bar{T}},\top,\bot$ does not affect the game). Let us now apply Lemma~\ref{lem:uniformly_optimal}. Such a Player $\A$ strategy $\s_\A$ is uniformly optimal if and only if:
	\begin{itemize}
		\item it is locally optimal. Since $\alpha[p_\alpha,p_T(u,u)] = \mu_v \circ g$, this is equivalent to the $\GF$-strategy $\s_\A(q_0)$ being optimal in the game form $\formNF$ w.r.t. the valuation $\alpha[p_\alpha,p_T(u,u)]$: $\s_\A(q_0) \in \Opt_\A(\gameNF{\formNF}{\alpha[p_\alpha,p_T(u,u)]})$.
		\item for all EC $H = (Q_H,\beta)$ in the MDP induced by the strategy $\s_\A$, if $v_H > 0$, all states $q \in Q_H$ have value 1 in the game $\Aconc_H^{\s_\A}$. In such an EC $H$, if $q_0 \notin Q_H$, then this straightforwardly holds. Assume now that we have $q_0 \in H$ and let $b \in \beta(q_0)$. Since the action $b$ is compatible with the EC $H$ and that the states $\top,\bot$ cannot be in $H$ since $q_0$ is not reachable from these states, is follows that for all $a \in \Supp(\s_\A(q_0))$, the probability to see either $\top$ of $\bot$ is 0. That is, $p_\alpha(\outCNF(a,b)) = 0$. Then the value of the state $q_0$ is not 0 if in the game $\Aconc_H^{\s_\A}$, for all actions $b \in \beta(q_0)$, there is a positive probability to see the state $q_T \in T$
		. That is, for all $a \in \Supp(\s_\A(q_0))$, we have $p_T(\outCNF(a,b)) > 0$. 
	\end{itemize} 
	Overall, we have that there is a uniformly optimal strategy in the game $\G$ if and only if there is a co-Büchi maximizing $\GF$ strategy in $\formNF$ w.r.t. $\alpha,p_\alpha,p_T$.
\end{proof}

\subsection{Proof of Theorem~\ref{lem:co_buchi_sufficient}}
\label{proof:thm_coBM_sufficient}
Consider Definition~\ref{def:co_buchi_maximizable_new} of coBM game forms. 
Before proving Theorem~\ref{lem:co_buchi_sufficient}, let us first define some valuations and probability functions in co-Büchi games.

\begin{definition}
	Consider a co-Büchi game $\G = \Games{\Aconc}{\coBu(T)}$ whose value is given by the valuation $v := \MarVal{\G}: Q \rightarrow [0,1]$. Let $u \in [0,1]$ be some value and consider an area of the game $\emptyset \neq Q_u \subseteq Q$ such that, for all $q \in Q_u$, we have $v(q) = u$. Let $Q_u^\gf := Q_u \setminus T$ and $Q_u^\lf := Q_u \cap T$. 
	
	We consider the valuation $\alpha^{Q_u}: \distribSet \rightarrow [0,1]$ of the outcomes such that, for all $d \in \distribSet$, we have: 
	\begin{equation*}
	\alpha^{Q_u}(d) := \begin{cases}
	\frac{1}{\distribFunc(d)[Q \setminus Q_u]} \sum_{q \in Q \setminus Q_u} \distribFunc(d)(q) \cdot v(q) \quad &\text{if } \distribFunc(d)[Q \setminus Q_u] \neq 0 \\
	0 \quad &\text{if } \distribFunc(d)[Q \setminus Q_u] = 0 \\
	\end{cases}
	\end{equation*}

	Furthermore, we let $p_{\alpha^{Q_u}}: \distribSet \rightarrow [0,1]$ be such that $p_{\alpha^{Q_u}}(d) := \distribFunc(d)[Q \setminus Q_u]$ for all $d \in \distribSet$. In addition, for all $S \subseteq Q_u$, we set 	 $p_{T}[S]: \distribSet \rightarrow [0,1]$ as
	\begin{equation*}
	p_{T}[S](d) := \begin{cases}
	\frac{1}{1 - p_{\alpha^{Q_u}}(d)} \distribFunc(d)[S] \quad &\text{if } p_{\alpha^{Q_u}}(d) \neq 1 \\
	0 \quad &\text{if } p_{\alpha^{Q_u}}(d) = 1 \\
	\end{cases}
	\end{equation*}
	\label{def:for_co_buchi_proof}
\end{definition}

Let us now state and prove a crucial lemma. This can be seen as a generalization of Lemma~\ref{lem:value_in_RM_up_to_u} in the context of co-Büchi games.
\begin{lemma}
	Consider a co-Büchi game $\G = \Games{\Aconc}{\coBu(T)}$ whose value is given by the valuation $v := \MarVal{\G}: Q \rightarrow [0,1]$. Let $u \in [0,1]$ be some value and consider an area of the game $\emptyset \neq Q_u \subseteq Q$ such that, for all $q \in Q_u$, we have $v(q) = u$. In addition, let $\distribSet_{\mathsf{Lp}} := \{ d \in \distribSet \mid p_{\alpha^{Q_u}}(q) = 0 \}$, $\distribSet_{\mathsf{Ex}} := \distribSet \setminus \distribSet_{\mathsf{Lp}}$ and $\alpha: \distribSet_{\mathsf{Ex}} \rightarrow [0,1]$ be a partial valuation of the Nature states such that $\alpha(d) := \mu_{v}(d)$ for all $d \in \distribSet_{\mathsf{Ex}}$. Let $\alpha_u := \alpha^{Q_u}$, $p_u := p_{\alpha^{Q_u}}$ and $p_T^S := p_T[Q_u \setminus S]$ for all $S \subseteq Q_u$.
	
	Then, assume that, for all $q \in Q_u^\lf$ we have $u_{\formNF_q,\lf}[\alpha] < u$. Let us define the sequence $(W_k[Q_u])_{k \in \N} \in (\mathcal{P}(Q_u))^\N$ by $W_0[Q_u] :=  Q_u^{\gf}$ and $W_{k+1}[Q_u] := \{ q \in W_k[Q_u] \mid u_{\formNF_{q},\lf T \gf \bar{T}}(\alpha_u[p_u,p_T^{W_k[Q_u]}]) \geq u \}$. Then, for all $k \geq 0$, we have $W_k[Q_u] \neq \emptyset$.
	\label{lem:crucial_co_buchi}
\end{lemma}
\begin{proof}
	In the following, the sets $W_k[Q_u]$ will be denoted $W_k$ for all $k \in \N$.
	
	First, let us show that $W_0 \neq \emptyset$. Towards a contradiction, if $W_0 = \emptyset$, then $Q_u = Q_u \cap T$. Then, consider the co-Büchi game $\G_{Q \setminus Q_u}$ from Definition~\ref{def:game_extracted}. Since all states in $Q_u$ are in $T$, the objective $\coBu(T \cup \{ \bot \})$ is in fact equivalent (i.e. it has the same measure w.r.t. all strategies $\s_\A,\s_\B$ for both players) to the objective $\Reach(\{ \top \})$. Then, we obtain a contradiction with Lemma~\ref{lem:value_in_RM_up_to_u}.
	
	
	Let us now assume towards a contradiction that $W_{i+1} = \emptyset$ whereas $W_i \neq \emptyset$ for some $i \geq 0$. We set:
	\begin{displaymath}
		x := \max (\max_{q \in Q_u^\lf} u_{\formNF_q,\lf}[\alpha],\max_{k \leq i} \max_{q \in W_k \setminus W_{k+1}} u_{\formNF_{q},\lf T \gf \bar{T}}(\alpha_u[p_u,p_T^{W_k}])) 
	\end{displaymath}
	Note that, by definition, $x < u$. Furthermore, let $q \in W_k \setminus W_{k+1}$ for some $k \in \N$ and $x \leq y$. We have $u_{\formNF_{q},\lf T \gf \bar{T}}(\alpha_u[p_u,p_T^{W_k}]) \leq y$ by definition of $x$. Therefore, $f_{\formNF,\gf \bar{T}}\langle \alpha_u[p_u,p_T^{W_k}]\rangle(y) \leq y$. That is: $f_{\formNF,\gf \bar{T}}\langle \alpha_u[p_u,p_T^{W_k}]\rangle(y) = u_{\formNF,\gf \bar{T}}\langle \alpha_u[p_u,p_T^{W_k}(y)]\rangle \leq y$. 
	Recall that $u_{\formNF,\gf \bar{T}}(\alpha_u[p_u,p_T^{W_k}(y)])$ is the greatest fixed point of the function $f_{\formNF,\bar{T}}\langle \alpha_u[p_u,p_T^{W_k}(y)]\rangle$. It follows that, for all $z > y$, we have:
	\begin{equation}
		f_{\formNF,\bar{T}}\langle \alpha_u[p_u,p_T^{W_k}(y)]\rangle(z) < z
		\label{eqn:ineq_f_lfp}
	\end{equation}
	
	In the following, to ease the proof, we assume that the values of states in $Q \setminus Q_u$ are fixed and given by the valuation $v$ (i.e. the total valuations of the states considered are in $[0,1]^{Q_u}$).
	
	Now, for all $y \in [0,1]$, we set $w_T(y): Q_u^\lf \rightarrow [0,1]$ such that, for all $q \in Q_u^\lf$, we have $w_T(y)(q) := y$.	
	
	Let us now consider the function $\OneStepSet{Q_u^\lf}{w_T(x)}$ denoted $\OneStep_x: [0,1]^{Q_u^\gf} \rightarrow [0,1]^{Q_u^\gf}$. Let $w_{\bar{T}} \in [0,1]^{Q_u^\gf}$. Assume that $M := \max_{q \in Q_u^\gf} w_{\bar{T}}(q) > x$. Then, let $k$ be the smallest index such that there is some state $q \in W_k \setminus W_{k+1}$ such that $w_{\bar{T}}(q) = M$. Let $m := \max_{q \in Q_u^{\gf} \setminus W_k} w_{\bar{T}}(q) < M$. Let us show that we have the inequality $\mu_w \preceq \alpha_u[p_u,p_T^{W_k}(\max(x,m) \sqcup M)]$ for $w := w_T(x) \sqcup w_{\bar{T}}: Q \rightarrow [0,1]$. Indeed, let $d \in \distribSet$. We have:
	\begin{align*}
		\mu_w(d) & = \sum_{q' \in Q} \distribFunc(d)(q') \cdot w(q') \\
		& = \sum_{q' \in Q \setminus Q_u}  \distribFunc(d)(q') \cdot w(q') + \sum_{q' \in Q_u \setminus W_{k}} \distribFunc(d)(q') \cdot w(q') + \sum_{ q' \in W_k} \distribFunc(d)(q') \cdot w(q') \\
		& \leq \sum_{q \in Q \setminus Q_u}  \distribFunc(d)(q) \cdot v(q') + \sum_{q' \in Q_u \setminus W_{k}} \distribFunc(d)(q) \cdot \max(x,m) + \sum_{ q' \in W_k} \distribFunc(d)(q') \cdot M \\
		& = p_u(d) \cdot \alpha_u(d) + \distribFunc(d)[Q_u \setminus W_k] \cdot \max(x,m) + \distribFunc(d)[W_k] \cdot M \\
		& = p_u(d) \cdot \alpha_u(d) + (1 - p_u(d)) \cdot \left(p_T^{W_T}(d) \cdot \max(x,m) + (1 - p_T^{W_T}(d)) \cdot M \right)\\
		& = \alpha_u[p_u,p_T^{W_k}(\max(x,m) \sqcup M)]
	\end{align*}
	
	Hence, we do have $\mu_w \preceq \alpha_u[p_u,p_T^{W_k}(\max(x,m) \sqcup M)]$. Let us now show that $\OneStep_x(w_{\bar{T}})(q) < M = w_{\bar{T}}(q)$. We have:
	\begin{displaymath}
		\OneStep_x(w_{\bar{T}})(q) = \OneStep(w_T(x) \sqcup w_{\bar{T}})(q) = \va_{\gameNF{\formNF_q}{\mu_w}}
	\end{displaymath}
	Hence, by Equation~(\ref{eqn:ineq_f_lfp}) and since $x \leq \max(x,m) < w_{\bar{T}}(q)$, we have:
	\begin{displaymath}
	\OneStep_x(w_{\bar{T}})(q) \leq \va_{\gameNF{\formNF_q}{\alpha_u[p_u,p_T^{W_k}(\max(x,m) \sqcup w_{\bar{T}}(q))]}} = f_{\formNF,\bar{T}}\langle \alpha_u[p_u,p_T^{W_k}(\max(x,m))]\rangle(w_{\bar{T}}(q)) < w_{\bar{T}}(q)
	\end{displaymath}
	It follows that the greatest fixed point $v_x := v_{\gf Q_u^\gf}[w_T(x)] \in [0,1]^{Q_u^\gf}$. of the function $\OneStep_x$ has its maximum less than or equal to $x$. 
	
	Furthermore, we have $\nu_{w_T(x) \sqcup v_x} \preceq \alpha[x]$. Indeed, for all $d \in \distribSet$, if $d \in \distribSet_{\mathsf{Lp}}$, then:
	\begin{align*}
	\mu_{w_T(x) \sqcup v_x}(d) & = \sum_{q' \in Q_u} \distribFunc(d)(q') \cdot (w_T(x) \sqcup v_x)(q') \leq \sum_{q' \in Q_u} \distribFunc(d)(q') \cdot x \leq x = \alpha[x]
	\end{align*}
	
	Furthermore, if $d \notin \distribSet_{\mathsf{Ex}}$, then:
	\begin{align*}
	\mu_{w_T(x) \sqcup v_x}(d) & = \sum_{q' \in Q \setminus Q_u}  \distribFunc(d)(q') \cdot v(q') + \sum_{q' \in Q_u} \distribFunc(d)(q')  \cdot (w_T(x) \sqcup v_x)(q') \\
	& \leq \sum_{q' \in Q \setminus Q_u}  \distribFunc(d)(q') \cdot v(q') + \sum_{q' \in Q_u} \distribFunc(d)(q')  \cdot x \\
	& \leq \sum_{q' \in Q \setminus Q_u}  \distribFunc(d)(q') \cdot v(q') + \sum_{q' \in Q_u} \distribFunc(d)(q')  \cdot u \\
	& = \sum_{q' \in Q \setminus Q_u}  \distribFunc(d)(q') \cdot v(q') + \sum_{q' \in Q_u} \distribFunc(d)(q')  \cdot v(q') \\
	& = \mu_{v}(d) = \alpha(d)
	\end{align*}
	That is, we do have $\nu_{w_T(x) \sqcup v_x} \preceq \alpha[x]$. It follows that, for all $q \in Q_u^\lf$, we have:
	\begin{displaymath}
		\OneStepGFP{Q_u^\lf}(w_T(x))(q) = \OneStep(w_T(x) \sqcup v_x)(q) = \va_{\gameNF{\formNF_q}{\mu_{w_T(x) \sqcup v_x}}} \leq \va_{\gameNF{\formNF_q}{\alpha[x]}} \leq x = w_T(x)
	\end{displaymath}
	That is, $\OneStepGFP{Q_u^\lf}(w_T(x)) \preceq w_T(x)$. Furthermore, note that $w_T(u)$ is the least fixed point of the function $\OneStepGFP{Q_u^\lf}$ (since the values of the co-Büchi games in all states in $Q_u$ is $u$ and is given by the least fixed point of this function). Since $x < u $ and $w_T(x) \prec w_T(u)$, there is a contradiction with Proposition~\ref{prop:no_derase_before_lfp}.
\end{proof}

We can now proceed to the proof of Theorem~\ref{lem:co_buchi_sufficient}.
\begin{proof}[Proof of Theorem~\ref{lem:co_buchi_sufficient}]
	Consider a co-Büchi game $\G = \Games{\Aconc}{\coBu(T)}$ and assume that all local interactions at states in $T$ are RM and all local interactions at states in $Q \setminus T$ are coBM. Let $v := \MarVal{\G}$ be the valuation giving the values of the states of the game. Let us exhibit a positional uniformly optimal Player $\A$ strategy $\s_\A$. 
	
	Let $u \in ]0,1]$ be a positive value and let $Q_u := \{ q \in Q \mid v(q) = u \} \neq \emptyset$ be a non-empty set. Let us define inductively a sequence of sets of states $(Q_u^i)_{i \in \N} \in (\mathcal{P}(Q_u))^\N$. Let $Q_u^0 := Q_u$. Then, for all $i \geq 0$, let us define $Q_u^{i+1} \subseteq Q_u^i$ in the following way. If $Q_u^i$ is empty, then $Q_u^{i+1} := \emptyset$. Now, assume that $Q_u^i \neq \emptyset$. Consider the notations from Definition~\ref{def:for_co_buchi_proof}. We want to apply Lemma~\ref{lem:crucial_co_buchi}. Let $S_u^i := \{ q \in Q_u^i \cap T \mid u_{\formNF_q,\lf}[\alpha] \geq u \} \subseteq Q_u^i$ with $\alpha$ the partial valuation from Lemma~\ref{lem:crucial_co_buchi}. If $S_u^i \neq \emptyset$, then $Q_u^{i+1} := Q_u^i \setminus S_u^i$. Furthermore, for all $q \in S_u^i$, let $\s_\A(q)$ be a reach maximizing strategy  w.r.t. the partial valuation $\alpha$ from Lemma~\ref{lem:crucial_co_buchi} (applied with $Q_u := Q_u^i$) in the game form $\formNF_{q}$. Assume now that $S_u^i = \emptyset$. Then, consider the limit of the sequence of states $W_u^i := W_{k}[Q_u^i] = W_{k+1}[Q_u^i] \neq \emptyset$, not empty, from  Lemma~\ref{lem:crucial_co_buchi}. Note that in particular, $W_u^i \subseteq Q \setminus T$. In that case, $Q_u^{i+1} := Q_u^i \setminus W_u^i$. Furthermore, for all $q \in W_u^i$, let $\s_\A(q)$ be a co-Büchi maximizing strategy w.r.t. the triplet $\alpha^{Q_u^i},p_{\alpha^{Q_u^i}},p_T[Q_u^i \setminus W_u^i]$ in the game form $\formNF_{q}$.
	
	We consider these definitions for all $u \in [0,1]$. Hence, the strategy $\s_\A$ is defined on all states of the game $\G$. Let us show that it is uniformly optimal. We apply Lemma~\ref{lem:uniformly_optimal}. The strategy $\s_\A$ is locally optimal, by definition. (Indeed, reach maximizing and co-Büchi maximizing $\GF$-strategies are optimal in the game forms considered w.r.t. the valuation of the outcomes, which correspond to the lift to the Nature states of the valuation $v$ -- it may in fact be bigger.) Furthermore, consider an EC $H = (Q_H,\beta)$. By Proposition~\ref{prop:ec_same_vale}, there exists some $u \in [0,1]$ such that $Q_H \subseteq Q_u$. Assume that $u > 0$. Let $i \in \N$ be the smallest index such that the set $H^i := Q_H \cap Q_u^i \setminus Q_u^{i+1}$ is not empty. In particular, $Q_H \subseteq Q_u^i$. 
	Let us show that $H^i \cap S_u^i = \emptyset$. Indeed, consider some $q \in H^i \cap S_u^i$ and an action $b \in \beta(q)$. Then, by choice of the $\GF$ strategy $\s_\A(q)$, there exists an action $a \in \Supp(\s_\A(q))$ such that $\delta(q,a,b) \notin E = \{ d \in \distribSet \mid \distribFunc(d)[Q \setminus Q_u^i] = 0 \}$. That is, $\distribFunc(\delta(q,a,b))[Q \setminus Q_u^i] > 0$. It follows that the action $b$ is not compatible with the EC $H$ (since $Q_H \subseteq Q_u^i$). Hence, $q$ cannot be in $H$. 
	In fact $H^i \subseteq W_u^i$. Let $q \in H^i$ and consider some action $b \in \beta(q)$. If there is an action $a \in \Supp(\s_\A(q))$ such that $p_T[Q_u^i \setminus W_u^i](\delta(q,a,b)) > 0$, then (since the $\GF$-strategy $\s_\A(q)$ is chosen co-Büchi maximizing), there is an action $a' \in \Supp(\s_\A(q))$ such that $p_{\alpha^{Q_u^i}}(\delta(q,a,b)) > 0$. That is, $\distribFunc(\delta(q,a,b))[Q \setminus Q_u^i] > 0$. That is, this action is not compatible with the EC $H$ since $Q_h \subseteq Q_u^i$. In fact, for all $b \in \beta(q)$, for all $a \in \Supp(\s_\A(q))$, we have $p_T[Q_u^i \setminus W_u^i](\delta(q,a,b)) = 0$. That is, (since $p_{\alpha^{Q_u^i}}(\delta(q,a,b)) = 0$): $\distribFunc(\delta(q,a,b))[Q_u^i \setminus W_u^i] = 0$. As $Q_H \subseteq Q_u^i$, then $\distribFunc(\delta(q,a,b))[Q_u^i] = 1$. That is, $\distribFunc(\delta(q,a,b))[W_u^i] = 1$. This holds for all actions $a \in \Supp(\s_\A(q))$, $b \in \beta(q)$ and states $q \in H^i \subseteq Q_u^i$. Since the EC $H$ is strongly connected, this implies that $Q_H \subseteq Q_u^i \subseteq Q \setminus T$. That is, $Q_H \cap T = \emptyset$. It follows that the value of the game $\Aconc_H^{\s_\A}$ is 1 for all states $q \in Q_H$ since it is not possible to ever see the set $T$. We conclude by applying Lemma~\ref{lem:uniformly_optimal}. 
\end{proof}

\subsection{Proof of Proposition~\ref{prop:rm_cobm}}
\label{appen:gf_rm_not_coBM}
We consider Definition~\ref{def:co_buchi_maximizable_new} of coBM game forms. 

\begin{figure}
	\begin{minipage}[b]{0.35\linewidth}
		\hspace*{-1cm}
		\centering
		\includegraphics[scale=1]{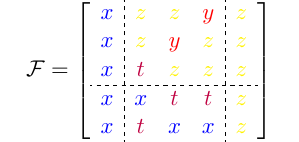}
		\caption{A game form $\formNF$ that is RM but not coBM.}
		\label{fig:gf_RM_not_coBM}
	\end{minipage}
	\begin{minipage}[b]{0.35\linewidth}
		\centering
		\includegraphics[scale=1]{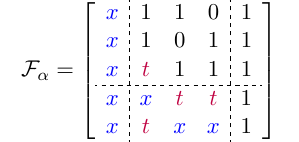}
		\caption{The game form $\formNF$ with a partial valuation witnessing that it is not coBM.}
		\label{fig:gf_RM_not_coBM_valued}         
	\end{minipage}
	\begin{minipage}[b]{0.25\linewidth}
		\centering
		\includegraphics[scale=1]{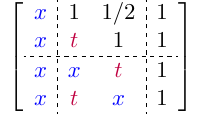}
		\caption{A simplification of the game form on the left.}
		\label{fig:gf_RM_not_coBM_valued_simplified}         
	\end{minipage}
\end{figure}

\begin{proof}[Proof of proposition~\ref{prop:rm_cobm}]
	Let us first tackle the decidability of whether or not a game form is coBM. As for the case of aBM game forms from Proposition~\ref{prop:rm_ebm}, we can encode the fact that a game form is coBM in the first order theory of the reals, just like it was done in \cite{BBSCSLarXiv} for RM game forms.
	
	Assume that a game form $\formNF$ is coBM, let us show that it is RM. Consider a partition $\outComeNF = \outComeNFEx \uplus \outComeNFLp$ and a partial valuation $\alpha_r: \outComeNFEx \rightarrow [0,1]$ of the outcomes. Then, let $\alpha: \outComeNF \rightarrow [0,1]$ be such that $\restriction{\alpha}{\outComeNFEx} := \alpha_r$. Let also $p_\alpha: \outComeNF \rightarrow [0,1]$ be such that $p_\alpha(o) := 0$ if $o \in \outComeNFLp$ and $p_\alpha(o) := 1$ if $o \in \outComeNFEx$. Furthermore, let $p_T: \outComeNF \rightarrow [0,1]$ such that $p_T := 1 - p_\alpha$. Then, for all $u_T,u_{\bar{T}} \in [0,1]$, we have $\alpha[p_\alpha,p_T(u_T \sqcup u_{\bar{T}})] = \alpha[p_\alpha,p_T(u_T \sqcup 1)] = \alpha_r[u_T]$. It follows that $u_{\formNF,\lf T \gf \bar{T}}(\alpha[p_\alpha,p_T]) = u_{\lf \formNF}(\alpha_r)$. Hence, since the game form $\formNF$ is coBM, there exists a $\GF$-strategy $\sigma_\A$ that is co-Büchi maximizing w.r.t. $\alpha,p_\alpha,p_T$. It is also reach maximizing w.r.t. $\alpha_r$. That is, the game form $\formNF$ is RM.
	
	Consider now the game form depicted in Figure~\ref{fig:gf_RM_not_coBM} (note that the colors have no semantics). Let us show that it is RM but coBM. Let us first show that it is not coBM. Consider the valuation $\alpha: \{ x,y,z,t \} \rightarrow [0,1]$ such that $\alpha(x),\alpha(y),\alpha(t) := 0$ and $\alpha(z) := 1$. Furthermore, let $p_\alpha: \{ x,y,z,t \} \rightarrow [0,1]$ be such that $p_\alpha(y) = p_\alpha(z) := 1$ and $p_\alpha(x) = p_\alpha(t) := 0$. Furthermore, let $p_T: \{ x,y,z,t \} \rightarrow [0,1]$ such that $p_T(x) := 0$ and $p_T(t) := 1$. In that case, we have $u_{\formNF,\lf T \gf \bar{T}}(\alpha[p_\alpha,p_T]) = 1$. 
	Then, the support of any optimal strategy w.r.t. to the valuation $\alpha,p_\alpha,p_T[1,1]$ is included in the three bottom rows. For any such strategy, there is a column in which there is outcome with positive probability to see $T$ (i.e. $t$) whereas there is no variable in $\{ y,z \}$. Hence, that game form is not coBM.
	
	Let us now argue that this game form is RM. 
	Consider a partition $\outComeNF = \outComeNFEx \uplus \outComeNFLp$ and partial valuation $\alpha: \outComeNFEx \rightarrow [0,1]$ of the outcomes. Let $u := u_{\formNF,\lf}[\alpha]$. If either $x \in \outComeNFLp$ or $z \in \outComeNFLp$, then $u = 0$. Let us assume that $x,z \in \outComeNFEx$. Then, we have $u \leq \alpha(x),\alpha(z)$. Then, if $y \in \outComeNFLp$, then a $\GF$-strategy playing the two top rows with probability $1/2$ is reach maximizing w.r.t. $\alpha$. Similarly, if $t \in \outComeNFLp$, then a $\GF$-strategy playing the two bottom rows with probability $1/2$ is reach maximizing w.r.t. $\alpha$. Overall, in any case there is a reach-maximizing $\GF$-strategy w.r.t. the valuation $\alpha$. That is, the game form $\formNF$ is RM.
\end{proof}

\section{Proof of Corollary~\ref{coro:two_var_unif_optim}}
\label{proof:coro_two_var_unif_optim}
First, we have that all two-variable game forms are coBM.
\begin{proposition}
	Any game form $\formNF = \langle \St_\A,\St_\B,\outComeNF,\outCNF \rangle$ such that $|\outComeNF| \leq 2$ is coBM.
	\label{prop:two_var_co_bm}
\end{proposition}
\begin{proof}
	Let $\outComeNF = \{ x,y \}$. We consider Definition~\ref{def:co_buchi_maximizable_new} of coBM game forms. Consider $\alpha,p_\alpha,p_T: \outComeNF \rightarrow [0,1]$ and let $u := u_{\formNF,\lf T \gf \bar{T}}(\alpha[p_\alpha,p_T])$. 
	
	If $p_\alpha(x),p_\alpha(y) > 0$ any optimal strategy is co-Büchi maximizing. However, if $p_\alpha(x) = p_\alpha(y) = 0$, then either there is an action $a \in \St_\A$ such that $p_T(\outCNF(a,B)) = \{ 0 \}$, and in that case $u = 1$ and $a$ is co-Büchi maximizing or for all $a \in \St_\A$, there is $b \in \St_\B$ such that $p_T(\outCNF(a,b)) > 0$, and in that case $u = 0$.
	
	Assume now that $p_\alpha(x) > 0$ and $p_\alpha(y) = 0$ (the argument is the same for the converse). If $p_T(y) = 0$, any optimal strategy is co-Büchi maximizing. Assume that $p_T(y) > 0$. Then $u \leq \alpha(x)$ (consider a Player $\B$ strategy playing uniformly over all actions). The arguments are then similar than in the previous paragraph: either there is an action $b \in \St_\A$ such that $p_\alpha(\outCNF(A,b)) = \{ 0 \}$, and in that case $u = 0$ or for all $b \in \St_\B$, there is $a \in \St_\A$ such that $p_T(\outCNF(a,b)) > 0$, playing uniformly over all actions for Player $\A$ is co-Büchi maximizing.
\end{proof}

Consider now the proof of Corollary~\ref{coro:two_var_unif_optim}.
\begin{proof}
	All local interactions are game forms with at most two outcomes, hence, by Proposition~\ref{prop:two_var_co_bm}, they are all coBM and therefore RM by Proposition~\ref{prop:rm_cobm}. The result then follows from Theorem~\ref{lem:rm_in_buchi} and Theorem~\ref{lem:co_buchi_sufficient}.
\end{proof}

\end{document}